\documentclass[12pt]{report}
\usepackage[toc,page]{appendix}
\usepackage[polish,english]{babel}
\usepackage[lf]{venturis} 

\usepackage[T1]{fontenc}
\usepackage[utf8]{inputenc}
\usepackage{graphicx}
\usepackage{hyperref}
\usepackage{lipsum}
\usepackage{amsmath}
\usepackage{amsthm}
\usepackage{amsfonts}
\usepackage{xcolor}
\usepackage[a4paper, total={6in, 8in}]{geometry}

\theoremstyle{theorem}
\newtheorem{theorem}{Theorem}[section]

\begin{document}
	\begin{titlepage}
		\begin{center}
			\textbf{University of Wrocław \\ Department of Physics and Astronomy \\ Theoretical Physics}\\
			
			\vspace{40pt}
			
			\large{Aleksander Kozak \\ 268274}
			
			\vspace{60pt}
			
			\noindent\rule{15cm}{0.5pt}\\
			\Huge{Scalar-tensor gravity in the Palatini approach}\\
			\noindent\rule{15cm}{0.5pt}
			
			\vspace{100pt}
			
			\normalsize{Master thesis under the supervision of} \\
			\Large{prof. dr hab. Andrzej Borowiec}
			
			\vspace{120pt}
			
			\normalsize{Wrocław \\ 2017}
			
		\end{center}
		
	\end{titlepage}
	\begin{abstract}
		The main objective of this thesis is to discuss scalar-tensor theories in the Palatini approach. Both scalar-tensor theories and Palatini formalism are means of alternating classical theory of gravity, general relativity, in order to account for phenomena being seemingly unexplainable on the ground of the Einstein theory or to serve as toy models used to test limitations of the theory in question. In the literature both Palatini approach and scalar-tensor theories have been widely discussed, but there are very few - if none - authors writing about a merge of these two ideas. The present paper is a result of an insufficient attention given to the topic of scalar-tensor theories in Palatini formalism.
		
		In the course of the thesis action functional for scalar-tensor theories of gravity will be introduced. This action functional differs significantly from the action defined in case of scalar-tensor theories in metric approach. We aim at analysing the theory using the language of invariants, allowing us to write down all equations in a frame-independent manner. We discover that invariants defined for the metric case not always have their counterparts in Palatini formalism. Also, two frames most frequently used in the literature are discussed: Einstein and Jordan frames. Possible applications of the theory developed in the first part of the thesis are presented. We show the equivalence between $f(R)$ and scalar-tensor theories of gravity, and exploit this fact by analysing the former using methods developed for scalar-tensor theories. We conclude the thesis with calculating the Friedmann equations for an empty universe of vanishing spatial curvature and preparing set-up for analysing inflationary behaviour. 
		
		\noindent\rule{8cm}{0.4pt}
		
		Głównym celem niniejszej pracy jest dyskusja skalarno-tensorowych teorii grawitacji w ujęciu Palatiniego. Zarówno teorie skalarno-tensorowe, jak i podejście Palatiniego są środkami modyfikacji klasycznej teorii grawitacji, jaką jest ogólna teoria względności, w celu wyjaśnienia zjawisk prawdopodobnie niewyjaśnialnych na gruncie teorii Einsteina lub w celu posłużenia jako \textit{toy models}, których używa się do badania ograniczeń kwestionowanej teorii. W literaturze podejście Palatiniego i teorie skalarno-tensorowe są szeroko dyskutowane, ale istnieje bardzo niewielu - lub wręcz nie ma wcale - autorów łączących oba pomysły. Niniejsza praca powstała jako skutek niedostatecznej uwagi poświęconej zagadnieniu teorii skalarno-tensorowych w podjeściu Palatiniego.
		
		W toku pracy wprowadzony będzie funkcjonał działania dla teorii skalarno-tensorowych. Działanie to znacznie różni się od działania dla teorii w podejściu metrycznym. Naszym celem jest analiza teorii używając języka niezmienników, pozwalającego na zapisanie równań w sposób niezależny od konforemnego układu odniesienia. Odkrywamy, że niezmienniki zdefiniowane dla przypadku metrycznego nie zawsze mają swoje odpowiedniki w podejściu Palatiniego. Dwa układy konforemne sa analizowane: Einsteina i Jordana. Możliwe zastosowania rozwiniętego formalizmu są prezentowane w dalszej części pracy. Pokazujemy równoważność pomiędzy teoriami $f(R)$ a grawitacją skalarno-tensorową, i używamy tego faktu aby analizować te pierwsze za pomocą metod rozwiniętych dla tych drugich. Rozdział zamknięty jest dyskusją równań Friedmanna dla pustego wszechświata o znikającej krzywiźnie przestrzennej oraz przygotowaniem zaplecza matematycznego do analizy zachowania inflacyjnego. 
		
		\vspace{10pt}
		
		\noindent\footnotesize{Keywords: \\ \textit{scalar-tensor gravity, $f(R)$ theories of gravity, conformal transformation, Jordan frame, Einstein frame, Palatini formalism, conformal invariants, extended theories of gravity, general theory of relativity, cosmology, Friedmann equations}}
		
			\vspace{5pt}
			
			\noindent\textit{\footnotesize{teorie skalarno-tensorowe, teorie $f(R)$, transformacja konforemna, układ Jordana, układ Einsteina, formalizm Palatiniego, niezmienniki konforemne, rozszerzone teorie grawitacji, ogólna teoria wzgledności, kosmologia, równania Friedmanna}}
	\end{abstract}
	\tableofcontents
	\chapter{Introduction: on the need for alternative theories of gravity}
	
			General theory of relativity (GR) founded by Albert Einstein in 1915 has been a very successful, self-consistent theory of gravity, not only accounting for different phenomena which until the moment of its formulation had not been satisfactorily explained, but also predicting existence of objects that would seem exotic, such as black holes. The theory itself has been thoroughly tested in the course of last century, confirming that the theory's foundations are well-motivated. For example, the Einstein Equivalence Principle (EEP) which can be regarded as a cornerstone of GR states that: the Weak Equivalence Principle (WEP) is valid, together with local Lorentz invariance (LLI) and local position invariance (LPI) saying that the outcome of any experiment cannot depend on the observer's position in the Universe \cite{test}. If EEP is valid, then gravity is necessarily an effect of curved spacetime. Practical consequences of EEP are the following: spacetime must have a symmetric metric of a Lorentzian signature, which, in turn, determines geodesics on that spacetime; also, locally in a freely falling frame one can use special relativity (SR) to describe all non-gravitational laws of physics. Since EEP narrows down the number of possible theories of gravity and compels us to choose a metric theory of gravity satisfying the postulates written above, the principle must be carefully tested. WEP was tested in the famous E$\ddot{\text{o}}$tv$\ddot{\text{o}}$s experiment \cite{eot}, LLI in Michelson-Morley-type experiments \cite{mich},\cite{sch},\cite{bril}, in test of time dilation \cite{ros}, tests of independence of the velocity of light of the source \cite{test},\cite{alv}, tests of isotropy of the speed of light \cite{test}, \cite{ris}, and LPI has been tested by the gravitational redshift experiment \cite{will}. Results of these experiments indicate unequivocally that EEP is valid, so that any meaningful theory of gravity should give the same predictions. Moreover, predicted by GR gravitational waves have been (finally) found \cite{abb}, which proved also the existence of black holes. 
			
			Despite many experimental triumphs, GR is not considered a fundamental theory describing gravitational interactions. There are various - theoretical and experimental - arguments for modifying the Einstein theory of gravity. First of all, GR cannot be satisfactorily quantized, as the attempts to renormalize it were futile \cite{capo}. It was shown that adding extra terms to the Einstein-Hilbert Lagrangian was necessary, which led to a vast modification of the standard theory, yielding fourth-order equations of motion. Secondly, GR is not a low-energy limit of theories regarded as fundamental. Low-energy limit of the string theory reproduces Brans-Dicke theory, not GR \cite{capo}, since dilaton fields couple non-minimally to the spacetime curvature. Another reason is that we have not hitherto carried out any reliable tests of GR working on the large scale. As far as cosmology is concerned, it is rather a dubious way of confirming the Einstein's theory on a scale of galaxies and the whole Universe since most of alternative theories of gravity admit Friedman-Lemaitre-Robertson-Walker metric. Conversely, cosmology suggests that in fact GR should be modified in order to account for the accelerated expansion of the Universe.
			
			A sound reason to modify GR is an attempt to incorporate fully the Mach's Principle (MP) into the theory. General relativity admits solutions which are anti-Machian, such as G$\ddot{\text{o}}$del Universe \cite{capo}, \cite{godel}. MP states that a local inertial frame is influenced by a motion of all matter in the Universe which, in practical terms, means that the Newton constant $G$ is not a real constant, but its value varies with the spacetime position. A theory which incorporates MP is the Brans-Dicke theory \cite{brans}, conceived mostly because of a philosophical need for including that principle in the formalism of the theory. 
			
			Another problem with the $\Lambda$CDM model, based on GR and Standard Model of particles enriched with the cosmological constant playing the role of the dark energy, and the mysterious dark matter, is that the value of $\Lambda$ being responsible for the current acceleration of the expansion of the Universe is incomprehesibly small when compared to the value predicted by quantum field theory (by 120 orders of magnitude smaller). Also, it is unclear why the value of energy density associated to the cosmological constant is comparable to matter energy density (the so-called coincidence problem) \cite{capo}. 
			
			In fact, adding cosmological constant to Einstein Field Equations (EFE) is not the only way to achieve accelerated expansion of the Universe. The cosmic speed-up requires only a component of negative pressure which dominates the energy content at the present epoch. This general requirement does not tell us anything about the nature of the component. Due to this fact, it is possible to explain the cosmic acceleration adding a fluid which behaves like dark matter at high densities and dark energy at low densities. For example, it can be achieved by introducing the Chaplygin gas (which may also account for the cosmic inflation) \cite{capo}, \cite{chap}, \cite{bor}. It is also possible to explain the accelerated expansion by modifying the geometric part of EFE. Instead of adding yet another component, we can simply introduce extra terms to the FRWL equations, leaving the matter unchanged \cite{mod}, \cite{capo}.
			
			As far as the mathematical reasons for modifying the Einstein's gravity are concerned, we can take the so-called Palatini formalism into consideration. Palatini formalism is going to be discussed at length in the first chapter, but here we can sketch out the general idea. In the standard gravity, the underlying assumption of geometric structures defined on spacetime is that the affine connection is the Levi-Civita connection of the metric. In Palatini approach, however, we regard these two objects as unrelated, since there is no reason whatsoever we should impose a relation between them a priori. In case of Einstein gravity, introducing Palatini formalism does not affect the resulting field equations in any way; however, in case of more complicated theories, such as scalar-tensor or $f(R)$ theories of gravity, both approaches usually give us different results, describing different physics \cite{capo}.
			
			The scalar-tensor theories of gravity, which will be diligently analyzed in this paper, are a very promising modification of the Einstein gravity. The main idea of the scalar-tensor theories of gravity is going to be introduced in the second chapter. In these theories, a scalar field is nonminimally coupled to the curvature scalar \cite{fuj}; provenance of the field will be discussed later on. Historically, the prototype of all contemporary scalar-tensor theories was the Brans-Dicke theory, already mentioned.  An interesting feature of the scalar-tensor theories of gravity is their equivalence to $f(R)$ theories, which basically means that the latter can be analyzed using the 'mathematical machinery' developed for the former. The reason why the scalar-tensor theories deserve some attention is that they can be successfully used to build credible models for cosmic inflation (where a scalar field called 'inflaton' is driving the accelerated expansion of the Universe \cite{guth}; this field, however, is introduced somehow ad hoc since a detailed particle physics mechanism remains unknown) and dark energy \cite{kuinv}.
			
			So far, the scalar-tensor theories of gravity have been considered mostly in a purely metric approach \cite{capo}, \cite{fuj}, \cite{kuinv}, \cite{sot} and the possible effects of adopting the Palatini approach have not been yet investigated in the literature. The idea that changing the formalism may lead to a different theory and thence to different experimental predictions which can help with discriminating between competing theories seems plausible. The main goal of this paper is to develop mathematical background for investigating scalar-tensor theories of gravity in Palatini approach and preparing them to be verified by means of comparing their predictions (especially post-Newtonian parameters) with actual data. 
			
			In the second chapter we give a brief overview of two alternative theories which are of a great importance to this paper - scalar-tensor and $f(R)$ theories. The notion of conformal transformation will be then introduced. The chapter will end with a discussion of Palatini formalism. The way it changes our view on the geometric structure of spacetime will be particularly stressed. In the third chapter the analysis of the scalar-tensor theories in Palatini approach will begin. Modified formulae relating geometric quantities of two different conformal frames will be presented. Next, we will postulate an action functional for scalar-tensor theories whose form remains invariant under a conformal transformation. This will be followed by a detailed analysis of how the arbitrary coefficients entering the action must transform. In the subsequent section the notion of invariants will be introduced. Having obtained quantities which remain invariant under the conformal change, we will attempt at writing the action fully in terms of them. Field equations will be also obtained, and conservation laws shall be discussed. In the fourth chapter, we consider practical applications of the formalism we developed. $f(R)$ theories will be analysed using the language of invariants and Friedmann equations will be presented (in case of vanishing spatial curvature and without any sources). At the very end of the paper, we present the conclusions we draw from the analysis and give a possible outline of future investigations in this field.
	\chapter{Preliminaries}
			In this chapter certain notions essential for understanding the following parts of the thesis will be discussed. First off, the notion of a 'conformal transformation' will be introduced, together with formulae relating two geometric objects calculated in two distinct conformal frames. Conformal transformation itself will be a very important tool since it establishes a mathematical equivalence between different parametrizations, although physical predictions may be incommensurable. Then, we will give a description of two aforementioned theories belonging to the big and fertile family of extended or modified theories of gravity: $f(R)$ and scalar-tensor theories. These theories will be discussed in the purely metric approach. At the end of the chapter Palatini approach - one of the possible ways of modifying theories of gravity - will be introduced and thoroughly discussed. 
			\section{Conformal transformations}
			In this section a mathematical tool called 'conformal transformation' will be introduced. This notion is of great relevance to the following parts of the thesis. The main idea of the conformal transformation is that it transforms a metric tensor on a given spacetime into another metric tensor according to the rule:
			\begin{equation}
			g_{\mu\nu}(x)\rightarrow \bar{g}_{\mu\nu}(x)=\Omega^2(x)g_{\mu\nu}(x)
			\end{equation}
			where $\Omega(x)$ is an arbitrary, nonvanishing function of spacetime position. This transformation is not equivalent to a change of coordinate frame since it does not preserve the line element:
			$$\bar{ds}^2=\Omega^2(x)ds^2$$
			This implies that conformal transformation changes distance between two points in a way that is not uniform, but depends on the position. What the conformal transformation leaves unchanged is the quotient of lengths of two vectors attached to the same point. For any two vectors $u^\alpha\neq0$ and $v^\alpha\neq0$, we have:
			$$\frac{g_{\mu\nu}u^\mu u^\nu}{g_{\alpha\beta}v^\alpha v^\beta}=\frac{\bar{g}_{\mu\nu}u^\mu u^\nu}{\bar{g}_{\alpha\beta}v^\alpha v^\beta}$$
			which also means that the angle between two vectors remains unchanged \cite{sok}:
			$$\text{cos}\angle(\textbf{v},\textbf{u})=\frac{\textbf{u}\circ\textbf{v}}{\parallel\textbf{u}\parallel\:\parallel\textbf{v}\parallel}$$
			Of course, on manifolds with a pseudo-Riemannian metric tensor the notion of angle between two vectors does not make any sense in general, but the quotient of the lengths is still invariant - as long as the vectors do not have zero length. The fact that the angle between two vectors is preserved applies to null vectors in particular, which preserves the causal structure of spacetime.
			
			We say that a space $(\mathcal{M},g)$ is conformally flat if there exists a coordinate frame $(x^\alpha)$ in which $g_{\mu\nu}=\Omega^2(x)\delta_{\mu\nu}$ for some function $\Omega(x)$. In case of a pseudo-Riemannian metric, it is deemed conformally flat if the metric tensor is a diagonal matrix with entries being $\pm\Omega^2(x)$. 
			
			In the metric approach, geometric quantities describing curvature of spacetime are function of metric tensor and its derivatives. Hence, if we perform a conformal change of the metric, all quantities dependent on it will change accordingly. We limit our interest to the Riemann and Ricci tensors, and scalar curvature. Formulae describing such change can be found in any textbooks on differential geometry, but for the sake of completeness we can write them below (also, in the next chapter we will introduce conformal frame in so-called Palatini formalism, and it will be possible to compare outcomes of calculations carried out in both approaches):
			\begin{equation}
			\bar{g}_{\mu\nu}=\Omega^2 g_{\mu\nu},\quad \bar{g}^{\mu\nu}=\Omega^{-2}g^{\mu\nu},\quad\text{det}(\bar{g}_{\mu\nu})=\Omega^{2N}\text{det}(g_{\mu\nu})
			\end{equation}
			where $N$ is the dimension of the manifold. Christoffel symbols, defined to be $\Gamma^\alpha_{\mu\nu}=\frac{1}{2}g^{\alpha\beta}\Big(\partial_\mu g_{\nu\beta}+\partial_\nu g_{\mu\beta}-\partial_\beta g_{\mu\nu}\Big)$ transform as follows:
			\begin{equation}
			\bar{\Gamma}^\alpha_{\mu\nu}=\Gamma^\alpha_{\mu\nu}+\Omega^{-1}\Big(2\delta^\alpha_{(\mu}\nabla_{\nu)}\Omega-g_{\mu\nu}g^{\alpha\beta}\nabla_\beta\Omega\Big)
			\end{equation}
			Riemann tensor, defined as $R^\alpha_{\:\mu\beta\nu}=\partial_\beta\Gamma^\alpha_{\mu\nu}-\partial_\nu\Gamma^\alpha_{\mu\beta}+\Gamma^\alpha_{\beta\sigma}\Gamma^\sigma_{\mu\nu}-\Gamma^\alpha_{\nu\sigma}\Gamma^\sigma_{\beta\mu}$, changes in the following way \cite{blasch}:
			\begin{equation}
			\begin{split}
			\bar{R}^\alpha_{\:\mu\beta\nu}&=R^\alpha_{\:\mu\beta\nu}+\frac{1}{\Omega}\Big[\delta^\alpha_\nu\Omega_{;\mu\beta}-\delta^\alpha_\beta\Omega_{;\mu\nu}+g_{\mu\beta}\Omega^{;\alpha}_{;\nu}-g_{\mu\nu}\Omega^{;\alpha}_{;\nu}\Big] +\frac{1}{\Omega^2}\Big[\delta^\alpha_\nu g_{\mu\beta}-\delta^\alpha_\beta g_{\mu\nu}\Big]g_{\sigma\tau}\Omega^{;\sigma}\Omega^{;\tau}+\\
			&+\frac{2}{\Omega^2}\Big[\delta^\alpha_\beta\Omega_{;\mu}\Omega_{;\nu}-\delta^\alpha_\nu\Omega_{;\mu}\Omega_{;\beta}+g_{\mu\nu}\Omega^{;\alpha}\Omega_{;\beta}-g_{\mu\beta}\Omega^{;\alpha}\Omega_{;\nu}\Big]
			\end{split}
			\end{equation}
			The Ricci tensor in the new conformal frame is given by the following formula (written using a slightly different convention \cite{capo}, \cite{shap}):
			\begin{equation}
			\begin{split}
			&\bar{R}_{\mu\nu}\equiv\bar{R}^\alpha_{\mu\alpha\nu}=R_{\mu\nu}-(N-2)\nabla_\mu\nabla_\nu\text{ln}\Omega-g_{\mu\nu}\Box\text{ln}\Omega+(N-2)\nabla_\mu\text{ln}\Omega\nabla_\nu\text{ln}\Omega-\\
			&+(N-2)g_{\mu\nu}g^{\alpha\beta}\nabla_\alpha\text{ln}\Omega\nabla_\beta\text{ln}\Omega
			\end{split}
			\end{equation}
			And, finally, the curvature scalar is \cite{capo}:
			\begin{equation}
			\bar{R}=\bar{g}^{\mu\nu}\bar{R}_{\mu\nu}=\Omega^{-2}\Big[R-2(N-1)\Box\text{ln}\Omega-(N-1)(N-2)\frac{g^{\mu\nu}\nabla_\mu\Omega\nabla_\nu\Omega}{\Omega^2}\Big]
			\end{equation}
			We can introduce a tensor which remains invariant under a conformal change. This tensor is called 'Weyl tensor' $C^{\alpha}_{\:\mu\beta\nu}$, and is defined to be the part of the Riemann tensor which is independent of the Ricci tensor (since some of the Riemann tensor components can be expressed in terms of Ricci tensor). The Weyl tensor is given by the following formula \cite{sok}:
			\begin{equation}
			\begin{split}
			R^{\alpha}_{\:\mu\beta\nu}= & C^{\alpha}_{\:\mu\beta\nu}+\frac{1}{N-2}\Big(R^\alpha_\beta g_{\mu\nu}+R_{\mu\nu}\delta^\alpha_\beta-R^\alpha_\nu g_{\mu\beta}-R_{\mu\beta}\delta^\alpha_\nu\Big)-\\
			& +\frac{R}{(N-1)(N-2)}\big(\delta^\alpha_\beta g_{\mu\nu}-\delta^\alpha_\nu g_{\mu\beta}\big)
			\end{split}
			\end{equation}
			The Weyl tensor is traceless:
			$$C^{\alpha}_{\:\mu\alpha\nu}=0$$
			It is invariant under a conformal change:
			$$C^{\alpha}_{\:\mu\beta\nu}(\Omega^2 g)=C^{\alpha}_{\:\mu\beta\nu}(g)$$
			and thence it is called a 'conformal curvature tensor'. In a conformally flat space, the Weyl tensor vanishes. Also, it can be proven that a scalar density given by the following formula: $C_{\alpha\beta\mu\nu}C^{\alpha\beta\mu\nu}\sqrt{-g}$ is an invariant quantity under a conformal change.
			
			If $\Omega$ is constant, then the transformation is called a 'scale transformation'. This simply means that we change our definition of units of length and time. Virtually, since both meter-sticks and clocks would be rescaled, laws of physics should remain unchanged, since there should be no preferred standard length or time. It seems plausible that the Universe is indeed scale invariant; however, there is no such invariance in our world \cite{fuj}. For example, if we measure particle masses using different meter-sticks and clocks, we get different values. This means that as soon as masses of particles are introduced, the scale invariance is broken. On the other hand, at very high energies mass of particles is negligible, and scale invariance provides convenient tools. Also, scale invariance may be conceived as a global transformation, and we now that all transformations should be local, i.e. depend on spacetime position. Along these lines, conformal transformation can be thought of as a localized scale transformation. 
			
			Conformal invariance can be viewed in some situations \cite{fuj} as an approximate symmetry. Massless fermions exhibit conformal invariance only under special conditions. Theory of electromagnetic field is completely conformally invariant. The massless, gravitational field does not posses such invariance however, and this is a result of presence of the Newton's constant, having the dimension of $L^3M^{-1}T^{-2}$. As it turns out, conformal invariance should not occur in realistic theories. On the other hand, they turn out to be a useful tool when investigating theories with a nonminimal coupling. 
		\section{Short review of alternative theories of gravity}
			Possible alternative theories of gravity may differ drastically from the original Einstein's gravity. The name, 'alternative theories of gravity', encompasses a variety of proposed modifications of GR, ranging from a simple change in Einstein-Hilbert Lagrangian to a completely new theories unifying quantum mechanics with gravity. Theories we will be dealing with throughout this thesis belong to the former class: they are nothing but a modest modification of GR by means of adding extra terms to the Lagrangian, including a nonminimally coupled scalar field or assuming Palatini formalism. In other words, all these theories start from the original Einstein's idea and then add some corrections, which should take over the behaviour of gravity at the regimes where it fails to explain various observed phenomena. We shall refer to those theories as 'Extended Theories of Gravity' (following \cite{scap}). 
			\subsection{Brans-Dicke theory and its generalizations}
			One of the simplest way of modifying GR is adding a scalar field into the theory. It seems very pleasing and plausible that we could alter the Einstein's gravity by including one of the most primitive of Nature's phenomena. This idea was exploited by Brans and Dicke, who in 1961 proposed a theory extending GR. They pioneered so-called scalar-tensor theories of gravity, which currently are being investigating by many researchers. Their model, however, was rather naive and revealed certain shortcomings not long after it had been conceived, but should be still regarded as a prototype of all modern scalar-tensor gravities.
			
			Original motivation for Brans and Dicke to introduce a modified gravity was the need for incorporating Mach's principle into the theory of gravity. As it was shortly discussed in the Introduction, Mach's principle states that local motion of particles is affected by a large-scale distribution of matter. Ernst Mach was an Austrian physicist and philosopher who was famous for criticizing the Newton's idea of absolute space. He targeted the bucket thought experiment which was supposedly proving the existence of absolute space. In his interpretation of the experiment, Mach argued that whether the surface of water in the vessel is at rest or curved, it is always with respect to the distant starts, which means that the difference between these two states must be related to the distant mass rather than to an absolute space \cite{ma}. A similar standpoint was held also by Einstein, who was inspired by Mach's thought when he was formulating his general theory of relativity. Ironically, Einstein's theory admits solutions which are clearly anti-Machian, such as de Sitter universe, where the evolution of universe is dominated by a cosmological constant and matter is entirely negligible \cite{bible}. The failure of GR to incorporate Mach's principle inspired Brans and Dicke to look for an alternative theory.
			
			Brans and Dicke's idea was to add a scalar field whose dependence on spacetime position would be translated into a variability of the gravitational coupling. This would clearly violate the strong equivalence principle saying that results of all experiments  carried out in freely falling laboratories should be independent of the spacetime position of the experimenter. However - Brans and Dicke argue - what we managed to establish performing very accurate experiments is the weak equivalence principle, stating that all gravitational accelerations are equal, regardless of the matter composition \cite{brans}. They postulated the following action functional:
			\begin{equation}
			S[g,\Phi]=\frac{1}{2\kappa^2}\int_{\Omega}d^4x\sqrt{-g}\Big(\Phi R- \frac{\omega}{\Phi}g^{\mu\nu}\nabla_\mu\Phi\nabla_\nu\Phi\Big)+\int_{\Omega}d^4x\sqrt{-g}\mathcal{L}_{\text{matter}}(g,\chi)
			\end{equation}
			Here, $R$ is the curvature scalar, $\frac{1}{2\kappa^2}=\frac{c^4}{16\pi}$ and $\chi$ represents generic matter fields. $\Phi$ plays a role of the gravitational coupling and its dimension is the same as the dimension of the Newton's $G^{-1}$, i.e. $MT^2L^{-3}$. What is important to notice is that the scalar field does not enter the matter part of the action. If the converse were true, WEP would be violated, and this can be demonstrated with a simple reasoning \cite{fuj}: in order to obtain the geodesic trajectory for a point mass particle in a given gravitational field, we seek extremum of the functional:
			\begin{equation}
			I(g)=-m\int d\tau
			\end{equation} 
			where $m$ is the mass of the particle and $\tau$ is the proper time. The mass can be pulled out of the integral, implying that WEP holds. If the scalar field is coupled to matter, however, factoring out the mass will not be possible anymore. This will result in lack of covariant conservation of point mass particle, and this was unacceptable in view of Dicke and Brans since high-precision experiments proved that WEP holds \cite{brans}, \cite{fuj}.
			
			Equations of motion are obtained by varying with respect to the dynamical variables. $\Phi$-variation gives us the following equation \cite{brans}:
			\begin{equation}
			2\frac{\omega}{\Phi}\Box\Phi-\frac{\omega}{\Phi^2}g^{\mu\nu}\nabla_\mu\Phi\nabla_\nu\Phi+R=0 \label{b1}
			\end{equation}
			This is a wave-like equation for $\Phi$ sourced by kinetic part of the Lagrangian density and the $\Phi R$ term. 
			
			Variation with respect to the metric tensor yields:
			\begin{equation}
			G_{\mu\nu}=\frac{\kappa^2}{\Phi}T_{\mu\nu} +\frac{\omega}{\Phi^2}\Big(\nabla_\nu\Phi\nabla_\nu\Phi-\frac{1}{2}g_{\mu\nu}g^{\alpha\beta}\nabla_\alpha\Phi\nabla_\beta\Phi\Big)+\frac{1}{\Phi}\Big(\nabla_\mu\nabla_\nu\Phi-g_{\mu\nu}\Box\Phi\Big) \label{b2}
			\end{equation}
			where the energy-momentum tensor is defined to be: $T_{\mu\nu}=\frac{2}{\sqrt{-g}}\frac{\partial(\sqrt{-g}\mathcal{L}_{\text{matter}})}{\partial g^{\mu\nu}}$. It can be shown that $\nabla_\mu T^{\mu\nu}=0$, so that all conservation laws are satisfied. 
			
			As we can see, the left hand side of the equation is the familiar Einstein tensor, and the first modification of GR shows up on the right hand side: in place of the gravitational coupling there is an inverse of the scalar field. Also, the next term is the energy-momentum tensor for the scalar field. The last term contains second derivatives of the scalar field and it originates from integrating by parts when calculating the variation of Ricci tensor. Equation \ref{b2} can now be contracted with a metric tensor, giving us a direct relation between the curvature, trace of energy-momentum tensor and the scalar field:
			\begin{equation}
			-R=\frac{\kappa^2}{\Phi}T-\frac{\omega}{\Phi^2}g^{\mu\nu}\nabla_\mu\Phi\nabla_\nu\Phi-\frac{3}{\Phi}\Box\Phi
			\end{equation}
			The definition of $R$ can be now substituted in \ref{b1}, yielding:
			\begin{equation}
			\Box\Phi=\frac{\kappa^2}{3+2\omega}T \label{eqsc}
			\end{equation}
			This is a truly remarkable result since it relates the scalar field directly to the distribution of matter. This is a clear implementation of Mach's principle. Let us also note that electromagnetic field does not contribute to the trace of energy-momentum tensor. 
			
			Brans and Dicke expected the parameter $\omega$ to be of order of unity, otherwise rendering the theory unnatural. However, testing the time-delay imposed a serious constraint on the lowest possible value of $\omega$, which happens to be:
			$$3.6\times 10^3 \leq \omega$$
			The unexpectedly high value of $\omega$ means that the theory is fine-tuned. However, fine-tuning can be avoided if we give the scalar field a sufficiently large mass, thus limiting its range. Adding a self-interaction potential can be a viable way of modifying the theory and making it still appealing \cite{far}. 
			
			Since the notion of conformal transformation has been already introduced, we may now perform a conformal change in order to see whether we can get rid of the nonminimal coupling. Let us consider the following conformal transformation of the metric:
			$$\bar{g}_{\mu\nu}=G\Phi g_{\mu\nu}$$
			where $G$ is the Newton's constant, accompanied by a redefinition of the scalar field: 
			$$\bar{\Phi}=\sqrt{\frac{2\omega+3}{G}}\text{ln}\Big(\frac{\Phi}{\Phi_0}\Big)$$
			where $\omega > -\frac{3}{2}$. The action functional now reads as follows:
			\begin{equation}
			S[\bar{g},\bar{\Phi}]=\frac{1}{2\kappa^2}\int_{\Omega}d^4x\sqrt{-\bar{g}}\Bigg(\frac{\bar{R}}{G}-\frac{1}{2}\bar{g}^{\mu\nu}\bar{\nabla}_\mu\bar{\Phi}\bar{\nabla}_\nu\bar{\Phi}\Bigg)+S_{\text{matter}}\Big[\frac{1}{G\Phi_0}e^{-\sqrt{\frac{G}{2\omega+3}}\bar{\Phi}}\bar{g},\chi\Big] \label{ace}
			\end{equation}
			As we can see, the new scalar field has the dimension of $G^{-\frac{1}{2}}$. We say that the action is now cast in so-called Einstein frame - previously it was written in Jordan frame. An extensive discussion of both frames will be presented in the next chapter. The Einstein frame action looks like the standard Einstein-Hilbert action, describing the familiar gravity. However, there are two important differences. First of all, the vacuum solution of the field equations can never be written in the form $\bar{R}_{\mu\nu}=0$ since on the right hand side there will be always the scalar field, permeating the spacetime in a way that cannot be avoided \cite{capo}. Second, the coupling between matter and scalar field now appears: $\bar{\Phi}$ enters now the matter part of the action functional. This has a profound consequence: the energy-momentum tensor is no longer conserved. Instead, its divergence is directly related to the trace of itself and derivatice of the conformal factor \cite{capo}, \cite{fatibene}:
			\begin{equation}
			\bar{\nabla}^\nu\bar{T}_{\mu\nu}=-\bar{T}\:\bar{\nabla}_\mu\text{ln}\Big(G\Phi_0 e^{\sqrt{\frac{G}{2\omega+3}}\bar{\Phi}}\Big)^{\frac{1}{2}} \label{em}
			\end{equation}
			This implies that there is a 'fifth' force acting on the particles, causing them to deviate from their standard trajectories. In order to obtain the magnitude of this force, let us carry out the following reasoning: let us assume that the energy-momentum tensor (already in the Einstein frame) is that of a dust fluid: 
			$$\bar{T}^{\mu\nu}=\bar{\rho}\bar{u}^\mu\bar{u}^\nu$$
			We can plug this definition of energy-momentum tensor back in \ref{em} and get an equation (irrelevant for us right now) which suggests that the total derivative of the tangent vector along a geodesic does not equal zero anymore, but:
			\begin{equation}
			\frac{D\bar{u}^\alpha}{d\tau}=\frac{1}{2}\sqrt{\frac{G}{2\omega+3}}\bar{\nabla}^\alpha\bar{\Phi}
			\end{equation}
			Because of spacetime dependence of the right hand side of the equation, universality of free fall is clearly violated. This is a feature of non-metric theories, and non-metricity - which is equivalent to WEP violation - is a property of a given conformal frame. In the initial frame (the Jordan frame) the energy-momentum tensor was conserved, so that the theory was metric. Here, however, in the so-called Einstein frame, its divergence is proportional to its trace. The only thing that remains unaffected by a conformal transformation is a geodesic for light, since in case of radiation, $T=0$ \cite{fatibene}. 
				\subsubsection{General scalar-tensor theories of gravity}
				Before we move on to generalizations of Brans-Dicke theory - to a wider class of scalar-tensor theories of gravity - let us analyze what a possible origin of the scalar field may be. It seems that Brans and Dicke did not ponder upon this issue, viewing the model they proposed just as a simple alternative to GR \cite{fuj}. However, the scalar field added to the Einstein-Hilbert Lagrangian may result from a deeper theories, which justify the nonminimal coupling between $\Phi$ and the curvature. For example, one of the theories supporting the scalar-tensor gravity is the Kaluza-Klein model and its modern extensions, where the scalar field appears as a result of compactification of higher dimensions \cite{fuj} - to be more precise, the scalar field is (a certain power of) the radius of an $n$-dimensional compactified space. The scalar field is in this case related to determinant of the $(4+n)$-dimensional metric, and since it enters the matter part of the action functional as well, it irrevocably must be coupled to matter. This somehow contradicts Brans and Dicke's assumptions about the scalar field not violating the WEP. Another possible explanations of origins of the scalar field involve more reliable fundamental theories, such as string theories and noncommutative geometries. In the latter case, the field can be viewed as a gauge field on a discrete space and identified with the Higgs fields \cite{fuj}. A detailed discussion of these theories is far beyond the scope of this thesis; what is important to note is that the addition of a scalar field has a sound theoretical motivation.
				
				Scalar-tensor theories of gravity use the concept of conformal transformation extensively because it provides a useful tool establishing a \textit{mathematical} equivalence between two conformal frames \cite{kuinv}. However, it has been known for a long time that two different conformal frames often describe different physics, meaning that a conceptual problem arises when we start to investigate a relation between observables and frames we use to work in. Our doubts about this ambiguity can be rephrased in the following way: is there one frame that should be regarded physical, whereas all the other ones, despite the (possible) mathematical equivalence do not define units we use to measure lengths? Is it possible that the discrepancies can be attributed to the fact that the theory has not been yet formulated in an invariant way with respect to some general space? Answer to this question is not an easy one, and this issue will be addressed in the next chapter. Some authors made a significant progress on this issue \cite{kuinv} and introduced quantities which remain invariant under a change of conformal frame, and then expressed all physical observables in terms of these invariants. This reasoning seems plausible to us and in the next chapter we will dedicate some time to obtaining such invariants; in this section, however, we will introduce a rough concept of an invariant quantity in scalar-tensor theories of gravity and show that, indeed, it remains unaffected by a conformal change. 
				
				A general scalar-tensor theory allows more unspecified functions than a pretty restricted Brans-Dicke theory. In case of the latter, the only free parameter we had at our disposal was the coefficient $\omega$; in case of the former, there will be four arbitrary functions of the scalar field entering the action functional. The postulated action for the theory looks as follows:
				\begin{equation}
				S[g,\Phi]=\frac{1}{2\kappa^2}\int_{\Omega}d^4x \sqrt{-g}\Big(\mathcal{A}(\Phi)R-\mathcal{B}(\Phi)g^{\mu\nu}\nabla_\mu\Phi\nabla_\nu\Phi-\mathcal{V}(\Phi)\Big)+S_{\text{matter}}\big[e^{2\alpha(\Phi)}g,\chi\big] \label{a1}
				\end{equation}
				Here, $\mathcal{A}(\Phi)$ describes the coupling between scalar field and the curvature scalar $R$, being a function of the metric tensor $g_{\mu\nu}$. In order to make gravity an attractive force, we must set $\infty>\mathcal{A}(\Phi)>0$. $\mathcal{B}(\Phi)$ is a kinetic coupling, and $\mathcal{V}(\Phi)$ is a self-interaction potential of the scalar field, which cannot take negative values. $\alpha(\Phi)$ is an anomalous coupling between the scalar field and matter. 
				
				Varying the action functional with respect to the metric tensor, we get the following equations of motion \cite{kuinv}:
				\begin{equation}
				\begin{split}
				&\mathcal{A}(\Phi)G_{\mu\nu}+\Big(\frac{1}{2}\mathcal{B}+\mathcal{A}''\Big)g_{\mu\nu}g^{\alpha\beta}\nabla_\alpha\Phi\nabla_\beta\Phi-\big(\mathcal{B}+\mathcal{A}''\big)\nabla_\mu\Phi\nabla_\nu\Phi+\mathcal{A}'\big(g_{\mu\nu}\Box-\nabla_\mu\nabla_\nu\big)\Phi - \\
				&+\frac{1}{2}g_{\mu\nu}\mathcal{V}-\kappa^2T_{\mu\nu}=0 
				\end{split} \label{eom}
				\end{equation}
				with the standard definition of the energy-momentum tensor $T_{\mu\nu}$. Variation with respect to the scalar field gives us:
				\begin{equation}
				R\mathcal{A}'+\mathcal{B}'g^{\mu\nu}\nabla_\mu\Phi\nabla_\nu\Phi+2\mathcal{B}\Box\Phi-\mathcal{V}'+2\kappa^2\alpha'T=0
				\end{equation}
				As we can see, as in case of the Brans-Dicke theory, the scalar field is sourced by the trace of energy-momentum tensor. The continuity equation takes the following form:
				\begin{equation}
				\nabla^\nu T_{\mu\nu}=\alpha'T\nabla_\mu\Phi
				\end{equation}
				which means that the energy-momentum tensor is conserved in those frames, where coupling between scalar field matter is not present.
				
				Two of the four arbitrary functions can be fixed by means of a proper conformal change accompanied by a redefinition of the scalar field:
				\begin{equation}
				g_{\mu\nu}=e^{2\bar{\gamma}(\bar{\Phi})}\bar{g}_{\mu\nu}
				\end{equation}
				\begin{equation}
				\Phi=\bar{f}(\bar{\Phi}) \label{ts}
				\end{equation}
				It is generally assumed that the first and second derivatives of $\bar{\gamma}$ exist. Moreover, the Jacobian of the transformation \ref{ts} is allowed to be singular at some isolated point. 
				
				If we plug the redefined scalar field and metric tensor back in the action functional, make use of the transformation relations and neglect boundary terms arising while integrating by parts, we end up with the action written in a different conformal frame, with the barred dynamical variables. In order for the Lagrangian to retain its form, the coefficients must transform in the following way:
				\begin{itemize}
					\item $\bar{\mathcal{A}}(\bar{\Phi})=e^{2\bar{\gamma}(\bar{\Phi})}\mathcal{A}(\bar{f}(\bar{\Phi}))$
					\item $\bar{\mathcal{B}}(\bar{\Phi})=e^{2\bar{\gamma}(\bar{\Phi})}\Bigg(\Big(\frac{d\Phi}{d\bar{\Phi}}\Big)^2\mathcal{B}(\bar{f}(\bar{\Phi}))-6\Big(\frac{d\bar{\gamma}}{d\bar{\Phi}}\Big)^2\mathcal{A}(\bar{f}(\bar{\Phi}))-6\frac{d\bar{\gamma}}{d\bar{\Phi}}\frac{d\mathcal{A}}{d\Phi}\frac{d\Phi}{d\bar{\Phi}}\Bigg)$
					\item $\bar{\mathcal{V}}(\bar{\Phi})=e^{4\bar{\gamma}(\bar{\Phi})}\mathcal{V}(\bar{f}(\bar{\Phi}))$
					\item $\bar{\alpha}(\bar{\Phi})=\alpha(\bar{f}(\bar{\Phi}))+\bar{\gamma}(\bar{\Phi})$
				\end{itemize}
				If these relations hold, then the action functional in a new conformal frame preserves its form:
				\begin{equation}
				S[\bar{g},\bar{\Phi}]=\frac{1}{2\kappa^2}\int_{\Omega}d^4x \sqrt{-\bar{g}}\Big(\bar{\mathcal{A}}(\bar{\Phi})\bar{R}-\bar{\mathcal{B}}(\bar{\Phi})\bar{g}^{\mu\nu}\bar{\nabla}_\mu\bar{\Phi}\bar{\nabla}_\nu\bar{\Phi}-\bar{\mathcal{V}}(\bar{\Phi})\Big)+S_{\text{matter}}\big[e^{2\bar{\alpha}(\bar{\Phi})}\bar{g},\chi\big]
				\end{equation}
				The transformation relations suggest that the condition imposed on $\mathcal{A}$ and $\mathcal{V}$ are satisfied in any conformal frame. In particular, if the potential vanishes in one conformal frame, then it is equal to zero in all related conformal frames. 
				
				If we investigate the way coefficients transform, we will notice that it is possible to write out such coefficients or combinations thereof which gain only a multiplier and remain otherwise unchanged. Following \cite{kuinv}, we can write:
				\begin{itemize}
					\item $\bar{\mathcal{A}}(\bar{\Phi})=e^{2\bar{\gamma}(\bar{\Phi})}\mathcal{A}(\bar{f}(\bar{\Phi}))$
					\item $e^{2\bar{\alpha}}(\bar{\Phi})=e^{2\bar{\gamma}(\bar{\Phi})}e^{2\alpha(\bar{f}(\bar{\Phi}))}$
					\item $\bar{\mathcal{V}}(\bar{\Phi})=e^{4\bar{\gamma}(\bar{\Phi})}\mathcal{V}(\bar{f}(\bar{\Phi}))$
					\item $\bar{\mathcal{F}}(\bar{\Phi}):=\sqrt{\frac{2\bar{\mathcal{A}}(\bar{\Phi})\bar{\mathcal{B}}(\bar{\Phi})+3(\bar{\mathcal{A}}'(\bar{\Phi}))^2}{4\bar{\mathcal{A}}^2(\bar{\Phi})}}=\bar{f}'\mathcal{F}(\bar{f}(\bar{\Phi}))$
				\end{itemize}
				By picking proper combinations of these quantities we can build invariants, which preserve their form under a conformal change (they are still expressed as functions of the same coefficients):
				\begin{enumerate}
					\item $\mathcal{I}_1(\Phi)=\frac{\mathcal{A}(\Phi)}{e^{2\alpha(\Phi)}}$
					\item $\mathcal{I}_2(\Phi)=\frac{\mathcal{V}(\Phi)}{(\mathcal{A}(\Phi))^2}$
					\item $\mathcal{I}_3(\Phi)=\pm \int_{\Phi_0}^{\Phi}d\Phi'\mathcal{F}(\Phi')$
				\end{enumerate}
				Alongside the invariants defined above, we may introduce an invariant metric, remaining unchanged under a conformal transformation:
				$$\hat{g}_{\mu\nu}=\mathcal{A}(\Phi)g_{\mu\nu}$$
				(invariance of this metric follows from transformation properties of both $\mathcal{A}$ and the metric tensor $g_{\mu\nu}$).
				
				Having introduced the invariants, we may now write the action functional in terms of them. We can also assume that the relation defining invariant $\mathcal{I}_3$ is invertible, so that we can express the scalar field as a function of it. This will give us an obvious advantage of frame-independence of resulting field equations; also, all observables will be expressed in terms of the invariants. 
				\begin{equation}
				S[\hat{g},\mathcal{I}_3]=\frac{1}{2\kappa^2}\int_\Omega d^4x \sqrt{-\hat{g}}\Big(\hat{R}-2\hat{g}^{\mu\nu}\hat{\nabla}_\mu\mathcal{I}_3\hat{\nabla}_\nu\mathcal{I}_3-\mathcal{I}_2\Big)+S_{\text{matter}}[\frac{1}{\mathcal{I}_1}\hat{g},\chi] \label{ex}
				\end{equation}
				As we can see, we ended up in an Einstein-like frame, where the scalar field is not coupled to the curvature, but enters the matter part of the action, hence violating the WEP. Correspondence with the action \ref{ace} is apparent. Performing variation with respect to the invariant $\mathcal{I}_3$ and the invariant metric, we obtain field equation written in a frame-independent form:
				\begin{itemize}
					\item $\delta\hat{g}$: $\hat{G}_{\mu\nu}+\hat{g}_{\mu\nu}\hat{g}^{\alpha\beta}\hat{\nabla}_\alpha\mathcal{I}_3\hat{\nabla}_\beta\mathcal{I}_3-2\hat{\nabla}_\mu\mathcal{I}_3\hat{\nabla}_\nu\mathcal{I}_3+\frac{1}{2}\hat{g}_{\mu\nu}\mathcal{I}_2-\kappa^2\hat{T}_{\mu\nu}=0$
					\item $\delta\mathcal{I}_3$: $\hat{\Box}\mathcal{I}_3-\frac{d\mathcal{I}_2}{d\mathcal{I}_3}+\frac{\kappa^2}{4}\frac{d\text{ln}\frac{1}{\mathcal{I}_1}}{d\mathcal{I}_3}\hat{T}=0$
				\end{itemize}
				As we expected, the scalar field is now sourced by matter fields (expressed by the trace of energy-momentum tensor). As we know from our previous considerations, inasmuch as the scalar field is coupled to matter, continuity equations do not hold anymore. In order to avoid it, giving up on simplicity of the Einstein-like frame, we could have defined another invariant metric: $\tilde{g}_{\mu\nu}=e^{2\alpha(\Phi)}g_{\mu\nu}$ and express the action functional in a Jordan-like frame. Both possibilities will be discussed at length in the next chapter.
			\subsection{$f(R)$ theories of gravity}
				Adding a scalar field to the original Einstein-Hilbert Lagrangian was, as it turned out, a procedure that was theoretically well-motivated. However, it is not the only possible way to alter the original theory of gravity. Because GR originates from a very specific choice of the Lagrangian - one could dub that choice 'the simplest' - a natural question arises whether the gravity is uniquely described by Einstein's model. Not long after general relativity had been introduced, people started thinking about altering the theory, usually just for the sake of understanding it better. In 1923 Eddington himself endeavoured to modify GR by adding higher order invariants \cite{sotir}. The initial attempts to alter GR were poorly justified, since it seemed pointless from the methodological point of view to complicate the theory without any serious reason. However, a trigger to modify it was yet to come; in the Introduction some of the theoretical and experimental motivations for modifying GR were discussed. At a certain moment it became clear that Einstein's theory was not compatible with quantum theory, and also new evidence showed up suggesting that large-scale phenomena are not well-described by GR (we could mention also the inflation problem). For example, higher order theories of gravity appear when one attempts to perform quantization on a curved spacetime and tackle the renormalization problem \cite{capo}, \cite{scap}. These issues fueled the scientists' interest in alternative theories of gravity, which now is a diverse and still developing field. The wide array of possible theories of gravity results from the fact that one can modify it in various ways. In this section, we will analyze one particular way of changing the Einstein's theory: we can straightforwardly replace the curvature scalar in the Einstein-Hilber Lagrangian with a function thereof. It means that now the action functional takes the following form:
				\begin{equation}
				S[g]=\frac{1}{2\kappa^2}\int_{\Omega}d^4x\sqrt{-g}f(R)+S_{\text{matter}}[g,\chi] \label{actf}
				\end{equation}
				Here, $\kappa^2=\frac{8\pi G}{c^4}$. Surprisingly, such a simple replacement can account for many observed phenomena, since the function $f(R)$, viewed as a series expansion, contains terms which are of a phenomenological interest. It must be noted at this point that $f(R)$ theories do not purport to be fundamental theories of gravity. Their real value is that they can be used to explain certain processes in a way alternative to GR, thus providing us with some insight into how the Einstein's theory works. This means that $f(R)$ theories are toy-theories one uses in order to question a particular theory and inquire about its limitations. $f(R)$ theories are analyzed mostly in two distinct approaches: in purely metric, where the only dynamical variable entering the Lagrangian is metric tensor, and in Palatini approach, where linear connection is thought of as being independent of metric tensor. The latter approach will be analyzed later in this chapter. Here, we focus our attention on metric $f(R)$ theories.
				
				 Having the action \ref{actf}, we can perform variation with respect to the metric tensor. The derivation process is not essential here and will be dealt with in the next chapter, so we just give the result of the calculations:
				 \begin{equation}
				 f'(R)R_{\mu\nu}-\frac{1}{2}f(R)g_{\mu\nu}-(\nabla_\nu\nabla_\nu-g_{\mu\nu}\Box)f'(R)=\kappa^2 T_{\mu\nu} \label{feqf}
				 \end{equation}
				 (where $f'(R)=\frac{df}{dR}$). This gives us a set of fourth order field equations in the metric derivatives. A subtlety related to boundary terms was skipped in the process of deriving these equations. This is a problem stemming from the fact that fixing the metric tensor variations on the boundary does not mean that the related term will vanish, and is similar to that of GR, where the boundary term is offset by the Gibbons-Hawking-York surface term \cite{sotir}, which is a total divergence added to the Lagrangian. Unfortunately, no such term can be found in case of $f(R)$ theories, and the boundary term must be removed in a different way. Usually, apart from fixing variations of the metric tensor on the boundary, one can fix some other terms. There is no unique prescription, however, for doing that, and choosing different degrees of freedom affects the Hamiltonian formulation of the theory. On the level of obtaining field equations this choice is luckily devoid of any meaning. For a more detailed discussion, see \cite{cast}. 
				 
				 If we now perform contraction of \ref{feqf} with the metric tensor, we get so-called master equation:
				 \begin{equation}
				 f'(R)R-2f(R)+3\Box f'(R)=\kappa^2 T
				 \end{equation}
				 This equation relates the curvature scalar and the trace o energy-momentum tensor in a way which is highly nontrivial. In case of GR, we had a direct, algebraic relation between these two quantities. Here, the relation is differential, and due t this fact $f(R)$ theories admit a wider variety of solutions compared to GR. For example, setting $T=0$ does not necessarily imply that $R=\text{const}$. On the other hand, GR field equations with cosmological constant for vacuum can be easily reproduced. One simply sets $R=\text{const}$ (maximally symmetric solution \cite{sotir}), and gets the following master equation:
				 $$f'(R)R-2f(R)=0$$
				 Let us assume that there is a value of $R$ which solves this equation: if $R=c$, then $f'(c)c-2f(c)=0$ (and hence, $c=\frac{2f(c)}{f'(c)}$). We now come back to the field equations for vacuum:
				 $$f'(c)R_{\mu\nu}-\frac{1}{2}g_{\mu\nu}f(c)=0$$
				 which is equivalent to:
				 $$R_{\mu\nu}=\frac{1}{4}g_{\mu\nu}\frac{2f(c)}{f'(c)}\equiv\frac{1}{4}g_{\mu\nu}c$$
				 This is maximally symmetric de Sitter spacetime, or GR with cosmological constant added.
				 
				 Another important aspect of the $f(R)$ theories in metric approach is that the energy-momentum tensor is conserved \cite{koivisto}:
				 \begin{equation}
				 \nabla_\mu T^{\mu\nu}=0
				 \end{equation}
				 
				 It can be shown that $f(R)$ theories of gravity are equivalent to scalar-tensor theories by means of a Legendre transformation. More on this topic can be found in the fourth chapter of this thesis. This is an important equivalence since it allows us to investigate $f(R)$ theories using tools developed for analyzing scalar-tensor theories of gravity.
				 
		\section{Palatini formalism}
		As it was mentioned in the Introduction, one of the mathematical motivations for modifying GR was a discernment between metric and affine structures of spacetime. In practical terms, viewing these two structures as independent of each other means that we no longer consider the linear connection $\Gamma^\alpha_{\mu\nu}$ on the spacetime to be a Levi-Civita connection of the metric tensor $g_{\mu\nu}$. It simply means that quantities which previously were functions of the metric tensor, such as the Riemann tensor, depend now solely on a metric-independent connection. This possibility was first analysed by Einstein himself, and the idea dates back to 1925 \cite{scap}, but due to a historical misunderstanding, it was dubbed a 'Palatini approach', named after an Italian physicists Attilio Palatini. Einstein also applied this formalism to GR, but in case of the Einstein-Hilbert Lagrangian both approaches result in the very same equations of motion because the independent connection turns out to be a Levi-Civita connection of $g_{\mu\nu}$. Therefore, due to the simplicity of E-H Lagrangian there is no particular reason to impose Palatini variation. The situation changes dramatically in case of the Extended Theories of Gravity, where the E-H Lagrangian is replaced with a more complicated function or a scalar field is added \cite{sotir}, \cite{capo}. Metric approach and Palatini approach are no longer compatible and, besides giving different equations, they also describe different physics and give contradictory predictions. 
		
		Physical understanding of the decoupling metric from linear connection is the following: the metric of a Lorentzian signature sets up the geometric structure of spacetime and allows one to measure distances, volumes and time, as well as it makes establishing causal structure on the spacetime possible. The (torsionless) connection, on the other hand, defines free-fall (and hence, parallel transport). Thus, the Principle of Equivalence and the Principle of Causality become now independent. In this way, the Palatini approach 'enriches the geometric structure of spacetime and generalizes the metric approach' \footnote{	S. Capozziello, V. Faraoni, \textit{Beyond Einstein Gravity: A Survey of Gravitational Theories for Cosmology and Astrophysics}, Springer (2011) page 68}.  
		
		Let us now focus on a concrete application of the Palatini formalism to gravity theories - to $f(R)$ theories, to be more specific. The action functional we postulate is exactly the same as in \ref{actf}, but with a subtle change: now the curvature scalar is regarded as a function of both the metric tensor and the connection: $\mathcal{R}(g,\Gamma)=g^{\mu\nu}\mathcal{R}_{\mu\nu}(\Gamma)$, so that:
		\begin{equation}
		S[g, \Gamma]=\frac{1}{2\kappa^2}\int_{\Omega}d^4x\sqrt{-g}f(\mathcal{R})+S_{\text{matter}}[g,\chi] \label{actpa}
		\end{equation}
		As we can see, the matter part still depends only on the metric tensor (and on some generic matter fields); the novelty is the gravity part. The fact that matter couples only to $g_{\mu\nu}$ means that either we are restricting ourselves only to considering some special fields or parallel transport is defined by the Levi-Civita connection of the metric $g_{\mu\nu}$ \cite{sotir}. These two facts render the theory \textit{metric} since it satisfies conditions imposed on a metric theory of gravity. In particular, it means that the energy-momentum tensor is conserved if we calculate the covariant derivative using the Levi-Civita connection (and it is not conserved if we choose to calculate the divergence using the independent connection) \cite{koivisto}.  
		
		If we perform variation of the action \ref{actpa} with respect to both dynamical variables, we obtain the following field equations:
		\begin{enumerate}
			\item $\delta g$: $f'(\mathcal{R})\mathcal{R}_{\mu\nu}-\frac{1}{2}g_{\mu\nu}f(\mathcal{R})=\kappa^2 T_{\mu\nu}$
			\item $\delta \Gamma$: $\nabla^\Gamma_\alpha\Big(\sqrt{-g}f'(\mathcal{R})g^{\mu\nu}\Big)=0$
		\end{enumerate}
		The first equation resembles somehow the Einstein field equations. The second one carries information about a relation between the metric tensor and the connection. This might seem to be a circular reasoning since in the argument of the derivative we have a function of the curvature scalar, which is yet to be determined. However, we can contract the first equation with the metric tensor $g^{\mu\nu}$ and get the master equation relating algebraically the curvature scalar and matter fields represented by the trace of the energy-momentum tensor:
		$$f'(\mathcal{R})\mathcal{R}-2f(\mathcal{R})=\kappa^2 T$$
		(What is worth mentioning now is that if we choose $f(\mathcal{R})\propto\mathcal{R}^2$ then $f'(\mathcal{R})=2\mathcal{R}$ and the left hand side of the above equation vanishes identically, meaning that also $T$ must be equal to zero. Then, only conformally invariant matter can be described by such theory \cite{sotir}). Having this relation, we can in principle solve it and look for roots giving us a direct relation between $\mathcal{R}$ and $T$: $\mathcal{R}=\mathcal{R}(T)$ and $f(\mathcal{R}(T))\equiv f(T)$, and the second field equation can be written as:
		\begin{equation}
		\nabla^\Gamma_\alpha\Big(\sqrt{-g}f'(T)g^{\mu\nu}\Big)=0
		\end{equation}
		Let us now introduce a second metric on the spacetime defined to be $\bar{g}_{\mu\nu}=f'(T(x))g_{\mu\nu}\equiv\Omega^2(x)g_{\mu\nu}$. This metric is obviously conformally related to the metric $g_{\mu\nu}$. We can also show that:
		$$\sqrt{-\text{det}(\bar{g}_{\alpha\beta})}\bar{g}^{\mu\nu}=\sqrt{-\text{det}(g_{\alpha\beta})}f'(T)g^{\mu\nu}$$
		And hence, the second field equations reads now as follows:
		\begin{equation}
		\nabla^\Gamma_\alpha\Big(\sqrt{-\bar{g}}\bar{g}^{\mu\nu}\Big)=0
		\end{equation}
		By a well-known theorem \cite{sok}, the connection used to defined the covariant derivative must be a Levi-Civita connection of the new metric $\bar{g}_{\mu\nu}$ (this will be also proved in the next chapter, step by step, for a scalar field coupled to the metric tensor):
		$$\bar{\Gamma}^\alpha_{\mu\nu}=\frac{1}{2}\bar{g}^{\alpha\beta}\big(\partial_\mu\bar{g}_{\nu\beta}-\partial_\nu\bar{g}_{\mu\beta}-\partial_\beta\bar{g}_{\mu\nu}\big)$$
		Also, $\mathcal{R}_{\mu\nu}=\mathcal{R}(g)_{\mu\nu}$. It means that we can eliminate the independent connection from the equations and treat it as a auxiliary field. Moreover, making use of the conformal transformation relation, we can now unravel the dependence of the curvature scalar $\mathcal{R}$ on the curvature scalar $R$ calculated using the metric tensor $g_{\mu\nu}$ only \cite{sotir}:
		\begin{equation}
		\mathcal{R}=R+\frac{3}{2(f'(T))^2}\nabla_\mu f'(T) \nabla^\mu f'(T)+\frac{3}{f'(T)}\Box f'(T)
		\end{equation}
		In a similar manner we can give an explicit relation between the Ricci tensor defined with respect to the independent connection and the Ricci tensor being a function of the metric tensor (actually, it would have been more appropriate to give this relation first, since the one shown above is its mere result; we do not write this lengthy formula here since it is only an intermediate step). Having this, we can basically plug these relation back in the first field equation and successfully reduce the number of field equations to one equation (albeit a very complex one):
		\begin{equation}
		\begin{split}
		G_{\mu\nu}(g)=&\frac{\kappa^2}{f'(T)} T_{\mu\nu}-\frac{1}{2}g_{\mu\nu}\Big(\mathcal{R}(T)-\frac{f(T)}{f'(T)}\Big)+\frac{1}{f'(T)}\big(\nabla_\mu\nabla_\nu-g_{\mu\nu}\Box\big)f'(T)-\\
		&+\frac{3}{2(f'(T))^2}\Big(\nabla_\mu f'(T)\nabla_\nu f'(T)-\frac{1}{2}g_{\mu\nu}\nabla_\alpha f'(T)\nabla^\alpha f'(T)\Big)
		\end{split} \label{difff}
		\end{equation}
		The theory turned out to be equivalent to GR with a modified source including derivatives of $T$, which did not occur in GR. This of course has some serious implications. For example, the vacuum for this theory is the same as in case of GR with the cosmological constant added \cite{bo}. 
		
		An important remark must be made here. Palatini formalism in case of $f(R)$ theories implies a bimetric structure of the theory. Such metric are conformally related to each other and one of them appears in matter part of the action functional (as the matter is assumed to be decoupled from the independent connection), and the other one builds geometric structures such as Riemann tensor. According to some authors \cite{popl}, \cite{alle}, roles played by these two metrics can be clearly divided: the metric building the geometric objects (being hence Levi-Civita-compatible with the independent connection, which is a flagship of the Palatini formalism) is the one conformally related the metric coupled the matter, used to measure distances. In fact, their physical meaning should be exactly opposite (at least in so-called Einstein frame). As we saw in the section dedicated to the Brans-Dicke theory, after a conformal change matter is coupled to the scalar field, thus violating WEP; the same happens in case of Palatini $f(R)$, where a vicarious role of the scalar field is being played by $f'(R)$, since $\bar{g}_{\mu\nu}=f'(T)g_{\mu\nu}$, and the action can be written in the following form:
		$$S[g,\bar{g}]=\frac{1}{2\kappa^2}\int_{\Omega}d^4x\sqrt{-\bar{g}}(f'(\bar{R}))^{-2}f(\bar{R})+S_{\text{matter}}[g,\chi]$$
		Particles follow geodesics defined by the metric $g_{\mu\nu}$, and an alleged deviation from their standard trajectories are observed only if we insist on treating the metric $\hat{g}_{\mu\nu}$ as the metric defining parallel transport. Hence, it seems more accurate to reverse interpretation of both metric tensors: the initial one is responsible for motion of the particles (as it was shown in the section dedicated to BD theories), and the new one is used to measure distances \cite{koivisto}.
		
		Surprisingly, despite the fact that $f(R)$ theory in Palatini approach might seem particularly appealing, it is in conflict with the Standard Model \cite{sotir}. Additionally, $f(R)$ theories in Palatini formalism exhibit a singular behavior when analysing the stellar structure, giving rise to infinite tidal forces on the surface. This basically means that the theory is at best incomplete \cite{nogo}. This stems from the fact that the equation \ref{difff} can be up to the third order in matter derivatives, whereas in case of GR only the first derivatives are present in the trace $T$. Metric tensor is an integral over all matter sources and any possible discontinuities of the latter (and their derivatives) will not translate to singularities/discontinuities of the metric \footnote{Of course, there is an ongoing debate regarding viability of Palatini $f(R)$ theories. Some authors address the issue of infinite tidal forces on the surface resulting from the fact that the conformal factor of the 'original' metric $g_{\mu\nu}$ is continuous but not $\mathcal{C}^1$ at the surface claiming that this discontinuity can be removed when one works with the conformal metric, $\bar{g}_{\mu\nu}$\cite{go}}. Despite these very serious shortcomings, one should keep in mind that the theory is thought of as a toy theory, and its main objective is to understand general theory better.
	\chapter{Scalar-tensor theories of gravity in Palatini approach}
		\section{Conformal transformations in Palatini formalism}
			If we follow the process of deriving conformal transformation formulae for the Riemann tensor, it becomes obvious that the underlying assumption we make is that the connection used to build up the tensor is the Levi-Civita connection, e.g. it is related to the metric tensor by a well-known formula: $\Big\{{\alpha\atop{\mu\nu}}\Big\}=\frac{1}{2}g^{\alpha\beta}(g_{\beta\mu,\nu}+g_{\beta\nu,\mu}-g_{\mu\nu,\beta})$. In order for a connection $\Gamma^\alpha_{\mu\nu}$ to be related to the metric tensor by the formula given above, the following conditions must be fulfilled: $\Gamma^\alpha_{\mu\nu}=\Gamma^\alpha_{\nu\mu}$ and $\nabla g_{\mu\nu}=0$. We can, however, relax the imposed constraints and start  considering the affine connection as a quantity entirely independent of the metric tensor. In this case, transformation relations will change since the covariant derivative of the metric tensor will not vanish in general. Hence, we postulate that the metric tensor and the affine connection transform under a conformal change independently of one another:

			\begin{equation}g_{\mu\nu}=e^{2\bar{\gamma_1}(\bar{\Phi}(x))}\bar{g}_{\mu\nu} \label{T1}\end{equation}
			and
			 \begin{equation}\Gamma^\alpha_{\mu\nu}=\bar{\Gamma}^\alpha_{\mu\nu}+\delta^\alpha_\mu\bar{\nabla}_\nu\:\bar{\gamma}_2(\bar{\Phi})+\delta^\alpha_\nu\bar{\nabla}_\mu\:\bar{\gamma}_2(\bar{\Phi})-\epsilon\bar{g}_{\mu\nu}\bar{g}^{\alpha\beta}\bar{\nabla}_\beta\bar{\gamma}_2(\bar{\Phi})\label{T2}\end{equation}
			
			All quantities depend on a spacetime position; however, the dependence of the functions $\bar{\gamma_1}, \bar{\gamma_2}$ is indirect, they depend on the position through a scalar field. The coefficient $\epsilon$ is introduced in order to gain the possibility of 'switching off' the last term. This may turn out very handy if we want to investigate the behaviour of the Ricci scalar under a change of conformal frame. It is important to note that when $\epsilon=0$, the formula is a so-called geodesic mapping, meaning that it preserves geodesics \cite{ker},\cite{diff}. More on the topic of geodesic maps can be found in Appendix A.\footnote{Another topic which remains beyond the scope of this work is implementation of Ehlers-Pirani-Schild (EPS) framework for gravity, in which metric $g$ and connection $\Gamma$ are treated as two independent objects, and where equations of motion enforce a compatibility condition relating these two objects: $\Gamma^\alpha_{\mu\nu}=\Big\{{\alpha\atop{\mu\nu}}\Big\}-\frac{1}{2}\Big(g^{\alpha\beta}-2\delta^\alpha_{(\mu}\delta^\beta_{\nu)}\Big)\nabla_\beta\text{ln}\phi\equiv\Gamma^\alpha_{\mu\nu}(\phi g)$ (or, to be put differently, $\nabla_\alpha g_{\mu\nu}=2A_\alpha g_{\mu\nu}$ for some covector $A_\mu$). This result is obtained from a couple of very basic postulates, based on the way particles and light rays move in spacetime; see \cite{lor},\cite{fatibene},\cite{eps}. The outcome of the postulates is a Lorentzian metric - or, to be more precise, a whole family of Lorentzian metric related by a pointwise conformal transformation. Since the upshot of the procedure is not a unique metric, but a family of conformally related metrics, one cannot observe a representative of this gauge. Physical observables should not depend upon a representative, but they should depend on the conformal structure as a whole. EPS formalism is of some interest since $f(R)$ theories of gravity in Palatini formalism are a case of integrable Extended Theories of Gravity, where ETG framework is implemented.}
			\\ The conformal transformation is accompanied by a reparametrization of the scalar field: 
			\begin{equation}
			\Phi=\bar{f}(\bar{\Phi}) \label{T3}
			\end{equation}
			If the calculations are performend in $n$ dimensions, the formulae relating Riemann tensors of two different conformal frames are the following:
			\begin{equation}
			\begin{split}
			R^\alpha_{\mu\beta\nu}& =\bar{R}^\alpha_{\mu\beta\nu}+\delta^\alpha_\nu\bar{\nabla}_\beta\bar{\nabla}_\mu\bar{\gamma_2}(\bar{\Phi})-\delta^\alpha_\beta\bar{\nabla}_\nu\bar{\nabla}_\mu\bar{\gamma_2}(\bar{\Phi})-\delta^\alpha_\nu\bar{\nabla}_\beta\bar{\gamma_2}(\bar{\Phi})\bar{\nabla}_\mu\bar{\gamma_2}(\bar{\Phi})+\delta^\alpha_\beta\bar{\nabla}_\nu\bar{\gamma_2}(\bar{\Phi})\bar{\nabla}_\mu\bar{\gamma_2}(\bar{\Phi}) +\\
			& \epsilon\Big[\bar{g}_{\mu\beta}\bar{g}^{\alpha\lambda}\bar{\nabla}_\nu\bar{\nabla}_\lambda\bar{\gamma_2}(\bar{\Phi})-\bar{g}_{\mu\nu}\bar{g}^{\alpha\lambda}\bar{\nabla}_\beta\bar{\nabla}_\lambda\bar{\gamma_2}(\bar{\Phi})+\delta^\alpha_\nu\bar{g}_{\mu\beta}\bar{g}^{\sigma\lambda}\bar{\nabla}_\sigma\bar{\gamma_2}(\bar{\Phi})\bar{\nabla}_\lambda\bar{\gamma_2}(\bar{\Phi})-\\
			&\delta^\alpha_\beta\bar{g}_{\mu\nu}\bar{g}^{\sigma\lambda}\bar{\nabla}_\sigma\bar{\gamma_2}(\bar{\Phi})\bar{\nabla}_\lambda\bar{\gamma_2}(\bar{\Phi}) +  \epsilon\Big(\bar{g}^{\alpha\lambda}\bar{g}_{\mu\nu}\bar{\nabla}_\lambda\bar{\gamma_2}(\bar{\Phi})\bar{\nabla}_\beta\bar{\gamma_2}(\bar{\Phi})-\bar{g}^{\alpha\lambda}\bar{g}_{\mu\beta}\bar{\nabla}_\lambda\bar{\gamma_2}(\bar{\Phi})\bar{\nabla}_\nu\bar{\gamma_2}(\bar{\Phi})\Big)\Big] +\\
			& \epsilon\Big[\bar{g}^{\alpha\lambda}\bar{\nabla}_\nu\bar{g}_{\mu\beta}\bar{\nabla}_\lambda\bar{\gamma_2}(\bar{\Phi})-\bar{g}^{\alpha\lambda}\bar{\nabla}_\beta\bar{g}_{\mu\nu	}\bar{\nabla}_\lambda\bar{\gamma_2}(\bar{\Phi})+\bar{g}_{\mu\beta}\bar{\nabla}_\nu\bar{g}^{\alpha\lambda}\bar{\nabla}_\lambda\bar{\gamma_2}(\bar{\Phi})-\bar{g}_{\mu\nu}\bar{\nabla}_\beta\bar{g}^{\alpha\lambda}\bar{\nabla}_\lambda\bar{\gamma_2}(\bar{\Phi})\Big]
			\end{split}
			\end{equation}
			
			The formula for the (symmetrized) Ricci curvature tensor reads as follows:
			
			\begin{equation}
			\begin{split}
			\hat{R}_{(\mu\nu)} &=\hat{\bar{R}}_{\mu\nu}-(n-1-\epsilon)\bar{\nabla}_{\mu}\bar{\nabla}_\nu\bar{\gamma_2}(\bar{\Phi})+(n-1-\epsilon^2)\bar{\nabla}_\nu\bar{\gamma_2}(\bar{\Phi})\bar{\nabla}_\mu\bar{\gamma_2}(\bar{\Phi})-\epsilon\bar{g}_{\mu\nu}\bar{g}^{\alpha\beta}\bar{\nabla}_{\alpha}\bar{\nabla}_\beta\bar{\gamma_2}(\bar{\Phi})-\\
			&\epsilon(n-1-\epsilon)\bar{g}_{\mu\nu}\bar{g}^{\alpha\beta}\bar{\nabla}_\alpha\bar{\gamma_2}(\bar{\Phi}){\nabla}_\beta\bar{\gamma_2}(\bar{\Phi})  +\epsilon\Big[\bar{g}_{\mu\nu}\bar{g}^{\alpha\beta}\bar{g}^{\sigma\lambda}\bar{\nabla}_\alpha\bar{g}_{\beta\sigma}-\bar{g}^{\alpha\lambda}\bar{\nabla}_\alpha\bar{g}_{\mu\nu}\Big]\bar{\nabla}_\lambda\bar{\gamma_2}(\bar{\Phi})
			\end{split}
			\end{equation}
			
			And, finally, contracting the previous formula with the metric tensor, we get an expression for the Ricci scalar:
			
			\begin{equation}
			\begin{split}
			\hat{R}& =e^{-2\bar{\gamma}_1(\bar{\Phi})}\Big[\hat{\bar{R}}-(n-1-\epsilon+n\epsilon)\bar{g}^{\mu\nu}\bar{\nabla}_\mu\bar{\nabla}_\nu\bar{\gamma}_2(\bar{\Phi})+\epsilon\bar{g}^{\mu\nu}g^{\lambda\sigma}\Big(n\bar{\nabla}_\mu\bar{g}_{\nu\sigma}-\bar{\nabla}_\sigma\bar{g}_{\nu\mu}\Big)\bar{\nabla}_\lambda\bar{\gamma}_2(\bar{\Phi}) + \\
			& (n-1-\epsilon^2-\epsilon n^2+\epsilon n+\epsilon^2 n)\bar{g}^{\mu\nu}\bar{\nabla}_\mu\bar{\gamma}_2(\bar{\Phi})\bar{\nabla}_\nu\bar{\gamma}_2(\bar{\Phi})\Big] \label{curvsc}
			\end{split}
			\end{equation}
			
			Now, since the function $\bar{\gamma}_2$ does not depend on spacetime position explicitly, the derivative of this quantity can be cast in the following form: 
			 $$\bar{\nabla}_\mu\bar{\gamma}_2(\bar{\Phi})=\frac{d \bar{\gamma}_2(\bar{\Phi})}{d\bar{\Phi}}\bar{\nabla}_\mu\bar{\Phi}\equiv\bar{\gamma}'_2\bar{\nabla}_\mu\bar{\Phi}$$
			 If we plug this into the expression for the Ricci scalar and assume $\epsilon=1$, $n=4$, we get the full transformation formula:
			\begin{equation}
			\hat{R}=e^{-2\bar{\gamma}_1(\bar{\Phi})}\Big[\hat{\bar{R}}-6\bar{g}^{\mu\nu}\Big(\bar{\gamma}''_2(\bar{\Phi})+(\bar{\gamma}'_2(\bar{\Phi}))^2\Big)\bar{\nabla}_\mu\bar{\Phi}\bar{\nabla}_{\nu}\bar{\Phi}-6\bar{g}^{\mu\nu}\bar{\gamma}'_2(\bar{\Phi})\bar{\nabla}_\mu\bar{\nabla}_\nu\bar{\Phi}+$$$$+\bar{\gamma}'_2(\bar{\Phi})\bar{g}^{\mu\nu}\bar{g}^{\tau\sigma}(4\bar{Q}_{\sigma\tau\mu}-\bar{Q}_{\mu\sigma\tau})\bar{\nabla}_{\nu}\bar{\Phi}\Big]
			\end{equation}
			where $\bar{Q}_{\sigma\tau\mu}\equiv \bar{\nabla}_\sigma\bar{g}_{\tau\mu}$. As we can see, up to a certain point the formula exactly matches calculations performed in the purely metric approach. However, the last term containing derivatives of the metric tensor is a novelty. The tensor $\bar{Q}_{\sigma\tau\mu}$ indicates whether the theory is metric; if it vanishes, then the connection is Levi-Civita with respect to the metric tensor. If the converse is true, connection and metric tensor are two independent variables. 
			
		\section{Action functional and equations of motion}
			
			We postulate the following action functional:
			
			\begin{equation}
			\begin{split}
			S[\Phi,g,\Gamma]=&\frac{1}{2\kappa^2}\int_{\Omega}d^4x\sqrt{-g}\Big[\mathcal{A}(\Phi)\hat{R}-\mathcal{B}(\Phi)g^{\mu\nu}\nabla_\mu\Phi\nabla_\nu\Phi-\mathcal{C}(\Phi)A^\mu\nabla_\mu\Phi-\mathcal{V}(\Phi)\Big]\\
			&+ S_{\text{matter}}[e^{2\alpha(\Phi)}g,\chi]
			\end{split} \label{actt}
			\end{equation}
			
			This action contains five arbitrary functions: $\{\mathcal{A}(\Phi),\mathcal{B}(\Phi),\mathcal{C}(\Phi),\mathcal{V}(\Phi),\alpha(\Phi)\}$ depending on the scalar field $\Phi$. The function $\mathcal{A}(\Phi)$ describes coupling between the scalar field and the scalar curvature, $\mathcal{B}(\Phi)$ is the kinetic coupling, $\mathcal{V}(\Phi)$ is a self-interacting potential of the scalar field, $\alpha(\Phi)$ is a coupling of the scalar field to the matter part of the action. The coefficient $\mathcal{C}(\Phi)$ does not have a clear interpretation yet; it multiplies the term linear in spacetime derivatives of the scalar field. As it will be shown alter on, if the coefficient $\mathcal{C}(\Phi)$ vanishes, then it is always possible to find a metric $\hat{g}$ conformally related to the metric $g$ such that the connection is Levi-Civita with respect to that metric. Also, the vector $A^\mu$ is defined to be: $A^\mu= g^{\mu\nu}g^{\alpha\beta}(Q_{\nu\alpha\beta}-Q_{\beta\alpha\nu})$; this definition stems from the transformation properties of the curvature scalar. It must be added in order to keep the form of action functional unchanged under a conformal change. 
			
			Variation of the action with respect to the metric tensor yields the first set of field equations:
			
			\begin{equation}
			\begin{split}
			&-\frac{1}{2}g_{\mu\nu}\mathcal{L}+\mathcal{A}(\Phi)\hat{R}_{\mu\nu}-\mathcal{B}(\Phi)\nabla_\mu\Phi\nabla_\nu\Phi-\mathcal{C}(\Phi)\nabla_\alpha\Phi\Big[\frac{1}{2}g^{\sigma\tau}\delta^{(\alpha}_\mu\delta^{\beta)}_{\nu}Q_{\beta\sigma\tau}+ \frac{1}{2}g_{\mu\nu}g^{\sigma\tau}g^{\alpha\beta}Q_{\beta\sigma\tau}\\
			&-g_{\mu\nu}g^{\sigma\alpha}g^{\tau\beta}Q_{\beta\sigma\tau}\Big]+\Big[\mathcal{C}'(\Phi)\nabla_\alpha\Phi\nabla_\beta\Phi+\mathcal{C}(\Phi)\nabla_\alpha\nabla_\beta\Phi\Big]\Big(\delta^{(\alpha}_\mu\delta^{\beta)}_{\nu}-g^{\alpha\beta}g_{\mu\nu}\Big)=\kappa^2T_{\mu\nu}
			\end{split} \label{eqp1}
			\end{equation}
			
			Variation with respect to the affine connection gives us the following equations:
			\begin{equation}
			\nabla_\tau\Big[\mathcal{A}(\Phi)\sqrt{-g}\Big(\delta^\tau_\lambda\delta^\sigma_\nu g^{\mu\nu}-g^{\nu(\mu}\delta^{\sigma)}_\lambda\delta^\tau_\nu\Big)\Big]-\sqrt{-g}\mathcal{C}(\Phi)\nabla_\nu\Big(g^{\sigma\mu}\delta^\nu_\lambda-g^{\nu(\mu}\delta^{\sigma)}_\lambda\Big)=0 \label{eqp2}
			\end{equation}
			
			These equations are somehow easier to analyse than the previous ones. As we can see, if the coefficient $\mathcal{C}(\Phi)$ vanishes, we end up with a somewhat simplified and, presumably, more familiar equations:
			\begin{equation}
			\nabla_\tau\Big[\mathcal{A}(\Phi)\sqrt{-g}\Big(\delta^\tau_\lambda\delta^\sigma_\nu g^{\mu\nu}-g^{\nu(\mu}\delta^{\sigma)}_\lambda\delta^\tau_\nu\Big)\Big]=0 \label{simpl}
			\end{equation}
			If we set now $\mu=\lambda$, we get:
			$$\nabla_\tau\Big[\mathcal{A}(\Phi)\sqrt{-g}\Big(\delta^\tau_\lambda\delta^\sigma_\nu g^{\lambda\nu}-g^{\nu(\lambda}\delta^{\sigma)}_\lambda\delta^\tau_\nu\Big)\Big]=\nabla_\tau\Big[\mathcal{A}(\Phi)\sqrt{-g}\Big( g^{\tau\sigma}-\frac{1}{2}g^{\tau\sigma}-2g^{\tau\sigma}\Big)\Big]=$$
			$$=-\frac{3}{2}\nabla_\tau\Big[\mathcal{A}(\Phi)\sqrt{-g} g^{\tau\sigma}\Big]=0$$
			This means that the second term in \ref{simpl} does not give any contribution (since $\nabla_\alpha \delta^\mu_\nu=0$, which can be easily verified), so we are left with:
			$$\nabla_\lambda\Big[\mathcal{A}(\Phi)\sqrt{-g} g^{\mu\sigma}\Big]=0$$
			If we now define a new metric, conformally related to the metric $g$: $$\hat{g}_{\mu\nu}=\mathcal{A}(\Phi)g_{\mu\nu} \rightarrow \hat{g}\equiv\text{det}(\hat{g}_{\mu\nu})=\mathcal{A}^4(\Phi)\text{det}(g_{\mu\nu})\equiv\mathcal{A}^4(\Phi)\:g$$ from which it follows that $$\sqrt{\hat{g}}\hat{g}^{\mu\nu}=\mathcal{A}(\Phi)\sqrt{g}g^{\mu\nu}$$ we see that 
			\begin{equation}
			\nabla_\lambda\Big[\sqrt{\hat{g}}\hat{g}^{\mu\sigma}\Big]=0 \label{ident}
			\end{equation} Also, since the covariant derivative of the Kronecker symbol vanishes, we can write:
			$$0=\nabla_\alpha \delta^\mu_\nu=\nabla_\alpha\Big(\hat{g}^{\mu\sigma}\hat{g}_{\sigma\nu}\Big)=\hat{g}^{\mu\sigma}\nabla_\alpha\hat{g}_{\sigma\nu}+\hat{g}_{\sigma\nu}\nabla_\alpha\hat{g}^{\mu\sigma}$$ which means that
			$$\nabla_\alpha \hat{g}^{\mu\nu}=-\hat{g}^{\sigma\mu}\hat{g}^{\lambda\nu}\nabla_\alpha\hat{g}_{\sigma\lambda}$$ Splitting the product in \ref{ident} using the Leibniz rule, we write now
			\begin{equation}
			\begin{split}
			&\nabla_\lambda\Big[\sqrt{\hat{g}}\hat{g}^{\mu\sigma}\Big]=\sqrt{\hat{g}}\nabla_\lambda\hat{g}^{\mu\sigma}+\hat{g}^{\mu\sigma}\nabla_\lambda\sqrt{\hat{g}}=\sqrt{\hat{g}}\nabla_\lambda\hat{g}^{\mu\sigma}+\hat{g}^{\mu\sigma}\big(\partial_\lambda\sqrt{\hat{g}}-\sqrt{\hat{g}}\Gamma^\tau_{\mu\tau}\big)=\\
			&=\sqrt{\hat{g}}\nabla_\lambda\hat{g}^{\mu\sigma}+\hat{g}^{\mu\sigma}\big(\frac{1}{2}\sqrt{\hat{g}}\hat{g}^{\kappa\tau}\partial_\lambda\hat{g}_{\kappa\tau}-\sqrt{\hat{g}}\Gamma^\tau_{\mu\tau}\big)=\sqrt{\hat{g}}\Big[\nabla_\lambda\hat{g}^{\mu\sigma}+\hat{g}^{\mu\sigma}\hat{g}^{\kappa\tau}\big(\frac{1}{2}\hat{g}^{\kappa\tau}\partial_\lambda\hat{g}_{\kappa\tau}-\\
			&+\frac{1}{2}\Gamma^\xi_{\mu\tau}\hat{g}_{\kappa\xi}-\frac{1}{2}\Gamma^\xi_{\mu\kappa}\hat{g}_{\xi\tau}\big)\Big]=\sqrt{\hat{g}}\Big[\nabla_\lambda\hat{g}^{\mu\sigma}+\frac{1}{2}\hat{g}^{\mu\sigma}\hat{g}^{\kappa\tau}\nabla_\lambda\hat{g}_{\kappa\tau}\Big]=\sqrt{\hat{g}}\Big[-\hat{g}^{\kappa\mu}\hat{g}^{\tau\sigma}\nabla_\lambda\hat{g}_{\kappa\tau}+\\
			&+\frac{1}{2}\hat{g}^{\mu\sigma}\hat{g}^{\kappa\tau}\nabla_\lambda\hat{g}_{\kappa\tau}\Big]=\sqrt{\hat{g}}\Big[-\hat{g}^{\kappa\mu}\hat{g}^{\tau\sigma}+\frac{1}{2}\hat{g}^{\mu\sigma}\hat{g}^{\kappa\tau}\Big]\nabla_\lambda\hat{g}_{\kappa\tau}=0
			\end{split}
			\end{equation}
			where the formula for derivative of metric determinant has been used. It follows now that: 
			\begin{equation}
			\Big[-\hat{g}^{\kappa\mu}\hat{g}^{\tau\sigma}+\frac{1}{2}\hat{g}^{\mu\sigma}\hat{g}^{\kappa\tau}\Big]\nabla_\lambda\hat{g}_{\kappa\tau}=0
			\end{equation}
			We can contract this equation with $\hat{g}_{\mu\sigma}$ and get:
			\begin{equation}
			\hat{g}^{\kappa\tau}\nabla_\lambda\hat{g}_{\kappa\tau}=0 \label{ident2}
			\end{equation}
			Having obtained this important result, let us consider the following identity:
			\begin{equation}
			4\nabla_\lambda\sqrt{\hat{g}}=\nabla_\lambda\big(\sqrt{\hat{g}}\hat{g}^{\kappa\tau}\hat{g}_{\kappa\tau}\big)=\sqrt{\hat{g}}\hat{g}^{\kappa\tau}\nabla_\lambda\hat{g}_{\kappa\tau}+\hat{g}_{\kappa\tau}\nabla_\lambda\big(\sqrt{\hat{g}}\hat{g}^{\kappa\tau}\big)=0
			\end{equation}
			by the virtue of \ref{ident} and \ref{ident2}. It means that:
			\begin{equation}
			\nabla_\lambda\Big[\sqrt{\hat{g}}\hat{g}^{\mu\sigma}\Big]=\sqrt{\hat{g}}\nabla_\lambda\hat{g}^{\mu\sigma}=0\rightarrow\nabla_\lambda\hat{g}^{\mu\sigma}=0\quad\text{and}\quad\nabla_\lambda\hat{g}_{\mu\sigma}=0
			\end{equation}
			If we express the last covariant derivative explicitly and manipulate the indices, we can write:
				$$0=\nabla_\lambda \hat{g}_{\mu\sigma}=\partial_\lambda \hat{g}_{\mu\sigma}-\Gamma^\rho_{\:\lambda\mu}\hat{g}_{\rho\sigma}-\Gamma^\rho_{\:\lambda\sigma}\hat{g}_{\mu\rho}$$
				
				$$0\stackrel{\mu\leftrightarrow\lambda}{=}\nabla_\mu \hat{g}_{\lambda\sigma}=\partial_\mu \hat{g}_{\lambda\sigma}-\Gamma^\rho_{\:\lambda\mu}\hat{g}_{\rho\sigma}-\Gamma^\rho_{\:\mu\sigma} \hat{g}_{\lambda\rho}$$
				
				$$0\stackrel{\sigma\leftrightarrow\lambda}{=}\nabla_\sigma \hat{g}_{\mu\lambda}=\partial_\sigma \hat{g}_{\mu\lambda}-\Gamma^\rho_{\:\sigma\mu}\hat{g}_{\rho\lambda}-\Gamma^\rho_{\:\lambda\sigma} \hat{g}_{\mu\rho}$$
			Adding the third equation to the second and subtracting the first one, we get the desired result:
			\begin{equation}
			2\Gamma^\rho_{\:\mu\sigma} \hat{g}_{\lambda\rho}=\partial_\mu \hat{g}_{\lambda\sigma}+\partial_\sigma \hat{g}_{\mu\lambda}-\partial_\lambda \hat{g}_{\mu\sigma}
			\end{equation}
			or, written in another way:
			\begin{equation}
			\Gamma^\rho_{\:\mu\sigma}=\frac{1}{2}\hat{g}^{\rho\lambda}\Big(\partial_\mu \hat{g}_{\lambda\sigma}+\partial_\sigma \hat{g}_{\mu\lambda}-\partial_\lambda \hat{g}_{\mu\sigma}\Big)
			\end{equation}
			which means that the connection is Levi-Civita with respect to the new metric and, hence, also geometric quantities describing curvature of spacetime, such as Riemann tensor, become functions of this metric.\\
				 The last equation of motion results from varying the action with respect to the scalar field:
			
			\begin{equation}
			\begin{split}
			\mathcal{A}'(\Phi)\hat{R}&+\mathcal{B}(\Phi)\Box\Phi+\frac{1}{\sqrt{-g}}\mathcal{B}(\Phi)\nabla_\mu\Phi\nabla_\nu\Big(\sqrt{-g}g^{\mu\nu}\Big)+\frac{1}{\sqrt{-g}}\mathcal{C}(\Phi)\nabla_{\nu}\Big(\sqrt{-g}A^\nu\Big)-\mathcal{V}'(\Phi)\\
			&=2\kappa^2\alpha'(\Phi)T
			\end{split} \label{eqp3}
			\end{equation}
			
		\section{Transformation relations}
			Following the logic of the previous chapter, we want now to apply Weyl transformation to metric tensor, given by the formula \ref{T1}, as well as to reparametrize the scalar field according to \ref{T3} and change the affine connection using \ref{T2}. If we do so, form of the action \ref{actt} should be preserved; what will substantially change are the five functions of the scalar field: $\{\mathcal{A}(\Phi),\mathcal{B}(\Phi),\mathcal{C}(\Phi),\mathcal{V}(\Phi),\alpha(\Phi)\}$. After we switch the conformal frame, we will end up with a set of new functions $\{\mathcal{\bar{A}}(\bar{\Phi}),\mathcal{\bar{B}}(\bar{\Phi}),\mathcal{\bar{C}}(\bar{\Phi}),\mathcal{\bar{V}}(\bar{\Phi}),\bar{\alpha}(\bar{\Phi})\}$ depending on a new scalar field $\bar{\Phi}$; moreover, the functional form of these function will not remain unchanged in general. \\
			To put it differently, we are looking now for transformations of the five functions induced by the transformations \ref{T1},\ref{T2},\ref{T3} leaving the action functional invariant (up to boundary terms):
			
			\begin{equation}
			\begin{split}
			S[g,\Gamma,\Phi]=S[\bar{g},\bar{\Gamma},\bar{\Phi}]+\text{(boundary terms)}
			\end{split}
			\end{equation}
			
			Using the postulated form of the action functional, we can write explicitly:
			\begin{equation}
			\begin{split}
			&\frac{1}{2\kappa^2}\int_{\Omega}d^4x\sqrt{-g}\Big[\mathcal{A}(\Phi)\hat{R}-\mathcal{B}(\Phi)g^{\mu\nu}\nabla_\mu\Phi\nabla_\nu\Phi-\mathcal{C}(\Phi)A^\mu\nabla_\mu\Phi-\mathcal{V}(\Phi)\Big]+ S_{\text{matter}}[e^{2\alpha(\Phi)}g,\chi] =\\
			& =\frac{1}{2\kappa^2}\int_{\Omega}d^4x\sqrt{-\bar{g}}\Big[\mathcal{\bar{A}}(\bar{\Phi})\hat{\bar{R}}-\mathcal{\bar{B}}(\bar{\Phi})\bar{g}^{\mu\nu}\bar{\nabla}_\mu\bar{\Phi}\bar{\nabla}_\nu\bar{\Phi}-\mathcal{\bar{C}}(\bar{\Phi})\bar{A}^\mu\bar{\nabla}_\mu\bar{\Phi}-\mathcal{\bar{V}}(\bar{\Phi})\Big]+ S_{\text{matter}}[e^{2\bar{\alpha}(\bar{\Phi})}\bar{g},\chi] \\
			&+\text{(boundary terms)}
			\end{split}
			\end{equation}
			
			The formula above holds if the 'old' functions are related to the 'new' ones by the followings transformation equations:
			\begin{itemize}
				\item $\mathcal{\bar{A}}(\bar{\Phi})=e^{2\bar{\gamma}_1(\bar{\Phi})}\mathcal{A}(\bar{f}(\bar{\Phi}))$
				\item $\mathcal{\bar{B}}(\bar{\Phi})=e^{2\bar{\gamma}_1(\bar{\Phi})}\Big[-12\mathcal{A}(\bar{f}(\bar{\Phi}))\bar{\gamma}'_1(\bar{\Phi})\bar{\gamma}'_2(\bar{\Phi})+6\mathcal{A}(\bar{f}(\bar{\Phi}))(\bar{\gamma}'_2(\bar{\Phi}))^2-6\mathcal{A}'(\bar{f}(\bar{\Phi}))\bar{\gamma}'_2(\bar{\Phi})+\mathcal{B}(\bar{f}(\bar{\Phi}))(\bar{f}'(\bar{\Phi}))^2+6\mathcal{C}(\bar{f}(\bar{\Phi}))\bar{f}'(\bar{\Phi})(\bar{\gamma}'_1(\bar{\Phi})-\bar{\gamma}'_2(\bar{\Phi}))\Big]$
				\item $\mathcal{\bar{C}}(\bar{\Phi})=e^{2\bar{\gamma}_1(\bar{\Phi})}\mathcal{C}(\bar{f}(\bar{\Phi}))\bar{f}'(\bar{\Phi})-2e^{2\bar{\gamma}_1(\bar{\Phi})}\bar{\gamma}'_2(\bar{\Phi})\mathcal{A}(\bar{f}(\bar{\Phi}))$
				\item $\mathcal{\bar{V}}(\bar{\Phi})=e^{4\bar{\gamma}_1(\bar{\Phi})}\mathcal{V}(\bar{f}(\bar{\Phi}))$
				\item $\bar{\alpha}(\bar{\Phi})=\alpha(\bar{f}(\bar{\Phi}))+\bar{\gamma}_1(\bar{\Phi})$
			\end{itemize}
			
			This was calculated using the fact that the vector $A^\mu(g,\Gamma)$, as long as $\bar{\gamma}_1\neq\bar{\gamma}_2$, transforms in a nontrivial way:
			\begin{equation}
			A^{\mu}(g,\Gamma)\rightarrow e^{-2\bar{\gamma}_1(\bar{\Phi})}\bar{A}^{\mu}(\bar{g},\bar{\Gamma})+6e^{-2\bar{\gamma}_1(\bar{\Phi})}\bar{f}'(\bar{\Phi})\bar{g}^{\mu\nu}(\bar{\gamma}'_1(\bar{\Phi})-\bar{\gamma}'_2(\bar{\Phi}))\bar{\nabla}_\nu\bar{\Phi}
			\end{equation}
			where $\bar{A}^{\mu}(\bar{g},\bar{\Gamma})=\bar{g}^{\mu\nu}\bar{g}^{\alpha\beta}(\bar{Q}_{\nu\alpha\beta}-\bar{Q}_{\beta\alpha\nu})$
			\\
			
			Following the reasoning carried out in \cite{kuinv}, analyzing the structure of the relations shown above we can deduce certain properties of the coefficients. First, if $\mathcal{A}(\Phi)$ is positive in any conformal frame (this coefficient cannot be equal to zero and, in order to make gravity an attractive force, it must be positive), then it is greater than zero in any frame related by means of a conformal transformation. The same property holds for the potential term $\mathcal{V}(\Phi)$; however, if it vanishes in one frame, then it is equal to zero in any other frame. 
			
			 By a proper choice of three function $\{\gamma_1,\gamma_2,f\}$ we are able to fix three of the five arbitrary functions $\{\mathcal{A},\mathcal{B},\mathcal{C},\mathcal{V},\alpha\}$; we shall call this 'fixing a conformal frame'. After we have fixed three functions, we still have freedom to specify the remaining two functions. By doing so, we choose a specific theory. For example, the three functions $\{\gamma_1,\gamma_2,f\}$ can be chosen in such a way that three coefficients $\{\mathcal{B},\mathcal{C},\alpha\}$ vanish, thus simplifying the calculations. Results obtained in a given frame can be always 'translated' to another frame if the two frames can be related by a conformal transformation accompanied by a reparametrization of the scalar field. 
			
				\subsection{Einstein and Jordan frames}
				In literature, two of all conformal frames are particularly widespread: Einstein and Jordan frames. Action in these two frames can be related by a conformal transformation with the general action \ref{actt}, so that they retain its properties \cite{kuinv}.  So far, however, all scalar-tensor theories of gravity have been analysed either in purely metric approach, or in Palatini approach but they were emerging from a different class of Extended Theories of Gravity, namely, from $f(R)$ theories. This resulted in omitting the coefficient $\mathcal{C}(\Phi)$ in all considerations, which now had to be added for self-consistence of the general theory. As a result, coefficients specifying a particular frame did not include this additional function of the scalar field. Since the role played by the coefficient $\mathcal{C}(\Phi)$ has not been yet understood, it is left arbitrary when we speak of 'fixing the frame'; it can be, however, set equal to zero if a proper choice of the function $\gamma_2$ is made.\\
				
				A distinction between these two frames in terms of measurements made in each of them was neatly explained by E. E. Flanagan (E. E. Flanagan 2004). His example of how choice of definition of the units we use to measure distances and time affects choice of a particular conformal frame gives us a taste of physical meaning of the discernment between both frames, and the example is worth quoting it here: 
			
				\vspace{2pt}
			
				\noindent\textit{Suppose we define units of length and time by taking the speed of light to be unity and by taking the unit of time to be determined by some atomic transition frequency (as in the current SI definition of the second). Measurements of the geometry of spacetime in these units yield the Jordan-frame metric. However, we can instead define a system of units as follows. Suppose that we have a nonspinning black hole. We can in principle take this to be a “standard” black hole (like the original platinum-iridium standard meter), and create other nonspinning black holes of the same size. Using these black holes we can operationally define a unit of time to be the inverse of the frequency of their fundamental quasinormal mode of vibration. If we define the
				speed of light to be unity, and measure the geometry of spacetime in these units, the result is the Einstein-frame	metric.}\footnote{E. E. Flanagan, \textit{The conformal frame freedom in theories of gravitation}, Class.Quant.Grav.21:3817 (2004), page 4.}
				
				\vspace{2pt}
				
				\noindent In the \textbf{Einstein frame}, the gravitational part of the action functional contains only Einstein gravity \cite{val}, but scalar field is present in the matter part of the Lagrangian, displaying an anomalous coupling. The scalar field, in other words, becomes a form of matter and is always present. Canonically, one fixes $\mathcal{A}=1$ and $\mathcal{B}=2$, while keeping the remaining three coefficients arbitrary functions of the new scalar field $\phi$. Without loss of generality, we may also consider the case when $\mathcal{C}=0$ as belonging to the Einstein frame. If this is the case, then, as it has been already shown, another metric, building the geometrical objects such as curvature scalar and being conformally related to the one used to measure distances can be introduced. Particles follow now geodesics determined by the 'new' metric, so it seems that there exists an additional 'fifth force' acting on particles, causing them to deviate from the trajectories determined by the 'old' metric. As a result, in the Einstein frame, the Principle of Equivalence can be violated \cite{val}. \\
				
				\noindent In the \textbf{Jordan frame}, the gravitational field is described by metric, connection and scalar field. Usually one assumes $\mathcal{A}=\Psi$ and $\alpha=0$, with the three functions $\{\mathcal{B},\mathcal{C},\mathcal{V}\}$ left arbitrary functions of $\Psi$, which is now assumed to be a new scalar field. In the Jordan frame, the scalar field is nonminimally coupled to the curvature but, unlike in the case of Einstein frame, it is absent from the matter part of the action. If the coefficient $\mathcal{C}=0$, then the metric used for measurement is also used to build geometrical objects. Freely falling particles move along geodesic of the corresponding geometry \cite{kuinv}. Also, the Principle of Equivalence is not violated in the Jordan frame.\\
				
				\noindent It is probably an uncontroversial and widely accepted statement that the Einstein and Jordan frames are mathematically equivalent, but they are very different under a physical point of view. This discernment raises an important question of which conformal frame is the physical one, meaning that it is self-consistent and it is possible to predict values of certain observables (which can be measured) working in such frame. According to \cite{val} and other authors, the Jordan frame is unphysical \textit{because it leads to negative definite, or indefinite kinetic energy for the scalar field; on the contrary, the energy density is positive definite in the Einstein frame}\footnote{V. Faraoni, E. Gunzig, P. Nardone, \textit{Conformal transformations in classical gravitational theories and in cosmology}, Fund.Cosmic Phys.20:121 (1999), page 15.}. As we can read in another paper by Faraoni (Faraoni, Gunzig 1999): 
				
				\vspace{2pt}
				
				\noindent\textit{The Jordan frame formulation of a scalar–tensor theory is not viable because the energy density of the gravitational scalar field present in the theory is not bounded from below (violation of the weak energy condition). The system therefore is	unstable and decays toward a lower and lower energy state ad infinitum} \footnote{V. Faraoni, E. Gunzig, \textit{Einstein frame or Jordan frame ?},Int. J. Theor. Phys. 38, 217, page 3.}
				
				\vspace{2pt}
				
				\noindent Cho (Cho 2003) puts it even more dramatically: 
				
				\vspace{2pt}
				
				\noindent\textit{When the quantum correction takes place ordinary matter must couple to the Brans–Dicke scalar field through the Jordan metric. So the quantum fluctuation (in particular, the mixing between the Jordan metric and the Brans–Dicke scalar field) inevitably induces a direct coupling of the Brans–Dicke scalar field to ordinary matter. This direct coupling, however, is precisely what Brans and Dicke have tried to avoid to ensure the weak equivalence principle} \footnote{Y. M. Cho, \textit{Quantum violation of the equivalence principle in Brans - Dicke theory}, Class. Quantum Grav. 14 2963 (1997), page 3.}
				
				\vspace{2pt}
				
				\noindent On the other hand, the Einstein frame is free of this problem, but exhibits - as it was said before - violation of the Principle of Equivalence. However, there are two possible objections to the arguments given above: first, as it is pointed out by Faraoni, the Einstein frame is physical for scalar-tensor theories of gravity only when the matter part is not considered; second, the discussion concerns only the purely metric approach, not metric-affine. What remains a fact is that the two frames, Einstein and Jordan, despite their mathematical equivalence, give different observational predictions.
				
				\noindent To make the discussion of the equivalence between Einstein and Jordan frames complete, we should also include arguments supporting the converse assertion: that in fact different frames are not only mathematically, but also physically equivalent. This position is taken by Flanagan, who clearly states that:
				
				\vspace{2pt}
				
				\noindent\textit{efforts to determine the “correct” choice of conformal frame are misguided, at least in the realm of classical physics. They are analogous to attempting to determine the “correct” choice of radial coordinate in the Schwarzschild spacetime. In that context, there is of course no correct radial coordinate, since all physical observables are coordinate invariants. In a similar way, all observable quantities in scalar-tensor theories are conformal-frame invariants} \footnote{E. E. Flanagan, \textit{The conformal frame freedom in theories of gravitation}, Class.Quant.Grav.21:3817 (2004), page 4.}
				
				\vspace{2pt}
				
				\noindent Indeed, an idea that all physical observables should be expressed in a frame-invariant way sounds very appealing, and this is what was originally intended by P. Kuusk and L. Jarv (P. Kuusk, L. Jarv \textit{et al.} 2014). In \cite{kuinv} it is shown that observables like post-Newtonian parameters are expressed in terms of invariant quantities (see section 'Invariants' of this paper), which do not depend on a choice of conformal frame.

		\section{Group structure of the coefficients $\{\mathcal{A},\mathcal{B},\mathcal{C},\mathcal{V},\alpha\}$}
			
			Having obtained the transformation formulae given above, we must check whether the coefficients transform in a correct way when we make a composition of two conformal transformations. If in some conformal frame - let us call it $\mathfrak{A}$ - we make use of the following variables: $\{g_{\mu\nu},\Gamma^{\alpha}_{\mu\nu},\Phi\}$, we are able to change the frame using three arbitrary functions $\{\bar{\gamma}_1,\bar{\gamma}_2,\bar{f}\}$, obtaining new independent variables in a frame $\mathfrak{B}$ related to the 'old' ones via:
			\begin{enumerate}
				\item $g_{\mu\nu}=e^{2\bar{\gamma}_1(\bar{\Phi})}\bar{g}_{\mu\nu}$
				\item $\Gamma^\alpha_{\mu\nu}=\bar{\Gamma}^\alpha_{\mu\nu}+\delta^\alpha_\mu\bar{\nabla}_\nu\:\bar{\gamma}_2(\bar{\Phi})+\delta^\alpha_\nu\bar{\nabla}_\mu\:\bar{\gamma}_2(\bar{\Phi})-\bar{g}_{\mu\nu}\bar{g}^{\alpha\beta}\bar{\nabla}_\beta\bar{\gamma}_2(\bar{\Phi})$
				\item $\Phi=\bar{f}(\bar{\Phi})$
			\end{enumerate}
			Using new set of functions $\{\bar{\bar{\gamma}}_1,\bar{\bar{\gamma}}_2,\bar{\bar{f}}\}$ we can perform the transformation once again and arrive at a frame $\mathfrak{C}$ with variables $\{\bar{\bar{g}}_{\mu\nu},\bar{\bar{\Gamma}}^{\alpha}_{\mu\nu},\bar{\bar{\Phi}}\}$. The question now is: if the frame $\mathfrak{C}$ is to be related to the initial frame $\mathfrak{A}$ by a single transformation making use of three functions $\{\gamma_1,\gamma_2,f\}$, then what is the correspondence between them and the functions $\{\bar{\gamma}_1,\bar{\gamma}_2,\bar{f}\}$ and $\{\bar{\bar{\gamma}}_1,\bar{\bar{\gamma}}_2,\bar{\bar{f}}\}$? If we investigate the way independent variables transform when we change the conformal frame, it will turn out that these functions should be related to each other in the following way:
			\begin{enumerate}
				\item $\Phi(\bar{\bar{\Phi}})=(\bar{f}\circ \bar{\bar{f}})(\bar{\bar{\Phi}})=g(\bar{\bar{\Phi}})$
				\item $\gamma_1(\bar{\bar{\Phi}})=\bar{\bar{\gamma}}_1(\bar{\bar{\Phi}})+\bar{\gamma}_1(\bar{\bar{f}}(\bar{\bar{\Phi}}))$
				\item $\gamma_2(\bar{\bar{\Phi}})=\bar{\bar{\gamma}}_2(\bar{\bar{\Phi}})+\bar{\gamma}_2(\bar{\bar{f}}(\bar{\bar{\Phi}}))$
			\end{enumerate}
			The coefficients $\{\mathcal{A},\mathcal{B},\mathcal{C},\mathcal{V},\alpha\}$ should transform accordingly, preserving the structure of formulae relating two conformal frames. Making use of the symbols defined above (the following has to be stressed here: symbols relating frames $\mathfrak{A}$ to $\mathfrak{C}$ are valid only in this part of the work, elsewhere they might have a different meaning), after some extremely tedious calculations, we can write:
			\begin{itemize}
				\item $\mathcal{\bar{\bar{A}}}(\bar{\bar{\Phi}})=e^{2\gamma_1(\bar{\bar{\Phi}})}\mathcal{A}(g(\bar{\bar{\Phi}}))$
				\item $\mathcal{\bar{\bar{B}}}(\bar{\bar{\Phi}})=e^{2\gamma_1(\bar{\bar{\Phi}})}\Big[-12\mathcal{A}(g(\bar{\bar{\Phi}}))\gamma'_1(\bar{\bar{\Phi}})\gamma'_2(\bar{\bar{\Phi}})+6\mathcal{A}(g(\bar{\bar{\Phi}}))(\gamma'_2(\bar{\bar{\Phi}}))^2-6\frac{d\mathcal{A}(g(\bar{\bar{\Phi}}))}{d\bar{\bar{\Phi}}}\gamma'_2(\bar{\bar{\Phi}})+\mathcal{B}(g(\bar{\bar{\Phi}}))\Big(\frac{dg(\bar{\bar{\Phi}})}{d\bar{\bar{\Phi}}}\Big)^2+6\mathcal{C}(g(\bar{\bar{\Phi}}))\frac{dg(\bar{\bar{\Phi}})}{d\bar{\bar{\Phi}}}\Big(\gamma'_1(\bar{\bar{\Phi}})-\gamma'_2(\bar{\bar{\Phi}})\Big)\Big]$
				\item $\mathcal{\bar{\bar{C}}}(\bar{\bar{\Phi}})=e^{2\gamma_1(\bar{\bar{\Phi}})}\mathcal{C}(g(\bar{\bar{\Phi}}))\frac{dg(\bar{\bar{\Phi}})}{d\bar{\bar{\Phi}}}-2e^{2\gamma_1(\bar{\bar{\Phi}})}\gamma'_2(\bar{\bar{\Phi}})\mathcal{A}(g(\bar{\bar{\Phi}}))$
				\item $\mathcal{\bar{\bar{V}}}(\bar{\bar{\Phi}})=e^{4\gamma_1(\bar{\bar{\Phi}})}\mathcal{V}(g(\bar{\bar{\Phi}}))$
				\item $\bar{\bar{\alpha}}(\bar{\bar{\Phi}})=\alpha(g(\bar{\bar{\Phi}}))+\gamma_1(\bar{\bar{\Phi}})$
			\end{itemize}
		\section{Invariants}
			
			Due to the way coefficients $\{\mathcal{A},\mathcal{B},\mathcal{C},\mathcal{V},\alpha\}$ transform, it is possible to construct - analogously to the procedure carried out in \cite{kuinv} - several quantities which remain invariant under a transformation of metric and connection, together with a reparametrization of scalar field. Such quantities are invariant in a sense that they preserve their form and are expressed by the same formula in every conformal frame. Also, their value at a given spacetime point stays the same as we move from one frame to another. Furthermore, since the conformal transformation is in principle independent of coordinate transformation, spacetime derivatives of invariant quantities are invariants themselves. 
			
			\newpage
			
			The invariants relevant to the following parts of this work are listed below:
		
			\begin{equation}
				\mathcal{I}_1(\Phi)=\frac{\mathcal{A}(\Phi)}{e^{2\alpha(\Phi)}} \label{i1}
				\end{equation}
			\begin{equation}
			\mathcal{I}_2(\Phi)=\frac{\mathcal{V}(\Phi)}{\mathcal{A}^2(\Phi)} \label{i2}
			\end{equation}
			\begin{equation}
			\mathcal{I}_3(\Phi)=\int^\Phi_{\Phi'_0}\sqrt{\frac{\frac{3}{2}\mathcal{C}^2(\Phi')-\mathcal{A}(\Phi')\mathcal{B}(\Phi')+3\mathcal{A}'(\Phi')\mathcal{C}(\Phi')}{\mathcal{A}^2(\Phi')}}\:d\Phi'\equiv\int^\Phi_{\Phi'_0}\mathcal{F}(\Phi')\:d\Phi' \label{i3}
			\end{equation}
			
			Let us now discuss the meaning of these invariants. Invariant $\mathcal{I}_1$ measures the coupling between scalar field and matter. It is easy to see that if this invariant is constant, then the field is minimally coupled \cite{kuinv}. Invariant $\mathcal{I}_2$ generalizes notion of the self-interacting potential $\mathcal{V}$. If it is equal to zero, then it must be vanishing in all frames related by a conformal transformation. Invariant $\mathcal{I}_3$ resembles a function measuring some kind of invariant distance in the space of scalar field; accordingly, the function $\mathcal{F}^2$ must play a role of a metric on this space, and clearly, since the space is one-dimensional, $\mathcal{F}$ is its determinant, transforming like a scalar density: $\bar{\mathcal{F}}(\bar{\Phi})=\mathcal{F}(\Phi)\frac{d\Phi}{d\bar{\Phi}}$. As we can see, constant values of $\mathcal{I}_3$ are possible only when $\mathcal{C}=0$ and $\mathcal{B}=0$, which means that the scalar field is not dynamical \cite{jarv}. 
			
			We can of course introduce further invariants by making a simple observation that an arbitrary function of invariant(s) is also invariant. From the transformation properties of derivatives of invariants with respect to scalar field:
			$$\bar{\mathcal{I}}'_i(\bar{\Phi})=\frac{d\bar{\mathcal{I}}'_i}{d\bar{\Phi}}=\frac{d\mathcal{I}_i}{d\Phi}\frac{d\Phi}{d\bar{\Phi}}=\mathcal{I}'_i(\Phi)\frac{d\Phi}{d\bar{\Phi}}$$
			it follows that a quotient $\frac{\mathcal{I}'_i}{\mathcal{I}'_j}\equiv\frac{d\mathcal{I}_i}{d\mathcal{I}_j}$ is invariant too.
			
			It is important to assume that it is possible to express scalar field as a function of any of the invariants. In the Einstein frame, it will be useful to express the scalar field in terms of $\mathcal{I}_3$, whereas in the Jordan frame, where usually the curvature scalar is multiplied by a scalar field, expressing $\Phi$ in terms of $\mathcal{I}_1$ is preferable. Needless to say, finding an inverse of any of the relations defining invariants may be very problematic, and we have to assume that we can express $\Phi(\mathcal{I}_i)$ as a Taylor expansion:
			$$\Phi(\mathcal{I}_i)=\sum_{n=0} \frac{1}{n!}\frac{d^n\Phi}{d\mathcal{I}^n_i}\Big|_{\mathcal{I}_i=\mathcal{I}_i|_{\Phi_0}}(\mathcal{I}_i-\mathcal{I}_i|_{\Phi_0})^n$$
			For example, in case of invariant $\mathcal{I}_3$, we have:
			$$\frac{d\mathcal{I}_3}{d\Phi}=\mathcal{F} \rightarrow \frac{d}{d\mathcal{I}_3}=\frac{1}{\mathcal{F}}\frac{d}{d\Phi}$$
			Since $\mathcal{I}_3$ contains an arbitrary constant, it can be chosen in such a way that $\mathcal{I}_3|_{\Phi_0}=0$, so that the expression for $\Phi$ reads as follows \cite{kuinv}:
			$$\Phi(\mathcal{I}_3)=\Phi_0+\frac{1}{\mathcal{F}}\mathcal{I}_3+\frac{1}{2}\frac{1}{\mathcal{F}}\frac{d}{d\Phi}\Big(\frac{1}{\mathcal{F}}\Big)\mathcal{I}^2_3+\ldots$$
			
			An intuitive meaning of the invariants is that they label theories which are mathematically equivalent. Indeed, if we switch the conformal frame, values of the invariants will stay constant, meaning that we can use them to label corresponding theories, remaining on the same orbit defined by the formulae \ref{T1}, \ref{T2}, \ref{T3}. If we take values of the invariants evaluated in two different conformal frames (or even for two different fixed theories) and it turns out they are different, then such frames cannot be related by a confromal change accompanied by a redefinition of the scalar field.

			Moreover, it is also possible to construct an invariant metric and an invariant connection. In case of the metric there is no unique way of doing so, but in this paper only two possibilities will be considered:
			\begin{equation}
			\hat{g}_{\mu\nu}=\mathcal{A}(\Phi)g_{\mu\nu} \label{g1}
			\end{equation}
			or
			\begin{equation}
			\tilde{g}_{\mu\nu}=e^{2\alpha(\Phi)}g_{\mu\nu} \label{g2}
			\end{equation}
			As for the affine connection, a single recipe for making it invariant has been found:
			\begin{equation}
			\hat{\Gamma}^\alpha_{\mu\nu}=\Gamma^\alpha_{\mu\nu}-\frac{\mathcal{C}(\Phi)}{2\mathcal{A}(\Phi)\mathcal{F}(\Phi)}\Big(\delta^\alpha_\mu\nabla_\nu\mathcal{I}_3(\Phi)+\delta^\alpha_\nu\nabla_\mu\mathcal{I}_3(\Phi)-g_{\mu\nu}g^{\alpha\beta}\nabla_\beta\mathcal{I}_3(\Phi)\Big) \label{con}
			\end{equation}
			
			Having introduced either of the invariant metrics, we can now measure distances in a frame-independent way. It means that any two observers belonging to two different conformal frames will end up with the same values of a given observable. For example, the distance between our planet and the Sun will be exactly the same if we use one of the invariant metrics to measure it. However, if observers of two different frames insist on using a normalized definition of a unit of length (for example, determined with respect to the speed of light set equal to one and atomic clocks), the outcomes will be different \cite{flanagan}. The same holds for a definition of parallel transport. Since we have obtained an invariant connection, determining geodesics in a unique way will become possible (again, under the assumption that all observers use the invariant connection).
			
				\subsection{Action in terms of the invariant metric $\hat{g}_{\mu\nu}$}
					\label{subsec:inv1}
					Having introduced the invariants, we may now attempt to write down the action functional \ref{actt} fully in terms of them. This approach will give us an obvious advantage, since no matter which frame we are working in, all equations will be of the same form; the action functional will be unaffected by change of the conformal frame and hence, the resulting equations will be written in terms of the invariants which are expressed by the same relations between the coefficients $\{\mathcal{A}(\Phi),\mathcal{B}(\Phi),\mathcal{C}(\Phi),\mathcal{V}(\Phi),\alpha(\Phi)\}$. \\
					If we substitute the metric $\hat{g}_{\mu\nu}$ and the connection $\hat{\Gamma}^\alpha_{\mu\nu}$ into the action \ref{actt}, and consider the scalar field $\Phi$ a function of the invariant $\mathcal{I}_3$ inverting the relation \ref{i3}, we get:
					\begin{equation}
					S[\hat{g},\hat{\Gamma},\mathcal{I}_3]=\frac{1}{2\kappa^2}\int_{\Omega}d^4x\sqrt{-\hat{g}}\big[\hat{R}(\hat{g},\hat{\Gamma})-\hat{g}^{\mu\nu}\hat{\nabla}_\mu\mathcal{I}_3\hat{\nabla}_\nu\mathcal{I}_3-\mathcal{I}_2\big]+S_\text{matter}\Big(\frac{1}{\mathcal{I}_1}\hat{g},\chi\Big) \label{invact1}
					\end{equation}
					As we can see, this action depends now on three new dynamical variables. Also, the action functional is now cast in an Einstein-like conformal frame. The scalar field is fully decoupled from the curvature, but it enters the matter part of the action, meaning that it still permeates the spacetime and acts as an additional source of gravitational interaction. To see this more clearly, let us perform variation with respect to the variables $\{\hat{g}_{\mu\nu},\hat{\Gamma}^\alpha_{\mu\nu},\mathcal{I}_3\}$:\\
					\begin{enumerate}
						\item $\delta\hat{g}$: $\hat{G}_{\mu\nu}-\hat{\nabla}_\alpha\mathcal{I}_3\hat{\nabla}_\beta\mathcal{I}_3\Big(\delta^\alpha_\mu\delta^\beta_\nu-\frac{1}{2}\hat{g}^{\alpha\beta}\hat{g}_{\mu\nu}\Big)+\frac{1}{2}\hat{g}_{\mu\nu}\mathcal{I}_2=\kappa^2\hat{T}_{\mu\nu}$
						\item $\delta\hat{\Gamma}$: $\hat{\nabla}_\lambda\big(\sqrt{-\hat{g}}\:\hat{g}^{\mu\nu}\big)=0$
						\item $\delta\mathcal{I}_3$: $2\hat{\Box}\mathcal{I}_3-\frac{d\mathcal{I}_2}{d\mathcal{I}_3}=-\kappa^2\frac{1}{\mathcal{I}_1}\frac{d\mathcal{I}_1}{d\mathcal{I}_3}\hat{T}$
					\end{enumerate}
					
					Let us now carefully analyse the obtained equations. If we consider the second equation, we immediately recognize the discussed relation between connection and metric tensor: if a connection is symmetric and the covariant derivative of the metric multiplied by its determinant vanishes, then the connection is necessarily Levi-Civita with respect to the metric. This shows an amazing result: after writing the action functionals in terms of invariants, initially independent invariant connection becomes Levi-Civita with respect to the invariant metric $\hat{g}_{\mu\nu}$. Consequently, the curvature scalar also depends on the metric. Apart from the presence of scalar field in the matter part of the action functional, this suggests that the Einstein-like frame is supposedly the simplest.
					
					Let us switch our attention to the first equation. We want to see whether the conservation of energy-momentum tensor is satisfied. In order to so, we need to calculate covariant derivative of the whole formula and contract one of the (upper) indices with the index of the covariant derivative:
					\begin{equation}
					\hat{\nabla}_\mu\hat{G}^{\mu\nu}-\hat{\nabla}_\mu\big(\hat{\nabla}_\alpha\mathcal{I}_3\hat{\nabla}_\beta\mathcal{I}_3\big)\Big(\hat{g}^{\alpha\mu}\hat{g}^{\beta\nu}-\frac{1}{2}\hat{g}^{\alpha\beta}\hat{g}^{\mu\nu}\Big)+\frac{1}{2}\frac{d\mathcal{I}_2}{d\mathcal{I}_3}\hat{g}^{\mu\nu}\hat{\nabla}_{\mu}\mathcal{I}_3=\kappa^2\hat{\nabla}_\mu\hat{T}^{\mu\nu} \label{cons1}
					\end{equation}
					Divergence of the Einstein tensor vanishes by the virtue of a well-known theorem. Also, because of the fact that connection is Levi-Civita, covariant derivative of the metric tensor equals zero. Furthermore, the second term can be greatly simplified:
					\begin{equation}
					\begin{split}
					& \hat{\nabla}_\mu\big(\hat{\nabla}_\alpha\mathcal{I}_3\hat{\nabla}_\beta\mathcal{I}_3\big)\Big(\hat{g}^{\alpha\mu}\hat{g}^{\beta\nu}-\frac{1}{2}\hat{g}^{\alpha\beta}\hat{g}^{\mu\nu}\Big)=\hat{g}^{\mu\nu}\hat{\nabla}_\mu\mathcal{I}_3\hat{\Box}\mathcal{I}_3+\hat{g}^{\alpha\mu}\hat{\nabla}_\alpha\mathcal{I}_3\:\hat{g}^{\beta\nu}\hat{\nabla}_\mu\hat{\nabla}_\beta\mathcal{I}_3+ \\
					& -\frac{1}{2}\hat{g}^{\alpha\beta}\hat{\nabla}_\beta\mathcal{I}_3\:\hat{g}^{\mu\nu}\hat{\nabla}_\mu\hat{\nabla}_\alpha\mathcal{I}_3-\frac{1}{2}\hat{g}^{\alpha\beta}\hat{\nabla}_\alpha\mathcal{I}_3\:\hat{g}^{\mu\nu}\hat{\nabla}_\mu\hat{\nabla}_\beta\mathcal{I}_3=\hat{g}^{\mu\nu}\hat{\nabla}_\mu\mathcal{I}_3\:\hat{\Box}\mathcal{I}_3+\hat{g}^{\alpha\beta}\hat{\nabla}_\alpha\mathcal{I}_3\:\hat{g}^{\mu\nu}\hat{\nabla}_\beta\hat{\nabla}_\mu\mathcal{I}_3+\\
					& -\frac{1}{2}\hat{g}^{\alpha\beta}\hat{\nabla}_\alpha\mathcal{I}_3\:\hat{g}^{\mu\nu}\hat{\nabla}_\mu\hat{\nabla}_\beta\mathcal{I}_3-\frac{1}{2}\hat{g}^{\alpha\beta}\hat{\nabla}_\alpha\mathcal{I}_3\:\hat{g}^{\mu\nu}\hat{\nabla}_\mu\hat{\nabla}_\beta\mathcal{I}_3=\hat{g}^{\mu\nu}\hat{\nabla}_\mu\mathcal{I}_3\:\hat{\Box}\mathcal{I}_3
					\end{split}
					\end{equation}
					So that \ref{cons1} boils down to:
					\begin{equation}
					-\hat{g}^{\mu\nu}\hat{\nabla}_{\mu}\mathcal{I}_3\big(\hat{\Box}\mathcal{I}_3-\frac{1}{2}\frac{d\mathcal{I}_2}{d\mathcal{I}_3}\big)=\kappa^2\hat{\nabla}_\mu\hat{T}^{\mu\nu} \label{conv2}
					\end{equation}
					Using the third equation of motion, we can express $\hat{\Box}\mathcal{I}_3$ as $$\hat{\Box}\mathcal{I}_3=-\frac{1}{2}\frac{d\mathcal{I}_2}{d\mathcal{I}_3}-\kappa^2\frac{1}{2\mathcal{I}_1}\frac{d\mathcal{I}_1}{d\mathcal{I}_3}\hat{T}$$
					Plugging this in \ref{conv2}, we get:
					\begin{equation}
					\hat{g}^{\mu\nu}\hat{\nabla}_{\mu}\mathcal{I}_3\frac{1}{2\mathcal{I}_1}\frac{d\mathcal{I}_1}{d\mathcal{I}_3}\hat{T}\equiv\frac{1}{2}\hat{\nabla}^\nu(\text{ln}\mathcal{I}_1)\:\hat{T}=\hat{\nabla}_\mu\hat{T}^{\mu\nu} \label{oldt}
					\end{equation}
					So that, clearly, the energy-momentum tensor is not conserved unless $\frac{d\mathcal{I}_1}{d\mathcal{I}_3}=0$. However, we can construct another quantity which is conserved; we simply need to add to the energy-momentum tensor $\hat{T}_{\mu\nu}$ yet another energy-momentum tensor, this time defined for the scalar field as:
					\begin{equation}
					-\kappa^2\hat{T}_\Phi^{\mu\nu}=\frac{1}{2}\frac{\partial \mathcal{L}}{\partial(\hat{\nabla}_\mu\mathcal{I}_3)}\hat{\nabla}^\nu\mathcal{I}_3+\frac{1}{2}\hat{g}^{\mu\nu}\hat{g}_{\alpha\beta}\hat{\nabla}^\alpha\mathcal{I}_3\hat{\nabla}^\beta\mathcal{I}_3+\frac{1}{2}\hat{g}^{\mu\nu}\mathcal{I}_2
					\end{equation}
					Taking the divergence of this tensor, we get:
					\begin{equation}
					\begin{split}
					-\kappa^2\hat{\nabla}_\mu\hat{T}_\Phi^{\mu\nu}&=\hat{\nabla}_\mu\Big(-\hat{\nabla}^\mu\mathcal{I}_3\hat{\nabla}^\nu\mathcal{I}_3+\frac{1}{2}\hat{g}^{\mu\nu}\hat{g}_{\alpha\beta}\hat{\nabla}^\alpha\mathcal{I}_3\hat{\nabla}^\beta\mathcal{I}_3+\frac{1}{2}\hat{g}^{\mu\nu}\mathcal{I}_2\Big)=-\hat{\Box}\mathcal{I}_3\hat{\nabla}^\nu\mathcal{I}_3+\frac{1}{2}\hat{\nabla}^\nu\mathcal{I}_2=\\
					&=-\Big(\hat{\Box}\mathcal{I}_3-\frac{1}{2}\frac{d\mathcal{I}_2}{d\mathcal{I}_3}\Big)\hat{\nabla}^\nu\mathcal{I}_3=\frac{1}{2}\kappa^2\frac{1}{\mathcal{I}_1}\frac{d\mathcal{I}_1}{d\mathcal{I}_3}\hat{T}\hat{\nabla}^\nu\mathcal{I}_3=\frac{\kappa^2}{2}\hat{\nabla}^\nu(\text{ln}\mathcal{I}_1)\:\hat{T} 
				\end{split}\label{newt}
				\end{equation}
					which was calculated using the equation of motion obtained for the scalar field. If we now form a new tensor:
					\begin{equation}
					\mathfrak{T}^{\mu\nu}=\hat{T}^{\mu\nu}+\hat{T}^{\mu\nu}_\Phi
					\end{equation}
					then, by the virtue of \ref{oldt} and \ref{newt}, we have:
					\begin{equation}
					\hat{\nabla}_{\mu}\mathfrak{T}^{\mu\nu}=0
					\end{equation}
					and the (new) energy-momentum tensor is conserved.
				\subsection{Action in terms of the invariant metric $\tilde{g}_{\mu\nu}$}
					Alternatively, we can express the action functional in terms of the invariant metric $\tilde{g}_{\mu\nu}=e^{2\alpha(\Phi)}g_{\mu\nu}$, and the invariant linear connection $\hat{\Gamma}^{\alpha}_{\mu\nu}$. We do not yet specify an invariant whose function the scalar field $\Phi$ should be. This will give us an action functional cast in a Jordan-like frame:
					\begin{equation}
					S[\tilde{g},\hat{\Gamma},\Phi]=\frac{1}{2\kappa^2}\int_{\Omega}d^4x\sqrt{-\tilde{g}}\Big[\mathcal{I}_1\hat{R}(\tilde{g},\hat{\Gamma})-\mathcal{I}_1\mathcal{F}^2\tilde{g}^{\mu\nu}\hat{\nabla}_\mu\Phi\hat{\nabla}_\nu\Phi-\mathcal{I}_2\mathcal{I}^2_1\Big]+S_\text{matter}[\tilde{g},\chi] \label{invm2}
					\end{equation}
					As we can see, this action functional does not contain the coefficient $\mathcal{C}$ as well. Also, the scalar field is not coupled to the matter fields, but there is a nonminimial coupling between the scalar field and the curvature scalar being present. These properties justify calling the frame 'Jordan-like'. In a general Jordan frame, however, the scalar curvature is coupled to the scalar field, which means that we should consider $\Phi$ to be a function of the invariant $\mathcal{I}_1$:
					$$\Phi=\Phi(\mathcal{I}_1)\quad\rightarrow\quad\hat{\nabla}_\alpha\Phi=\frac{d\Phi}{d\mathcal{I}_1}\hat{\nabla}_\alpha\mathcal{I}_1$$
					Substituting this result to the action functional, we get (we focus now only on the kinetic term):
					$$\tilde{g}^{\mu\nu}\mathcal{I}_1\mathcal{F}^2\big(\frac{d\Phi}{d\mathcal{I}_1}\big)^2\hat{\nabla}_\mu\mathcal{I}_1\hat{\nabla}_\nu\mathcal{I}_1=\tilde{g}^{\mu\nu}\mathcal{I}_1\big(\frac{d\mathcal{I}_3}{d\mathcal{I}_1}\big)^2\hat{\nabla}_\mu\mathcal{I}_1\hat{\nabla}_\nu\mathcal{I}_1$$
					since $\mathcal{F}=\frac{d\mathcal{I}_3}{d\Phi}$. The action functional takes the following form:
					\begin{equation}
					S[\tilde{g},\hat{\Gamma},\mathcal{I}_1]=\frac{1}{2\kappa^2}\int_{\Omega}d^4x\sqrt{-\tilde{g}}\Big[\mathcal{I}_1\hat{R}(\tilde{g},\hat{\Gamma})-\tilde{g}^{\mu\nu}\mathcal{I}_1\big(\frac{d\mathcal{I}_3}{d\mathcal{I}_1}\big)^2\hat{\nabla}_\mu\mathcal{I}_1\hat{\nabla}_\nu\mathcal{I}_1-\mathcal{I}_4\Big]+S_\text{matter}[\tilde{g},\chi] \label{invm3}
					\end{equation}
					
					Now, the invariant $\mathcal{I}_1$ plays a role of a dynamical scalar field - analogously to the invariant $\mathcal{I}_3$ in the previous chapter. Here, however, choosing the invariant used previously would not prove useful, since as we will see in the following parts of this paper, in some cases, for a special choice of the conformal frame, the invariant $\mathcal{I}_3$ vanishes, hence rendering it impossible to be used as a function of the scalar field $\Phi$. 
					
					For simplicity, we introduced another invariant, $\mathcal{I}_4$, defined in the following way:
					$$\mathcal{I}_4=\mathcal{I}^2_1\mathcal{I}_2$$
					denoting a modified potential.
					
					Let us now obtain equations of motion for the theory. Variation with respect to all three dynamical variables yields the following formulae:
					\begin{enumerate}
						\item $\delta\tilde{g}$: $\hat{G}_{\mu\nu}(\tilde{g},\hat{\Gamma})-\Big(\frac{d\mathcal{I}_3}{d\mathcal{I}_1}\Big)^2\hat{\nabla}_\alpha\mathcal{I}_1\hat{\nabla}_\beta\mathcal{I}_1\big(\delta^\alpha_\mu\delta^\beta_\nu-\frac{1}{2}\tilde{g}_{\mu\nu}\tilde{g}^{\alpha\beta}\big)+\frac{1}{2}\tilde{g}_{\mu\nu}\frac{\mathcal{I}_4}{\mathcal{I}_1}=\frac{\kappa^2}{\mathcal{I}_1}\tilde{T}_{\mu\nu}$
						\item $\delta\hat{\Gamma}$: $\hat{\nabla}_\alpha\big(\mathcal{I}_1\sqrt{-\tilde{g}}\tilde{g}^{\mu\nu}\big)=0$
						\item $\delta\mathcal{I}_3$: $\hat{R}(\tilde{g},\hat{\Gamma})-\tilde{g}^{\mu\nu}\Bigg[\Big(\frac{d\mathcal{I}_3}{d\mathcal{I}_1}\Big)^2+2\mathcal{I}_1\frac{d\mathcal{I}_3}{d\mathcal{I}_1}\frac{d^2\mathcal{I}_3}{d\mathcal{I}^2_1}\Bigg]+\frac{2}{\sqrt{-\tilde{g}}}\hat{\nabla}_\mu\Big(\sqrt{-\tilde{g}}\tilde{g}^{\mu\nu}\mathcal{I}_1\Big(\frac{d\mathcal{I}_3}{d\mathcal{I}_1}\Big)^2\hat{\nabla}_\nu\mathcal{I}_1\Big)-+\frac{d\mathcal{I}_4}{d\mathcal{I}_1}=0$
					\end{enumerate}
					From now on, let us denote $\frac{d\mathcal{I}_i}{d\mathcal{I}_1}$ by simply $\mathcal{I}'_i$.
					
					These equations need to be carefully analysed. The first one resembles the standard Einstein Field Equations; however, on the right hand side of the equation we have the energy-momentum tensor multiplied by an inverse of the invariant $\mathcal{I}_1$, which clearly plays the role of an effective Newton constant, depending now on the spacetime position. Second equation tells us that the invariant connection used to build the curvature scalar and defining covariant derivative is Levi-Civita with respect to a new metric, $g'_{\mu\nu}=\mathcal{I}_1\tilde{g}_{\mu\nu}$. The third equation gives us a relation between the curvature constant and the scalar field. As we can see, unlike in the case of the Einstein-like frame analysed in the previous subsection, the scalar field is not sourced now by the trace of energy-momentum tensor.
					
					The second equation tells us that the curvature scalar is in fact defined in terms of the metric $g'_{\mu\nu}$, which is conformally related to the invariant metric $\tilde{g}_{\mu\nu}$; hence, we can perform a conformal transformation and write the action in terms of quantities being fully dependent on the invariant metric. In order to achieve this, we can make use of the standard formula relating curvature scalars of two conformal frames; here, however, the transformation is defined by the function $\mathcal{I}_1$, and using the terminology introduced at the beginning of this chapter, it corresponds to the function $\bar{\gamma}_1(\Phi)=\frac{1}{2}\text{ln}\mathcal{I}_1$. Also, we need to identify $\gamma_1=\gamma_2$. It must be stressed that right now we do not change the conformal frame; we merely seek a relation between $\hat{R}(g')$ and $\tilde{R}(\tilde{g})$. Furthermore, the curvature scalar used in \ref{invm3} is a hybrid of both: we still use the tensor $\tilde{g}^{\mu\nu}$ to contract indices, but the Ricci tensor is built fully from the tensor $g'_{\mu\nu}$. That is why it will be convenient to write:
					$$\mathcal{I}_1\tilde{g}^{\mu\nu}\hat{R}_{\mu\nu}(g')=\mathcal{I}^2_1g'^{\mu\nu}\hat{R}_{\mu\nu}(g')=\mathcal{I}^2_1\hat{R}(g')$$
					which makes using the formula \ref{curvsc} possible, reading now as follows:
					\begin{equation}
					\begin{split}
					&\mathcal{I}^2_1\hat{R}(g')=\mathcal{I}^2_1\Bigg[\frac{1}{\mathcal{I}_1}\Big[\tilde{R}(\tilde{g})-6\tilde{g}^{\mu\nu}\tilde{\nabla}_\mu\tilde{\nabla}_\nu\Big(\frac{1}{2}\text{ln}\mathcal{I}_1\Big)-6\tilde{g}^{\mu\nu}\tilde{\nabla}_\mu\Big(\frac{1}{2}\text{ln}\mathcal{I}_1\Big)\tilde{\nabla}_\nu\Big(\frac{1}{2}\text{ln}\mathcal{I}_1\Big)\Big]\Bigg]=\\
					&=\mathcal{I}_1\Big[\tilde{R}(\tilde{g})-6\tilde{g}^{\mu\nu}\tilde{\nabla}_\mu\Big(\frac{1}{2\mathcal{I}_1}\tilde{\nabla}_\nu\mathcal{I}_1\Big)-\frac{3}{2\mathcal{I}^2_1}\tilde{g}^{\mu\nu}\tilde{\nabla}_\mu\mathcal{I}_1\tilde{\nabla}_\nu\mathcal{I}_1\Big]=\\
					&=\mathcal{I}_1\Big[\tilde{R}(\tilde{g})+\frac{3}{2\mathcal{I}^2_1}\tilde{\nabla}_\mu\mathcal{I}_1\tilde{\nabla}_\nu\mathcal{I}_1\Big]-3\tilde{g}^{\mu\nu}\tilde{\nabla}_\mu\tilde{\nabla}_\nu\mathcal{I}_1
					\end{split}
					\end{equation}
					Plugging this back in the action functional \ref{invm3}, we get:
					\begin{equation}
					\begin{split}
					S[\tilde{g},\mathcal{I}_1]=&\frac{1}{2\kappa^2}\int_{\Omega}d^4x\sqrt{-\tilde{g}}\Bigg[\mathcal{I}_1\tilde{R}(\tilde{g})-\tilde{g}^{\mu\nu}\Bigg(\mathcal{I}_1\Big(\frac{d\mathcal{I}_3}{d\mathcal{I}_1}\Big)^2-\frac{3}{2\mathcal{I}_1}\Bigg)\hat{\nabla}_\mu\mathcal{I}_1\hat{\nabla}_\nu\mathcal{I}_1-3\tilde{g}^{\mu\nu}\tilde{\nabla}_\mu\tilde{\nabla}_\nu\mathcal{I}_1-\mathcal{I}_4\Bigg]+\\
					&+S_\text{matter}[\tilde{g},\chi]
					\end{split}
					\end{equation}
					Clearly, after the conformal transformation of the curvature scalar the action functional depends only on two independent variables: the invariant metric and the invariant quantity $\mathcal{I}_1$. Let us focus our attention on the term containing second derivative of the invariant $\mathcal{I}_1$. This term can be viewed as a divergence of a vector density, since:
					$$-3\sqrt{-\tilde{g}}\tilde{g}^{\mu\nu}\tilde{\nabla}_\mu\tilde{\nabla}_\nu\mathcal{I}_1=-3\tilde{\nabla}_\nu\Big(\sqrt{-\tilde{g}}\tilde{g}^{\mu\nu}\tilde{\nabla}_\nu\mathcal{I}_1\Big)$$
					
					This means that the formula for the action functional reduces to:
					\begin{equation}
					\begin{split}
					S[\tilde{g},\mathcal{I}_1]=&\frac{1}{2\kappa^2}\int_{\Omega}d^4x\sqrt{-\tilde{g}}\Bigg[\mathcal{I}_1\tilde{R}(\tilde{g})-\tilde{g}^{\mu\nu}\Bigg(\mathcal{I}_1\Big(\frac{d\mathcal{I}_3}{d\mathcal{I}_1}\Big)^2-\frac{3}{2\mathcal{I}_1}\Bigg)\hat{\nabla}_\mu\mathcal{I}_1\hat{\nabla}_\nu\mathcal{I}_1-\mathcal{I}_4\Bigg]+S_\text{matter}[\tilde{g},\chi] \label{invm4}
					\end{split}
					\end{equation}
					For simplicity, let us introduce another invariant $\mathcal{I}_5$:
					$$\mathcal{I}_5=\mathcal{I}_1\Big(\frac{d\mathcal{I}_3}{d\mathcal{I}_1}\Big)^2-\frac{3}{2\mathcal{I}_1}$$
					We can now perform variation of the action and obtain field equations. Varying with respect to the metric tensor is now more complicated than in case of Palatini formulation, mostly because we encounter a problem with boundary conditions imposed on the variation of metric. Discussion of all subtleties related to this topic is beyond the scope of this paper and can be found in \cite{cast}. Without paying much attention to boundary conditions, neglecting terms containing divergence of a vector density, we may now find field equations having the formula for variation of the curvature tensor:
					\begin{equation}
					\begin{split}
					\delta_{\tilde{g}}\tilde{R}=&\delta\tilde{g}^{\mu\nu}\tilde{R}_{\mu\nu}+\tilde{\nabla}_\alpha\Big(\tilde{g}^{\mu\nu}(\delta\tilde{\Gamma}^\alpha_{\mu\nu})-\tilde{g}^{\mu\alpha}(\delta\tilde{\Gamma}^\nu_{\mu\nu})\Big)=\\
					&=\delta\tilde{g}^{\mu\nu}\tilde{R}_{\mu\nu}+\big(\tilde{g}_{\mu\nu}\tilde{\Box}-\tilde{\nabla}_\mu\tilde{\nabla}_\nu\big)\delta\tilde{g}^{\mu\nu}
					\end{split}
					\end{equation}
					We may now transform the term $\mathcal{I}_1\big(\tilde{g}_{\mu\nu}\tilde{\Box}-\tilde{\nabla}_\mu\tilde{\nabla}_\nu\big)\delta\tilde{g}^{\mu\nu}$ according to:
					$$\mathcal{I}_1\big(\tilde{g}_{\mu\nu}\tilde{\Box}-\tilde{\nabla}_\mu\tilde{\nabla}_\nu\big)\delta\tilde{g}^{\mu\nu}=\text{\textbf{div}}+\delta\tilde{g}^{\mu\nu}\big(\tilde{g}_{\mu\nu}\tilde{\Box}-\tilde{\nabla}_\mu\tilde{\nabla}_\nu\big)\mathcal{I}_1$$
					Having introduced the above formulae, we may now arrive at a set of field equations obtained by varying with respect to the metric tensor:
					\begin{equation}
					\tilde{G}^{\mu\nu}-\frac{\mathcal{I}_5}{\mathcal{I}_1}\tilde{\nabla}_\alpha\mathcal{I}_1\tilde{\nabla}_\beta\mathcal{I}_1\Big(\tilde{g}^{\alpha\mu}\tilde{g}^{\beta\nu}-\frac{1}{2}\tilde{g}^{\nu\mu}\tilde{g}^{\beta\alpha}\Big)+\frac{1}{2\mathcal{I}_1}\tilde{g}^{\mu\nu}\mathcal{I}_4+\frac{1}{\mathcal{I}_1}\big(\tilde{g}^{\mu\nu}\tilde{\Box}-\tilde{\nabla}^\mu\tilde{\nabla}^\nu\big)\mathcal{I}_1=\frac{\kappa^2}{\mathcal{I}_1}\tilde{T}^{\mu\nu} \label{eqn}
					\end{equation}
					Equation governing evolution of the scalar field reads as follows:
					\begin{equation}
					\tilde{R}+\mathcal{I}'_5\tilde{g}^{\mu\nu}\tilde{\nabla}_\mu\mathcal{I}_1\tilde{\nabla}_\nu\mathcal{I}_1+2\mathcal{I}_5\tilde{\Box}\mathcal{I}_1-\mathcal{I}'_4=0 \label{skalar}
					\end{equation}
					In the Jordan-like frame, the energy-momentum tensor should be conserved. Taking the divergence of the equation \ref{eqn}, we have:
					\begin{equation}
					\begin{split}
					&\tilde{G}^{\mu\nu}\tilde{\nabla}_\mu\mathcal{I}_1-\mathcal{I}'_5\tilde{\nabla}_\mu\mathcal{I}_1\tilde{\nabla}_\beta\mathcal{I}_1\tilde{\nabla}_\alpha\mathcal{I}_1\Big(\tilde{g}^{\alpha\mu}\tilde{g}^{\beta\nu}-\frac{1}{2}\tilde{g}^{\nu\mu}\tilde{g}^{\beta\alpha}\Big)-\\
					&+\mathcal{I}_5\Big(\tilde{g}^{\alpha\mu}\tilde{g}^{\beta\nu}-\frac{1}{2}\tilde{g}^{\nu\mu}\tilde{g}^{\beta\alpha}\Big)\Big(\tilde{\nabla}_\mu\tilde{\nabla}_\alpha\mathcal{I}_1\tilde{\nabla}_\beta\mathcal{I}_1+\tilde{\nabla}_\alpha\mathcal{I}_1\tilde{\nabla}_\mu\tilde{\nabla}_\beta\mathcal{I}_1\Big)+\\
					&+\big(\tilde{\nabla}^\nu\tilde{\Box}-\tilde{\Box}\tilde{\nabla}^\nu\big)\mathcal{I}_1+\frac{1}{2}\mathcal{I}'_4\tilde{\nabla}^\nu\mathcal{I}_1=\kappa^2\tilde{\nabla}_\mu\tilde{T}^{\mu\nu} \label{cons}
					\end{split}
					\end{equation}
					Let us now focus our attention on the term $\big(\tilde{\nabla}^\nu\tilde{\Box}-\tilde{\Box}\tilde{\nabla}^\nu\big)\mathcal{I}_1$. We may express in terms of Ricci tensor components (keeping in mind that, for a torsionless connection, we have an identity: $[\nabla_\mu,\nabla_\nu]V^\alpha=R^\alpha_{\mu\beta\nu}V^\beta$, for some vector field $V^\alpha$):
					\begin{equation}
					\begin{split}
					&\big(\tilde{\nabla}^\nu\tilde{\Box}-\tilde{\Box}\tilde{\nabla}^\nu\big)\mathcal{I}_1=\tilde{g}^{\alpha\nu}\tilde{g}^{\sigma\tau}\big(\tilde{\nabla}_\alpha\tilde{\nabla}_\sigma\tilde{\nabla}_\tau-\tilde{\nabla}_\sigma\tilde{\nabla}_\tau\tilde{\nabla}_\alpha\big)\mathcal{I}_1=\\
					&=\tilde{g}^{\alpha\nu}\tilde{g}^{\sigma\tau}\big(\tilde{\nabla}_\alpha\tilde{\nabla}_\sigma\tilde{\nabla}_\tau-\tilde{\nabla}_\sigma\tilde{\nabla}_\alpha\tilde{\nabla}_\tau\big)\mathcal{I}_1=\tilde{g}^{\alpha\nu}\tilde{g}^{\sigma\tau}\big(\tilde{\nabla}_\alpha\tilde{\nabla}_\sigma-\tilde{\nabla}_\sigma\tilde{\nabla}_\alpha\big)\tilde{\nabla}_\tau\mathcal{I}_1=\\
					&=\tilde{g}^{\alpha\nu}[\tilde{\nabla}_\alpha,\tilde{\nabla}_\sigma]\tilde{\nabla}^\sigma\mathcal{I}_3=\tilde{g}^{\alpha\nu}\tilde{R}^\sigma_{\alpha\mu\sigma}\tilde{\nabla}^\mu\mathcal{I}_1=-\tilde{g}^{\alpha\nu}\tilde{R}^\sigma_{\alpha\sigma\mu}\tilde{\nabla}^\mu\mathcal{I}_1=-\tilde{R}^{\nu\mu}\tilde{\nabla}_\mu\mathcal{I}_1
					\end{split}
					\end{equation}
					This term and $\tilde{\nabla}_\mu\mathcal{I}_1\tilde{R}^{\mu\nu}$ coming from $\tilde{\nabla}_\mu\mathcal{I}_1\tilde{G}^{\mu\nu}$ clearly cancel out. The second term in \ref{cons} can be simplified as well:
					$$\mathcal{I}'_5\tilde{\nabla}_\mu\mathcal{I}_1\tilde{\nabla}_\beta\mathcal{I}_1\tilde{\nabla}_\alpha\mathcal{I}_1\Big(\tilde{g}^{\alpha\mu}\tilde{g}^{\beta\nu}-\frac{1}{2}\tilde{g}^{\nu\mu}\tilde{g}^{\beta\alpha}\Big)=\frac{1}{2}\mathcal{I}'_5\tilde{g}^{\alpha\beta}\tilde{\nabla}_\alpha\mathcal{I}_1\tilde{\nabla}_\beta\mathcal{I}_1\tilde{\nabla}^\nu\mathcal{I}_1$$
					The third term reduces to:
					$$\mathcal{I}_5\Big(\tilde{g}^{\alpha\mu}\tilde{g}^{\beta\nu}-\frac{1}{2}\tilde{g}^{\nu\mu}\tilde{g}^{\beta\alpha}\Big)\Big(\tilde{\nabla}_\mu\tilde{\nabla}_\alpha\mathcal{I}_1\tilde{\nabla}_\beta\mathcal{I}_1+\tilde{\nabla}_\alpha\mathcal{I}_1\tilde{\nabla}_\mu\tilde{\nabla}_\beta\mathcal{I}_1\Big)=$$
					$$=\mathcal{I}_5\big(\tilde{\Box}\mathcal{I}_1\tilde{\nabla}^\nu\mathcal{I}_1+\tilde{\nabla}^\mu\mathcal{I}_1\tilde{\nabla}_\mu\tilde{\nabla}^\nu\mathcal{I}_1-\frac{1}{2}\tilde{\nabla}^\mu\mathcal{I}_1\tilde{\nabla}^\nu\tilde{\nabla}_\mu\mathcal{I}_1-\frac{1}{2}\tilde{\nabla}^\mu\mathcal{I}_1\tilde{\nabla}_\nu\tilde{\nabla}^\nu\mathcal{I}_1\big)=\mathcal{I}_5\tilde{\Box}\mathcal{I}_1\tilde{\nabla}^\nu\mathcal{I}_1$$
					Putting the results altogether and pulling out $\tilde{\nabla}^\nu\mathcal{I}_1$ of the parenthesis, we get:
					\begin{equation}
					\begin{split}
					&\tilde{\nabla}^\nu\mathcal{I}_1\Big(-\frac{1}{2}\tilde{R}+\frac{1}{2}\mathcal{I}'_4-\frac{1}{2}\mathcal{I}'_5\tilde{g}^{\alpha\beta}\tilde{\nabla}_\alpha\mathcal{I}_1\tilde{\nabla}_\beta\mathcal{I}_1-\mathcal{I}_5\tilde{\Box}\mathcal{I}_1\Big)=\\
					&=-\frac{1}{2}\tilde{\nabla}^\nu\mathcal{I}_1\Big(\tilde{R}+\mathcal{I}'_5\tilde{g}^{\mu\nu}\tilde{\nabla}_\mu\mathcal{I}_1\tilde{\nabla}_\nu\mathcal{I}_1+2\mathcal{I}_5\tilde{\Box}\mathcal{I}_1-\mathcal{I}'_4\Big)=\kappa^2\tilde{\nabla}_\mu\tilde{T}^{\mu\nu}=0
					\end{split}
					\end{equation}
					by the virtue of \ref{skalar}. Hence, the energy-momentum tensor is conserved.
					\section{Conclusions}
					In this chapter we have introduced certain mathematical tools which will prove useful in the next part of this paper. At the beginning, we postulated an action functional preserving its form under a conformal change. An immediate consequence of our definition of functions relating dynamical variables of two conformal frames were transformation formulae expressing five arbitrary coefficients evaluated in the 'new' conformal frame in terms of coefficients coming from the 'old' frame. These relations allowed us to write out three quantities which remain invariant under conformal change, thus being good candidates for replacing frame-dependent coefficients $\{\mathcal{A},\mathcal{B},\mathcal{C},\mathcal{V},\alpha\}$ in a particular formulation of the theory. Moreover, two invariant metrics and one invariant linear connection were introduced (together with a definition of the scalar field in terms of the invariants) which allowed us to come up with an invariant way of measuring distances, volumes and time, and defining parallel transport that all observers related by a conformal transformation will agree upon. Within the formalism of invariants two distinct action functionals were presented; what was interesting, the coefficient $\mathcal{C}$ disappeared from the action in either case. Another important feature is that both the invariant metric $\hat{g}_{\mu\nu}$ and the invariant metric $\tilde{g}_{\mu\nu}$ result in the same action functionals and, consequently, the same equations of motion as in case of purely metric formulation of the scalar-tensor theories of gravity. This should not come as a surprise since we mentioned in the second chapter that in fact Palatini theory of gravity is a metric theory. If the independent connection entered also the matter part of action, then predictions given by these two approaches would be drastically different. However, we may still wonder why a theory with highly complicated Lagrangian turns out to be a metric theory. The answer is that we simply decided to bring the theory into the metric form by a particular choice of the functions $\{\gamma_1,\gamma_2,f\}$; as it was proven at the beginning of this chapter, sufficient condition for a connection in one frame to be Levi-Civita connection of some metric tensor is vanishing of $\mathcal{C}$. If this coefficient equals zero, then existence of a metric tensor having the property mentioned is guaranteed. Moreover, such metric tensor will be necessarily conformally related to the 'old' metric tensor, which simply means that we can treat quantities entering the action functional as functions of the original metric from the very beginning. 
					
					However, there is one caveat in stating that both approaches give the same result if we choose a particular set of variables. Equations of motion in each frame are written in terms of invariants, and a special role is played by the invariant $\mathcal{I}_3$ reducing to the scalar field itself in the case of Einstein frame in the metric approach. This invariant in case of the Palatini approach is a different function of the coefficients and thus, of the scalar field. In the metric case, it is defined as:
					$$\mathcal{I}^{(M)}_3(\Phi)=\pm\int_{\Phi_0}^{\Phi}d\Phi'\sqrt{\frac{2\mathcal{A}(\Phi')\mathcal{B}(\Phi')+3(\mathcal{A}'(\Phi'))^2}{4\mathcal{A}^2(\Phi')}}$$
					whereas in the Palatini approach it is given by:
					$$\mathcal{I}^{(P)}_3(\Phi)=\pm \int^\Phi_{\Phi_0}d\Phi\sqrt{\frac{\frac{3}{2}\mathcal{C}^2(\Phi')-\mathcal{A}(\Phi')\mathcal{B}(\Phi')+3\mathcal{A}'(\Phi')\mathcal{C}(\Phi')}{\mathcal{A}^2(\Phi')}}'$$
					For example, if $\mathcal{B}=\mathcal{C}=0$ and $\mathcal{A}\neq\text{const}$, then the latter vanishes, but the former is a well-defined function of the scalar field. Conversely, if $\mathcal{C}$ and $\mathcal{A}$ are constant, but $\mathcal{B}$ vanishes, the metric invariant equals zero identically in all conformal frames, but the Palatini invariant is a function of the scalar field. It is also worth noting that $\mathcal{I}^{(M)}_3(\Phi)$ is not an invariant when considered in the Palatini formalism; this, of course, can be viewed as a consequence of decoupling metric transformation properties from those of connection: $\gamma_1\neq\gamma_2$ in general. If we find a frame in which $\Gamma=\Gamma(g)$, then the theory reduces effectively to the metric one, but with the invariants being different functions of the scalar field. As we will see later on, this discrepancy will turn out very important when analysing $f(R)$ theories.
			\chapter{Applications: $f(R)$ theories of gravity and Friedmann equations}
			So far, we have been developing a theory for scalar-tensor gravity in the Palatini approach. We have set up a background for qualitative and quantitative analysis of particular theories possibly modelling various phenomena, which so far have been investigated with use of general relativity, $f(R)$ theories or scalar-tensor theories of gravity. In this chapter we will look at a couple of possible applications of the machinery developed earlier in this paper with a particular emphasis on correspondence between $f(R)$ and scalar-tensor theories (both in metric and in Palatini approach), showing that effectively the former can be analysed by means of the latter, and calculating the Friedmann equations. Naturally, theories discussed here do not purport to be fundamental theories but there exist chances it will explain certain gravitational effects more consistently (e.g. inflation).
			\section{Equivalence between $f(R)$ and scalar-tensor theories of gravity}
					Remarkably, it is possible to show in a simple way that $f(R)$ and scalar-tensor theories are equivalent to each other - at least in certain cases. However, there is no such equivalence between two different formulation of the same theory; dynamics obtained in Palatini formalism differs from the metric approach. Let us first focus on the latter case.
					\subsection{Metric approach}
							Although the main objective of this paper is to analyse scalar-tensor theories in Palatini approach, we need to dwell on the metric case as well in order to work out what the differences between these two formulations are and how our understanding of them should be altered. Let us start with the action for $f(R)$ theories:
							\begin{equation}
							S[g]=\frac{1}{2\kappa^2}\int_{\Omega}d^4x\sqrt{-g}f(R)+S_{\text{matter}}[g,\chi^m] \label{ac2}
							\end{equation}
							We need to introduce now an auxiliary field $\chi$ in a way that is not modifying the dynamics \cite{capo}, \cite{sot}, \cite{sotir}. Let us try out the following action:
							\begin{equation}
							S[g,\chi]=\frac{1}{2\kappa^2}\int_{\Omega}d^4x\sqrt{-g}\Big(f(\chi)+f'(\chi)(R-\chi)\Big)+S_{\text{matter}}[g,\chi^m] \label{ac3}
							\end{equation}
							It can be seen that adding the new field $\chi$ did not change the dynamics in any way. If variation with respect to the new field is performed, we get (assuming $f''(\chi)\neq 0$):
							$$f''(\chi)(\chi-R)=0$$
							which obviously means that $R=\chi$, and having plugged it back in \ref{ac3}, the action \ref{ac2} is restored. Now, let us make the next step and assume that a scalar field might be identified with the derivative of the function $f$ in the following way: $\Phi=f'(\chi)$. Also, this relation must be invertible, so that the auxiliary field $\chi$ can be thought of as a function of $\Phi$: $\chi=\chi(\Phi)$. Then, our action function reads in the following way:
							\begin{equation}
							S[g,\Phi]=\frac{1}{2\kappa^2}\int_{\Omega}d^4x\sqrt{-g}\Big(\Phi R-V(\Phi)\Big)+S_{\text{matter}}[g,\chi^m] \label{ac4}
							\end{equation}
							where $V(\Phi)$ denotes the potential of self-interaction and is defined as $V(\Phi)=\Phi \chi(\Phi)-f(\chi(\Phi))$. This allowed us to cast $f(R)$ theory in a form effectively equivalent to the Brans-Dicke theory with the parameter $\omega=0$. Let us call this particular case \textbf{metric $f(R)$ in Jordan frame} (this name is, of course, far from being entirely adequate and somehow devoid of imaginativeness, but it denotes accurately the class of theories in a given frame and a given approach). We can treat this action functional a bit more rigorously and write out four coefficients $\{\mathcal{A},\mathcal{B},\mathcal{C},
							\alpha\}$ dependent on the scalar field introduced in \ref{a1}: 
							\begin{itemize}
								\item $\mathcal{A}(\Phi)=\Phi$
								\item $\mathcal{B}(\Phi)=0$
								\item $\mathcal{V}(\Phi)=V(\Phi)$
								\item $\alpha(\Phi)=0$
							\end{itemize}
							Equations of motion are easy to write; one has to make use of \ref{b1}, \ref{b2} and add the potential term to the equations; this gives us:
							\begin{itemize}
								\item $G_{\mu\nu}-\frac{1}{\Phi}(\nabla_\mu\nabla_\nu-g_{\mu\nu}\Box)\Phi+\frac{1}{2\Phi}g_{\mu\nu}V=\frac{\kappa^2}{\Phi}T_{\mu\nu}$
								\item $R-\frac{dV}{d\Phi}=0$
							\end{itemize} 
							We can contract the first equation with the metric tensor and obtain:
							$$-R+\frac{3}{\Phi}\Box\Phi+\frac{2}{\Phi}V=\kappa^2 T$$
							Making use of the second equation, w can also write:
							\begin{equation}
							3\Box\Phi=\Phi\frac{dV}{d\Phi}-2V+\kappa^2 T \label{rem}
							\end{equation}
							As we can see, equation \ref{eqsc} has been reproduced with $\omega=0$ and the potential added.
							
							This class of theories can be successfully labelled by the invariants introduced in the second chapter (since we are dealing now with the metric theory). It will be enough to write out three thereof, as the other ones can be constructed from them:
							\begin{itemize}
								\item $\mathcal{I}^{(M)}_1(\Phi)=\frac{\mathcal{A}(\Phi)}{e^{2\alpha(\Phi)}}=\Phi$
								\item $\mathcal{I}^{(M)}_2(\Phi)=\frac{\mathcal{V}(\Phi)}{(\mathcal{A}(\Phi))^2}=\frac{V(\Phi)}{\Phi^2}$
								\item $\mathcal{I}^{(M)}_3(\Phi)=\pm\int_{\Phi_0}^{\Phi}\sqrt{\frac{2\mathcal{A}(\Phi')\mathcal{B}(\Phi')+3(\mathcal{A}'(\Phi'))^2}{4\mathcal{A}^2(\Phi')}}d\Phi'=\pm\sqrt{\frac{3}{4}}\int_{\Phi_0}^{\Phi}\frac{1}{\Phi'}d\Phi'=\pm\sqrt{\frac{3}{4}}\:\text{ln}\Big(\frac{\Phi}{\Phi_0}\Big)$
							\end{itemize}
							Let us notice that the last relation is clearly invertible, so that we can express the scalar field in a frame-independent way (we need to choose the plus sign): 
							$$\Phi=\Phi_0\:e^{\sqrt{\frac{4}{3}}\mathcal{I}^{(M)}_3}$$ 
							
							We can now perform a conformal change and express the action \ref{ac4} in terms of the invariant metric $\hat{g}_{\mu\nu}=\mathcal{A}(\Phi)g_{\mu\nu}=\Phi g_{\mu\nu}$ and the scalar field defined above, being a function of the invariant $\mathcal{I}^{(M)}_3$; this yields:
							\begin{equation}
							S[\hat{g},\mathcal{I}^{(M)}_3]=\frac{1}{2\kappa^2}\int_{\Omega}d^4x\sqrt{-g}\Big( \hat{R}-2\hat{g}^{\mu\nu}\hat{\nabla}_\mu\mathcal{I}^{(M)}_3\hat{\nabla}_\nu\mathcal{I}^{(M)}_3-\mathcal{I}^{(M)}_2\Big)+S_{\text{matter}}[\frac{1}{\Phi_0}e^{-\sqrt{\frac{4}{3}}\mathcal{I}^{(M)}_3}\hat{g},\chi] \label{ac5}
							\end{equation}
							which looks exactly like \ref{ex}. Here, however, equivalence with the Brans-Dicke theory is less obvious, as the original theory was written in the Jordan frame, and the action functional shown above is cast in the Einstein frame. If we want to reproduce the exact result of \ref{ace}, then we need to carry out yet another, trivial conformal transformation (rescaling) and scalar field redefinition:
							$$\hat{g}'_{\mu\nu}=G\:\hat{g}_{\mu\nu}, \qquad \mathcal{I}'^{(M)}_3=\frac{2}{\sqrt{G}}\mathcal{I}^{(M)}_3$$
							and plug this in \ref{ac5} (let us also note that multiplying by a number does not affect invariance of both the metric and the scalar field - represented here by the invariant $\mathcal{I}'^{(M)}_3$):
							\begin{equation}
							\begin{split}
							S[\hat{g}',\mathcal{I}'^{(M)}_3]&=\frac{1}{2\kappa^2}\int_{\Omega}d^4x\sqrt{-\bar{g}'}\Bigg(\frac{\hat{R}'}{G}-\frac{1}{2}\hat{g}'^{\mu\nu}\hat{\nabla}'_\mu\mathcal{I}'^{(M)}_3\hat{\nabla}'_\nu\mathcal{I}'^{(M)}_3-\mathcal{I}'^{(M)}_2\Bigg)+\\
							&+S_{\text{matter}}\Big[\frac{1}{G\Phi_0}e^{-\sqrt{\frac{G}{3}}\mathcal{I}'^{(M)}_3}\hat{g}',\chi\Big] \label{ac6}
							\end{split}
							\end{equation}
							We will call this particular form \textbf{metric $f(R)$ in Einstein frame}. Equations of motion are easy to write:
							\begin{itemize}
								\item $\delta \hat{g}'$: $\hat{G}'_{\mu\nu}+G\Big[\frac{1}{4}\hat{g}'_{\mu\nu}\hat{g}'^{\alpha\beta}-\frac{1}{2}\delta^\alpha_\mu\delta^\beta_\nu\Big]\hat{\nabla}'_\alpha\mathcal{I}'^{(M)}_3\hat{\nabla}'_\beta\mathcal{I}'^{(M)}_3+\frac{G}{2}\hat{g}'_{\mu\nu}\hat{\nabla}'_\alpha\mathcal{I}'^{(M)}_2=G\kappa^2 T'_{\mu\nu}$
								\item $\delta \mathcal{I}'^{(M)}_3$: $\hat{\Box}'\mathcal{I}'^{(M)}_3-\frac{d\mathcal{I}'^{(M)}_2}{d\mathcal{I}'^{(M)}_3}=-\frac{2}{\sqrt{3}}\kappa^2\sqrt{G}\:T'$
							\end{itemize}
							Due of nonminimal coupling between matter and the scalar field, the energy-momentum tensor $T'_{\mu\nu}$ is not conserved. This result is not very surprising, as we simply reproduced the action functional of the second chapter starting from different assumptions. What is interesting to notice is that the action \ref{ace} has been written in an invariant form without any knowledge of invariants introduced later . 
					\subsection{Palatini approach}
							Let us now see what happens if we apply Palatini variation to $f(R)$ theories in scalar-tensor form. We expect to reproduce results of the previous chapter and obtain a well-known result, that Palatini and metric approaches are not compatible and hence, there is no conformal transformation connecting these two formalisms. We proceed by writing the same action as in \ref{ac2}, but this time we regard the curvature scalar as a function of both the metric and the independent connection: $\mathcal{R}=\mathcal{R}(g,\Gamma)$. In a similar manner we can introduce an auxiliary field and end up in \textbf{Palatini $f(R)$ in Jordan frame}:
							\begin{equation}
							S[g,\Gamma,\Phi]=\frac{1}{2\kappa^2}\int_{\Omega}d^4x\sqrt{-g}\Big(\Phi \mathcal{R}-V(\Phi)\Big)+S_{\text{matter}}[g,\chi^m] \label{pa1}
							\end{equation}
							where $V(\Phi)$ is defined as in the previous section. Let us notice that in the Palatini approach we do not let the connection enter the matter part of the action; otherwise, we would be calling such theory 'metric-affine'. Also, despite it apparent similarity to \ref{ac4}, unlike that action functional this one is not equivalent to BD theory with $\omega=0$ because the connection $\Gamma^\alpha_{\mu\nu}$ is not a Levi-Civita connection of the metric $g_{\mu\nu}$. Original BD theory was of course formulated in metric approach, but it will be shown that an initial Palatini formulation can be brought to a BD-like form with $\omega=-\frac{3}{2}$. Let us first, analogously to what we did in the previous subsection, write out five coefficients $\{\mathcal{A},\mathcal{B},\mathcal{C},\mathcal{V},\alpha\}$ introduced for scalar-tensor theories of gravity in Palatini approach:
							\begin{itemize}
								\item $\mathcal{A}(\Phi)=\Phi$
								\item $\mathcal{B}(\Phi)=0$
								\item $\mathcal{C}(\Phi)=0$
								\item $\mathcal{V}(\Phi)=V(\Phi)$
								\item $\alpha(\Phi)=0$
							\end{itemize}
							We can also list the invariant quantities built from these coefficients:
							\begin{itemize}
								\item $\mathcal{I}^{(P)}_1(\Phi)=\frac{\mathcal{A}(\Phi)}{e^{2\alpha(\Phi)}}=\Phi$
								\item $\mathcal{I}^{(P)}_2(\Phi)=\frac{\mathcal{V}(\Phi)}{(\mathcal{A}(\Phi))^2}=\frac{V(\Phi)}{\Phi^2}$
								\item $\mathcal{I}^{(P)}_3(\Phi)=\pm\int^\Phi_{\Phi_0}\sqrt{\frac{\frac{3}{2}\mathcal{C}^2(\Phi')-\mathcal{A}(\Phi')\mathcal{B}(\Phi')+3\mathcal{A}'(\Phi')\mathcal{C}(\Phi')}{\mathcal{A}^2(\Phi')}}\:d\Phi'=0$
							\end{itemize}
							The invariant metrics and connections are thence:
							\begin{equation}
							\hat{g}_{\mu\nu}=\Phi g_{\mu\nu}, \qquad \tilde{g}_{\mu\nu}=g_{\mu\nu}, \qquad \hat{\Gamma}^\alpha_{\mu\nu}=\Gamma^\alpha_{\mu\nu}
							\end{equation}
							Let us notice that variables used coincide with the invariant quantities if we want to work in the Jordan frame. For this reason, we shall keep up analysing the theory in this frame. In order to obtain equations of motion, all we have to do is to come back to equations \ref{eqp1}, \ref{eqp2} and \ref{eqp3} and plug in functions of the scalar field given above. But, since this case was analysed in section dedicated to investigation of action functional expressed in terms of $(\tilde{g}_{\mu\nu},\hat{\Gamma}^\alpha_{\mu\nu},\mathcal{I}^{(P)}_1(\Phi))$, we can write:
							\begin{equation}
							S[\tilde{g},\hat{\Gamma},\mathcal{I}^{(P)}_1]=\frac{1}{2\kappa^2}\int_{\Omega}d^4x\sqrt{-\tilde{g}}\Big[\mathcal{I}^{(P)}_1\hat{R}(\tilde{g},\hat{\Gamma})-\mathcal{I}^{(P)}_4\Big]+S_\text{matter}[\tilde{g},\chi^m] \label{inac1}
							\end{equation}
							(where $\mathcal{I}^{(P)}_4=\Big(\mathcal{I}^{(P)}_1\Big)^2\mathcal{I}^{(P)}_2$). We have already varied this action; we remember that the invariant connection was a function of the scalar field and the invariant metric $\tilde{g}_{\mu\nu}$, and hence could be viewed as an auxiliary field and eliminated from the action, yielding a different one:
							\begin{equation}
							\begin{split}
							S[\tilde{g},\mathcal{I}^{(P)}_1]=&\frac{1}{2\kappa^2}\int_{\Omega}d^4x\sqrt{-\tilde{g}}\Bigg[\mathcal{I}^{(P)}_1\tilde{R}(\tilde{g})+\tilde{g}^{\mu\nu}\frac{3}{2{\mathcal{I}_1^{(P)}}}\hat{\nabla}_\mu\mathcal{I}^{(P)}_1\hat{\nabla}_\nu\mathcal{I}^{(P)}_1-\mathcal{I}^{(P)}_4\Bigg]+S_\text{matter}[\tilde{g},\chi^m] \label{inac2}
							\end{split}
							\end{equation}
							As we can see, this action corresponds to BD theory with $\omega=-\frac{3}{2}$; this clearly shows that Palatini variation alters the physics described by the theory, as in case of the purely metric approach we got a BD-theory with $\omega=0$. These two theories are different, which could be also seen by comparing values of the invariant. Invariants $\mathcal{I}^{(P)}_3$ and $\mathcal{I}^{(M)}_3$ should correspond to one another, but in case of the metric approach value of the invariant is different from zero, whereas in case of Palatini approach this invariant vanishes. Thus, the theories are not commensurable, as it was said in the previous chapter.
							
							Without any further discussion, we can plug all relevant quantities in \ref{eqn} and \ref{skalar}, and obtain equations of motion:
							\begin{enumerate}
								\item $\tilde{G}^{\mu\nu}+\frac{3}{2(\mathcal{I}^{(P)}_1)^2}\tilde{\nabla}_\alpha\mathcal{I}^{(P)}_1\tilde{\nabla}_\beta\mathcal{I}^{(P)}_1\Big(\tilde{g}^{\alpha\mu}\tilde{g}^{\beta\nu}-\frac{1}{2}\tilde{g}^{\nu\mu}\tilde{g}^{\beta\alpha}\Big)+\frac{1}{2\mathcal{I}^{(P)}_1}\tilde{g}^{\mu\nu}\mathcal{I}^{(P)}_4+\frac{1}{\mathcal{I}^{(P)}_1}\big(\tilde{g}^{\mu\nu}\tilde{\Box}-\tilde{\nabla}^\mu\tilde{\nabla}^\nu\big)\mathcal{I}^{(P)}_1=\frac{\kappa^2}{\mathcal{I}^{(P)}_1}\tilde{T}^{\mu\nu}$
								\item $\tilde{R}+\frac{3}{2(\mathcal{I}^{(P)}_1)^2}\tilde{g}^{\alpha\beta}\tilde{\nabla}_\alpha\mathcal{I}^{(P)}_1\tilde{\nabla}_\beta\mathcal{I}^{(P)}_1-\frac{3}{\mathcal{I}^{(P)}_1}\tilde{\Box}\mathcal{I}^{(P)}_1-(\mathcal{I}^{(P)}_4)'=0$
							\end{enumerate}
							Let us now contract the first equation with the metric tensor $\tilde{g}_{\mu\nu}$:
							\begin{equation}
							-\tilde{R}-\frac{3}{2(\mathcal{I}^{(P)}_1)^2}\tilde{g}^{\alpha\beta}\tilde{\nabla}_\alpha\mathcal{I}^{(P)}_1\tilde{\nabla}_\beta\mathcal{I}^{(P)}_1+\frac{2}{\mathcal{I}^{(P)}_1}\mathcal{I}^{(P)}_4+\frac{3}{\mathcal{I}^{(P)}_1}\tilde{\Box}\mathcal{I}^{(P)}_1=\frac{\kappa^2}{\mathcal{I}^{(P)}_1}\tilde{T}
							\end{equation}
							This formula can be added to the second equation of motion, yield the following, interesting result:
							\begin{equation}
							\kappa^2\tilde{T}+\mathcal{I}^{(P)}_1(\mathcal{I}^{(P)}_4)'-2\mathcal{I}^{(P)}_4=0
							\end{equation}
							This equation is an analogue of \ref{rem} (or \ref{eqsc} with the potential added) with $\omega=-\frac{3}{2}$. There is, however, a small, yet very important difference between these two results: in case of the metric approach, relation between the scalar field and matter is dynamical, as there is the d'Alembert operator acting on the scalar field. Here, the relation is purely algebraic, and this happens only for one special value of $\omega$. In the vacuum, when $\tilde{T}_{\mu\nu}=0$, a solution of the above equation is $\mathcal{I}^{(P)}_1=\text{const}$ - the Einstein theory with cosmological constant added \cite{sotir}. It must be noted that such theory was not originally considered by Brans and Dicke, as they did not include the potential term in the action. If the potential $\mathcal{I}^{(P)}_4$ were absent from the action functional, the theory would be ill-posed since it would lead to a nonsensical equation $\tilde{T}=0$, satisfied identically only by radiation. However, in case of $f(R)$ theories the potential cannot vanish, for it would contradict the underlying assumption $f''(R)\neq 0$ (because the potential is given by $V(R)=f'(R)R-f(R)$, which can be satisfied only by $f(R)=\text{const}\cdot R \rightarrow f''(R)=0$) \cite{sotir}. 
							
							For completeness, let us also express the action functional in terms of the invariant metric $\hat{g}_{\mu\nu}=\Phi g_{\mu\nu}$. This will allow us to make use of formulae introduced in the section dedicated to treatment of the theory in the Einstein-like frame. For this reason, we shall refer to the following formulation as \textbf{Palatini $f(R)$ in the Einstein frame}. We will, however, encounter here a problem: the dynamical independent variable being a 'replacement' of the scalar field was the invariant $\mathcal{I}^{(P)}_3$, but in this case it is identically equal to zero. However, the main reason for introducing such invariant was our need for writing the action functional in an invariant form. Here, due to the absence of two of the coefficients, we can make use of a 'vicarious' invariant $\mathcal{I}^{(P)}_1$, which seems to be a natural choice since $\mathcal{I}^{(P)}_1=\Phi=\mathcal{A}(\Phi)$. This does not alter the dynamics, it only means that the potential becomes a function of the 'new' variable $\mathcal{I}^{(P)}_1$:
							\begin{equation}
							S[\hat{g},\hat{\Gamma},\mathcal{I}^{(P)}_1]=\frac{1}{2\kappa^2}\int_{\Omega}d^4x\sqrt{-\hat{g}}\big[\hat{R}(\hat{g},\hat{\Gamma})-\mathcal{I}^{(P)}_2\big]+S_\text{matter}\Big(\frac{1}{\mathcal{I}^{(P)}_1}\hat{g},\chi\Big) \label{ill}
							\end{equation}
							Let us now compute the equations of motion:
							\begin{itemize}
								\item $\hat{G}_{\mu\nu}(\hat{g},\hat{\Gamma})+\frac{1}{2}\hat{g}_{\mu\nu}\mathcal{I}^{(P)}_2=\kappa^2 \hat{T}_{\mu\nu}$
								\item $\hat{\nabla}_\alpha\big(\sqrt{-\hat{g}}\hat{g}^{\mu\nu}\big)=0$
								\item $(\mathcal{I}^{(P)}_2)'=\frac{\kappa^2}{\mathcal{I}^{(P)}_1}\hat{T}$
							\end{itemize}
							The second equation tells us that the invariant connection is a Levi-Civita connection of $\hat{g}_{\mu\nu}$. The last equation suggests that in case of $\hat{T}_{\mu\nu}=0$, the only possible solution is $\mathcal{I}^{(P)}_2=\Lambda=\text{const}$, which substituted in the first equation yields again the Einstein vacuum field equations with cosmological constant $\hat{R}_{\mu\nu}=\frac{1}{2}\hat{g}_{\mu\nu}\Lambda$. Also, taking the divergence of the first equation and combining it with the one describing the scalar field, we get:
							\begin{equation}
							\frac{1}{2}\hat{T}\hat{\nabla}_\nu\big(\text{ln}\mathcal{I}^{(P)}_1\big)=\hat{\nabla}^\mu\hat{T}_{\mu\nu}
							\end{equation}
							However, following the prescription given in the section \hyperref[subsec:inv1]{(3.5.1)}, we can add the energy-momentum tensor defined for the scalar field $\kappa^2 \hat{T}^\Phi_{\mu\nu}=-\frac{1}{2}\hat{g}_{\mu\nu}\mathcal{I}^{(P)}_2$; they altogether make up a quantity which remains conserved: $\mathfrak{\hat{T}}_{\mu\nu}=\hat{T}^\Phi_{\mu\nu}+\hat{T}_{\mu\nu} \rightarrow \hat{\nabla}_\mu \mathfrak{\hat{T}}^{\mu\nu}=0$.
							
					\section{Friedmann equations}
							In this chapter we will apply our results to the problem of determining cosmic evolution of the Universe. In the standard, Einsteinian approach one calculates field equation and makes the following assumptions: the Earth is not located in a preferred position in the Universe or, to push this assertion even further, all places in the Universe are equivalent - this is so-called \textit{cosmological principle}, extending the idea put forth by Copernicus that in fact our planet does not occupy a central position in the solar system. It must be remembered that the cosmological principle is more of an assumption rather that an observational fact. There exist cosmological models that do not take the principle as their tenet \cite{pleb}. However, if we assume that the cosmological principle is correct, then it has a deep geometrical meaning for our theory: spacetime must be isotropic around every point \cite{bible}, which makes it homogeneous \cite{pleb}. This translates into a form of the metric describing such spacetime:
							\begin{equation}
							ds^2=dt^2-\frac{a(t)}{(1+\frac{1}{4}r^2)^2}[dr^2+r^2(d\theta^2+\text{sin}^2(\theta)d\phi^2)]
							\end{equation}
							where $\theta, \phi$ measure angles in spherical coordinates, $a(t)$ is the scale factor, $k$ is the spatial curvature and $r$ is related to the standard radius $\rho$ by $\rho=\frac{r}{1+\frac{1}{4}r^2}$. It was first obtained in the 1920s and 1930s by Alexander Friedmann, Georges Lemaître, Howard P. Robertson and Arthur Geoffrey Walker \cite{bible}. The Einstein Field Equations are necessary to give us a relation between the scale factor, 'measuring' the size of the Universe, and cosmic time. Solution of these equation is called 'Friedman equations'; it was discovered by aforementioned Alexander Friedman in 1922 \cite{fried}. In order to solve EFE, one needs to calculate first all geometric objects depending on the metric which enter the equations of motion. Because we will make use of them later on, it seems reasonable to list them below \cite{sok} (switching however to the classical variable $\rho$).
							
							\noindent \textbf{Connection:}
							\begin{equation}
							\begin{split}
							& \Gamma^0_{11}=\frac{a\dot{a}}{1-k\rho^2}, \qquad \Gamma^0_{22}=a\dot{a}\rho^2, \qquad \Gamma^0_{33}=a\dot{a}\rho^2\text{sin}^2\theta, \\
							& \Gamma^1_{01}=\Gamma^2_{02}=\Gamma^3_{03}=\frac{\dot{a}}{a}, \qquad \Gamma^1_{11}=\frac{k\rho}{1-k\rho^2}, \qquad \Gamma^1_{22}=-(1-k\rho^2)\rho, \\
							& \Gamma^1_{33}=-(1-k\rho^2)\rho\text{sin}^2\theta, \qquad \Gamma^2_{12}=\Gamma^3_{13}=\frac{1}{\rho}, \\
							& \Gamma^{2}_{33}=-\text{sin}\theta\text{cos}\theta, \qquad \Gamma^3_{23}=\text{ctg}\theta
							\end{split}
							\end{equation}
							
							\noindent \textbf{Ricci tensor:}
							\begin{equation}
							R^0_{\:\:0}=-3\frac{\ddot{a}}{a}, \qquad R^1_{\:\:1}=R^2_{\:\:2}=R^3_{\:\:3}=-\frac{1}{a^2}(a\ddot{a}+2\dot{a}^2+2k)
							\end{equation}
							
							\noindent \textbf{Curvature scalar:}
							\begin{equation}
							R=-\frac{6}{a^2}(a\ddot{a}+\dot{a}^2+2k)
							\end{equation}
							
							If we assume that matter filling the Universe can be modelled by perfect fluid, then its energy-momentum tensor has the following form \cite{bible}:
							\begin{equation}
							T_{\mu\nu}=(\rho+P)u_\mu u_\nu-Pg_{\mu\nu}
							\end{equation}
							where  $u_\mu$ is 4-velocity between the fluid and an observer, $\rho$ is the energy density of a given component of matter (not to be mistaken for the radial component), and $P$ denotes pressure. They can be related by so-called equation of state, giving us a direct dependence of $P$ on $\rho$: 
							\begin{equation}
							P=P(\rho)
							\end{equation}
							Having matter on the right hand side of EFE, we can write the Friedmann equations (without cosmological constant):
							\begin{equation}
							\frac{\dot{a}^2+k}{a^2}=\frac{8\pi \rho}{3}
							\end{equation}
							and 
							\begin{equation}
							\frac{\ddot{a}}{a}=-\frac{4\pi}{3}(\rho+3P)
							\end{equation}
							which, in fact, is a result of the first Friedmann equation and the conservation of energy-momentum tensor:
							\begin{equation}
							\nabla_\mu T^{\mu\nu}=0 \rightarrow \dot{\rho}+3\frac{\dot{a}}{a}(\rho+P)=0
							\end{equation}
							These standard results allow us to determine evolution of the entire Universe and so far, enriched with the cosmological constant, cold dark matter and inflationary scenario, have been very successful in explaining its dynamics, structure formation and particle abundances \cite{capo}. However, as it was mentioned in the Introduction, certain aspects of the theory seem somehow \textit{ad hoc}, such as the inflation. In the standard GR there is no scalar field that could cause the expansion of early Universe to accelerate rapidly; it has to be inserted in the equations by hand. Thus, we may wonder whether considering some alternative theory (viewed more as a toy model) could possibly lead to a simpler and more homogeneous description of phenomena observed in the Universe. In this chapter we will take the first step toward this aim: we will compute Friedmann equations for a very simple case of an empty universe, endowed with a zero spatial curvature. We will carry out all calculations using both invariant metrics.
								\subsection{Friedmann equations in the Einstein-like frame}
								The Einstein-like frame is characterized by the following set of variables:
								\begin{itemize}
									\item metric: $\hat{g}_{\mu\nu}=\mathcal{A}(\Phi)g_{\mu\nu}=\text{diag}(\mathcal{A}(\Phi),-a^2(t)\mathcal{A}(\Phi),-a^2(t)\mathcal{A}(\Phi)\rho^2,-a^2(t)\mathcal{A}(\Phi)\rho^2\text{sin}^2(\theta))$
									\item connection: $\hat{\Gamma}^\alpha_{\mu\nu}=\hat{\Gamma}^\alpha_{\mu\nu}(\hat{g})$
									\item scalar field: $\mathcal{I}_3=\mathcal{I}_3(\Phi)$
								\end{itemize}
								Field equation were given in \hyperref[subsec:inv1]{(3.5.1.)}. We need to compute now the invariant Ricci tensor and scalar curvature for the invariant metric $\hat{g}_{\mu\nu}$; the easiest way will be to use conformal transformation formulae since we know the Ricci tensor and curvature scalar expressed in terms of the metric $g_{\mu\nu}$. Because both quantities depend on two different metrics, they can be related by formulae given in the first chapter with a proper identification of the conformal factor $\Omega(x)\equiv\Big[\mathcal{A}(\Phi(x))\Big]^{1/2}$, so that:
								\begin{equation}
								\begin{split}
								&\hat{R}_{\mu\nu}=R_{\mu\nu}-2\nabla_\mu\nabla_\nu\text{ln}(\mathcal{A})^{1/2}-g_{\mu\nu}\Box\text{ln}(\mathcal{A})^{1/2}+2\nabla_\mu\text{ln}(\mathcal{A})^{1/2}\nabla_\nu\text{ln}(\mathcal{A})^{1/2}-\\
								&+2g_{\mu\nu}g^{\alpha\beta}\nabla_\alpha\text{ln}(\mathcal{A})^{1/2}\nabla_\beta\text{ln}(\mathcal{A})^{1/2}
								\end{split} \label{formula}
								\end{equation}
								The coefficient $\mathcal{A}$ depends on the spacetime position indirectly, through the scalar field. Because of the homogeneity assumption, the field cannot depend on spatial position, it can only a function of the cosmic time: $\Phi=\Phi(t)$. Hence, the formula given above should simplify greatly since the only nonvanishing derivative will be the one with respect to time, $\partial_t\equiv\partial_0$. In order to avoid using Christoffel symbols, we can utilize the well-known formula for the Laplace-Beltrami operator $\Box=\frac{1}{\sqrt{-g}}\partial_\mu\big(g^{\mu\nu}\sqrt{-g}\partial_\nu\big)$, where $g=-a^6(t)\rho^4\text{sin}^2(\theta)$. This allows us to cast the seemingly complicated formula in a more clear form:
								\begin{equation}
								\begin{split}
								\hat{R}_{\mu\nu}=& R_{\mu\nu}-\delta^0_{\:\:\mu}\delta^0_{\:\:\nu}\partial_0\Big(\frac{\partial_0\mathcal{A}}{\mathcal{A}}\Big)+\frac{1}{\mathcal{A}}\Gamma^0_{\mu\nu}\partial_0\mathcal{A}-\frac{1}{2}g_{\mu\nu}\frac{1}{a^3}\partial_0\Big(a^3\frac{\partial_0\mathcal{A}}{\mathcal{A}}\Big)+\\
								& + \frac{1}{2}\Big(\frac{\partial_0\mathcal{A}}{\mathcal{A}}\Big)^2(\delta^0_{\:\:\mu}\delta^0_{\:\:\nu}-g_{\mu\nu})
								\end{split}
								\end{equation}
								Taking into account that $\partial_0\mathcal{A}=\frac{d\mathcal{A}}{d\mathcal{I}_3}\frac{d\mathcal{I}_3}{dt}\equiv\mathcal{A}'\mathcal{\dot{I}}_3 $, we can write out all components of the Ricci tensor:
								\begin{equation}
								\hat{R}_{00}=-3\frac{\ddot{a}}{a}-\frac{3}{2}\frac{\mathcal{A}''}{\mathcal{A}}\mathcal{\dot{I}}^2_3-\frac{3}{2}\frac{\mathcal{A}'}{\mathcal{A}}\mathcal{\ddot{I}}_3-\frac{3}{2}\frac{\dot{a}}{a}\frac{\mathcal{A}'}{\mathcal{A}}\mathcal{\dot{I}}_3+\frac{6}{4}\Big(\frac{\mathcal{A}'}{\mathcal{A}}\Big)^2\mathcal{\dot{I}}^2_3
								\end{equation}
								and 
								\begin{equation}
								\hat{R}_{ii}=g_{ii}\Big\{-\frac{\ddot{a}}{a}-2\Big(\frac{\dot{a}}{a}\Big)^2-\frac{5}{2}\frac{\dot{a}}{a}\frac{\mathcal{A}'}{\mathcal{A}}\mathcal{\dot{I}}_3-\frac{1}{2}\frac{\mathcal{A}''}{\mathcal{A}}\mathcal{\dot{I}}^2_3-\frac{1}{2}\frac{\mathcal{A}'}{\mathcal{A}}\mathcal{\ddot{I}}_3\Big\}
								\end{equation}
								Formula describing the curvature tensor reads as follows:
								\begin{equation}
								\hat{R}=\frac{1}{\mathcal{A}}\Big[-6\frac{\ddot{a}}{a}-6\Big(\frac{\dot{a}}{a}\Big)^2+\frac{6}{4}\Big(\frac{\mathcal{A}'}{\mathcal{A}}\Big)^2\mathcal{\dot{I}}^2_3-9\frac{\dot{a}}{a}\frac{\mathcal{A}'}{\mathcal{A}}\mathcal{\dot{I}}_3-3\frac{\mathcal{A}''}{\mathcal{A}}\mathcal{\dot{I}}^2_3-3\frac{\mathcal{A}'}{\mathcal{A}}\mathcal{\ddot{I}}_3\Big]
								\end{equation}
								Hence, the $00$ component of field equations for the metric tensor in Palatini approach is the following:
								\begin{equation}
								\begin{split}
								& \hat{R}_{00}-\hat{g}_{00}\hat{R}-\frac{1}{2}\mathcal{\dot{I}}^2+\frac{1}{2}\hat{g}_{00}\mathcal{I}_2=\\
								& 3\frac{\dot{a}}{a}\frac{\mathcal{A}'}{\mathcal{A}}\mathcal{\dot{I}}_3+\frac{3}{4}\Big(\frac{\mathcal{A}'}{\mathcal{A}}\Big)^2\mathcal{\dot{I}}^2_3+3\Big(\frac{\dot{a}}{a}\Big)^2-\frac{1}{2}\mathcal{\dot{I}}^2+\frac{1}{2}\hat{g}_{00}\mathcal{I}_2=0
								\end{split} \label{zero}
								\end{equation}
								Now, let us denote $\frac{\dot{a}}{a}\equiv H(t)$. We can now perform a coordinate transformation and write the cosmic time as a function of a new, invariant time $\hat{t}$: $\frac{d}{dt}=\sqrt{\mathcal{A}}\frac{d}{d\hat{t}}$. This transformation can be accompanied by a redefinition of the scale factor: $\hat{a}(\hat{t})=\sqrt{\mathcal{A}}a(t)$, which introduces a new, invariant Hubble parameter, related to the 'old' one via:
								$$\hat{H}=\frac{\dot{\hat{a}}}{\hat{a}}=\frac{1}{\sqrt{A}}\Big(H+\frac{1}{2}\frac{\mathcal{A}'}{\mathcal{A}}\mathcal{\dot{I}}_3\Big)$$
								Hubble parameter defined in this way is invariant under a conformal change $g_{\mu\nu}=e^{2\bar{\gamma}}\bar{g}_{\mu\nu}$, $\mathcal{\bar{A}}=e^{2\bar{\gamma}}\mathcal{A}$, $\bar{a}=e^{-\bar{\gamma}}a$ and redefinition of the cosmic time $d\bar{t}=e^{-\bar{\gamma}}dt$:
								\begin{equation}
								\begin{split}
								& \hat{H}=\frac{1}{\sqrt{A}}\Big(H+\frac{1}{2}\frac{\mathcal{A}'}{\mathcal{A}}\mathcal{\dot{I}}_3\Big)=\\
								& \frac{e^{\bar{\gamma}}}{\sqrt{\mathcal{\bar{A}}}}\Bigg(\frac{\frac{d}{dt}\Big(e^{\bar{\gamma}}\bar{a}\Big)}{e^{\bar{\gamma}}\bar{a}}+\frac{1}{2}\frac{\frac{d}{d\mathcal{I}_3}\Big(e^{-2\bar{\gamma}}\mathcal{\bar{A}}\Big)}{e^{-2\bar{\gamma}}\mathcal{\bar{A}}}\Bigg)=\\
								&\frac{e^{\bar{\gamma}}}{\sqrt{\mathcal{\bar{A}}}}\Big(\frac{d\bar{\gamma}}{dt}+\frac{1}{\bar{a}}\frac{d\bar{a}}{dt}-\frac{d\bar{\gamma}}{dt}+\frac{1}{2}\frac{1}{\mathcal{\bar{A}}}\frac{d\mathcal{\bar{A}}}{dt}\Big)=\\
								& =\frac{1}{\sqrt{\mathcal{\bar{A}}}}\Big(\frac{1}{\bar{a}}\frac{d\bar{a}}{d\bar{t}}+\frac{1}{2}\frac{1}{\mathcal{\bar{A}}}\frac{d\mathcal{\bar{A}}}{d\bar{t}}\Big)
								\end{split}
								\end{equation}
								This redefinition of both the scale factor and the cosmic time is equivalent to writing the invariant metric in another form, preserving structure of the FRWL metric:
								$$ds^2=d\hat{t}^2-\hat{a}^2(\hat{t})(d\rho^2+\rho^2d\Omega^2) \rightarrow \hat{g}_{\mu\nu}=\text{diag}(1,-\hat{a}^2(\hat{t}),-\rho^2\hat{a}^2(\hat{t}),-\hat{a}^2(\hat{t})\rho^2\text{sin}^2(\theta))$$
								Coming back to the equation \ref{zero}, we can write it in an invariant form:
								\begin{equation}
								\hat{H}^2=\frac{1}{6}\Big(\frac{d\mathcal{I}_3}{d\hat{t}}\Big)^2-\frac{1}{6}\mathcal{I}_2 \label{eq1}
								\end{equation}
								Let us now compute field equations for spatial components $ii$ (all tensors are diagonal, so that there are no components different from the mentioned):
								\begin{equation}
								\hat{R}_{ii}-\frac{1}{2}\hat{g}_{ii}\hat{R}+\frac{1}{2}\hat{g}^{00}\hat{g}_{ii}\mathcal{\dot{I}}^2_3+\frac{1}{2}\hat{g}_{ii}\mathcal{I}_2=0
								\end{equation}
								Plugging in formulae obtained so far, after some calculations we end up with the following result:
								\begin{equation}
								2\frac{\ddot{a}}{a}+\Big(\frac{\dot{a}}{a}\Big)^2+2\frac{\dot{a}}{a}\frac{\mathcal{A}'}{\mathcal{A}}\mathcal{\dot{I}}_3-\frac{3}{4}\Big(\frac{\mathcal{A}'}{\mathcal{A}}\Big)^2\mathcal{\dot{I}}^2_3+\frac{\mathcal{A}''}{\mathcal{A}}\mathcal{\dot{I}}^2_3+\frac{\mathcal{A}'}{\mathcal{A}}\mathcal{\ddot{I}}_3+\frac{1}{2}\mathcal{\dot{I}}^2_3+\frac{1}{2}A\mathcal{I}_2=0
								\end{equation}
								Now, we replace time derivative with derivative with respect to the invariant cosmic time according to $\frac{d}{dt}=\sqrt{\mathcal{A}}\frac{d}{d\hat{t}}$, and simultaneously redefine the scale factor: $a=\frac{1}{\sqrt{\mathcal{A}}}\hat{a}$. After very tedious calculations, we arrive at the following equation:
								\begin{equation}
								2\frac{d\hat{H}}{d\hat{t}}+3\hat{H}^2=-\frac{1}{2}\Big(\frac{d\mathcal{I}_3}{d\hat{t}}\Big)^2-\frac{1}{2}\mathcal{I}_2
								\end{equation}
								The last formula results from the equation governing evolution of the scalar field. It is given by:
								\begin{equation}
								2\frac{1}{\sqrt{-\hat{g}}}\partial_0\Big(\hat{g}^{00}\sqrt{-\hat{g}}\:\partial_0\mathcal{I}_3\Big)-\frac{d\mathcal{I}_2}{d\mathcal{I}_3}=0
								\end{equation}
								where $\hat{g}=-\mathcal{A}^4a^6(t)\rho^4\text{sin}^2(\theta)$. Short calculation shows that:
								\begin{equation}
								\frac{6\dot{a}\mathcal{\dot{I}}_3}{\mathcal{A}}+\frac{2\mathcal{\dot{A}}\mathcal{\dot{I}}_3}{\mathcal{A}^2}+\frac{2\mathcal{\ddot{I}}_3}{\mathcal{A}}-\frac{d\mathcal{I}_2}{d\mathcal{I}_3}=0
								\end{equation}
								Again, we plug in invariant quantities and the solution is:
								\begin{equation}
								\frac{d^2}{d\hat{t}^2}\mathcal{I}_3+3\hat{H}\frac{d\mathcal{I}_3}{d\hat{t}}-\frac{1}{2}\frac{d\mathcal{I}_2}{d\mathcal{I}_3}=0 \label{eq2}
								\end{equation}
								Our results match closely those derived in \cite{kuinv} (apart from differences in numerical factors). From this point, it is easy to obtain inflationary behaviour of the Universe. We need to introduce the energy-momentum tensor for the scalar field according to the prescription given in \hyperref[subs:inv1]{(3.5.1.)}:
								$$\kappa^2 \hat{T}^\Phi_{\mu\nu}=\hat{\nabla}_\mu\mathcal{I}_3\hat{\nabla}_\nu\mathcal{I}_3-\frac{1}{2}\hat{g}_{\mu\nu}\hat{g}^{\alpha\beta}\hat{\nabla}_\alpha\mathcal{I}_3\hat{\nabla}_\beta\mathcal{I}_3-\frac{1}{2}\hat{g}_{\mu\nu}\mathcal{I}_2$$
								We can now define energy density of the field:
								\begin{equation}
								\rho_\Phi:=\frac{1}{\kappa^2}\hat{T}^0_0=\frac{1}{2\kappa^2}\Big(\dot{\mathcal{I}}^2_3-\mathcal{I}_2\Big)
								\end{equation} 
								together with its pressure:
								\begin{equation}
								p_\Phi:=-\frac{1}{3\kappa^2}\hat{T}^i_i=\frac{1}{2\kappa^2}\Big(\dot{\mathcal{I}}^2_3+\mathcal{I}_2\Big)
								\end{equation}
								For simplicity, we denoted differentiation with respect to the time $\hat{t}$ by a dot; this should not be mistaken for differentiating with respect to the cosmic time $t$. 
								
								In order for the theory to exhibit inflationary behaviour, the potential $\mathcal{I}_2$ must necessarily dominate over the kinetic term $\mathcal{I}_2\gg\dot{\mathcal{I}}^2_3$, so that the equation of state is simply $p_\Phi=-\rho_\Phi$ \cite{infl}. Also, the so-called slow-roll condition enforces vanishing of the second time derivative of the field $\mathcal{I}_3$, and the equations \ref{eq1},\ref{eq2} now take the form:
								\begin{equation}
								\hat{H}^2\approx-\frac{1}{6}\mathcal{I}_2
								\end{equation}
								\begin{equation}
								3\hat{H}\frac{d\mathcal{I}_3}{d\hat{t}}-\frac{1}{2}\frac{d\mathcal{I}_2}{d\mathcal{I}_3}\approx0
								\end{equation}
								The number of e-folds, measuring the amount of inflation, is given by the following formula:
								\begin{equation}
								N(\mathcal{I}_3)=\int_{\mathcal{I}_3}^{{\mathcal{I}_3}_{\text{end}}}d\text{ln}\hat{a}\approx\int_{\mathcal{I}_3}^{{\mathcal{I}_3}_{\text{end}}}\frac{\mathcal{I}_2}{\mathcal{I}'_2}d\mathcal{I}'_3
								\end{equation}
								where ${\mathcal{I}_3}_{\text{end}}$ is the value of the scalar field at the end of inflation. This quantity is clearly frame-independent. 
								
								Of course, what we presented here was merely a translational task of casting the usual Friedmann equations in a frame-independent form, preparing a set-up for analysis of realistic phenomena, such as inflation. We must also make now an important remark: origins of the scalar field present in the theory are still unknown, while $f(R)$ theories introduce scalar fields in a natural way, through a Legendre transformation. However, in the Einstein frame (we are still working in the Palatini approach) the scalar field has no dynamics, so that it cannot account for inflationary behaviour. On the other had, in the metric Einstein frame the scalar field preserves its dynamics. Even if the matter is present in the theory, Palatini theory in Einstein frame is not a viable model for explaining inflation; this is obviously a consequence of the field equations, which suggest that the connection is in fact Levi-Civita of the metric, and the theory reduces to GR with scalar field coupled to matter. A starting point different from Palatini $f(R)$ in Einstein frame may of course result in a correct dynamic of the scalar field.
								\subsection{Friedmann equations in the Jordan-like frame}
								This theory is fully metric and hence can be characterized by only two variables. Before we performed the conformal transformation metric was assumed to be the FRWL metric; after the transformation it gets multiplied by a conformal factor. Similarly, the scalar field was replaced with its invariant counterpart. We end up with following variables:
								\begin{itemize}
									\item  metric: $\tilde{g}_{\mu\nu}=e^{2\alpha(\Phi)}(\Phi)g_{\mu\nu}=\text{diag}(e^{2\alpha(\Phi)},-a^2(t)e^{2\alpha(\Phi)},-a^2(t)e^{2\alpha(\Phi)}\rho^2,-a^2(t)e^{2\alpha(\Phi)}\rho^2\text{sin}^2(\theta))$
									\item scalar field: $\mathcal{I}_1=\mathcal{I}_1(\Phi)$
								\end{itemize}
								In the Jordan-like frame equations of motion (without matter sources) take the following form:
								\begin{equation}
								\tilde{R}+\mathcal{I}'_5\tilde{g}^{\mu\nu}\tilde{\nabla}_\mu\mathcal{I}_1\tilde{\nabla}_\nu\mathcal{I}_1+2\mathcal{I}_5\tilde{\Box}\mathcal{I}_1-\mathcal{I}'_4=0 \label{eqn1}
								\end{equation} 
								and
								\begin{equation}
								\tilde{G}^{\mu\nu}-\frac{\mathcal{I}_5}{\mathcal{I}_1}\tilde{\nabla}_\alpha\mathcal{I}_1\tilde{\nabla}_\beta\mathcal{I}_1\Big(\tilde{g}^{\alpha\mu}\tilde{g}^{\beta\nu}-\frac{1}{2}\tilde{g}^{\nu\mu}\tilde{g}^{\beta\alpha}\Big)+\frac{1}{2\mathcal{I}_1}\tilde{g}^{\mu\nu}\mathcal{I}_4+\frac{1}{\mathcal{I}_1}\big(\tilde{g}^{\mu\nu}\tilde{\Box}-\tilde{\nabla}^\mu\tilde{\nabla}^\nu\big)\mathcal{I}_1=0 \label{eqn2}
								\end{equation}
								The first thing to do is to obtain explicit forms of the Ricci tensor and the curvature scalar, i.e. express them in terms of the scale factor $a(t)$ and, possibly, the scalar field and the coefficient $\alpha$. Analogously to what we did in the previous subsection, we will be aiming at writing the Friedmann equations in a fully invariant way. This can be achieved by using the formula \ref{formula} and identifying $\Omega(x)=e^{\alpha(\Phi(x))}$. We need to keep in mind that, due to the homogeneity assumption, the only nonvanishing derivative is the time derivative. Components of the Ricci tensor read as follows:
								\begin{equation}
								\tilde{R}_{00}=-3\Big(\frac{\ddot{a}}{a}+\alpha'\dot{\mathcal{I}}^2_1+\alpha'\ddot{\mathcal{I}}_1+\frac{\dot{a}}{a}\alpha'\dot{\mathcal{I}}_1\Big)
								\end{equation}
								\begin{equation}
								\tilde{R}_{ii}=g_{ii}\Big[-\frac{\ddot{a}}{a}-2\Big(\frac{\dot{a}}{a}\Big)^2-5\frac{\dot{a}}{a}\alpha'\dot{\mathcal{I}}_1-\alpha''\dot{\mathcal{I}}^2_1-\alpha'\ddot{\mathcal{I}}_1-2(\alpha')^2\dot{\mathcal{I}}^2_1\Big]
								\end{equation}
								where the prime denotes differentiating with respect to the scalar field $\mathcal{I}_1$. The curvature scalar can be easily computed; it turns out to be:
								\begin{equation}
								\tilde{R}=e^{-2\alpha}\Big\{-6\Big[\frac{\ddot{a}}{a}+\Big(\frac{\dot{a}}{a}\Big)^2+3\frac{\dot{a}}{a}\alpha'\dot{\mathcal{I}}_1+\alpha''\dot{\mathcal{I}}^2_1+\alpha'\ddot{\mathcal{I}}_1+(\alpha')^2\dot{\mathcal{I}}^2_1\Big]\Big\}
								\end{equation}
								Analogously to the steps taken in the previous subsection, we can change coordinates and redefine the scale factor, so that the metric tensor preserves its form:
								\begin{equation}
								\frac{d}{d\tilde{t}}=e^{-\alpha}\frac{d}{dt}\qquad\text{and}\qquad\tilde{a}(\tilde{t})=e^\alpha a(t)
								\end{equation}
								These new variables make up an invariant Hubble parameter $\tilde{H}$ (invariant in a sense precisely corresponding to the one discussed in the case of Einstein-like frame), defined to be:
								\begin{equation}
								\tilde{H}:=\frac{\frac{d}{d\tilde{t}}\tilde{a}}{\tilde{a}}
								\end{equation}
								Expressed in these new variables, the Einstein tensor takes on the following simplified form:
								\begin{equation}
								\tilde{G}_{00}=3e^{2\alpha}\tilde{H}^2
								\end{equation}
								\begin{equation}
								\tilde{G}_{ii}=e^{2\alpha}g_{ii}\Big[2\frac{d\tilde{H}}{d\tilde{t}}+3\tilde{H}^2\Big]
								\end{equation}
								We aim now at writing the equations \ref{eqn1}, \ref{eqn2} fully in terms of scalar field functions and the invariant Hubble parameter and invariant cosmic time. As a result we get three equation whereof one is redundant, but nevertheless written below for the sake of completeness:
								\begin{equation}
								3\tilde{H}^2-\frac{1}{2}\Big(\frac{d\mathcal{I}_1}{d\tilde{t}}\Big)^2\frac{\mathcal{I}_5}{\mathcal{I}_1}+\frac{\mathcal{I}_4}{2\mathcal{I}_1}+3\tilde{H}\frac{d\mathcal{I}_1}{d\tilde{t}}=0 \label{eeq1}
								\end{equation}
								\begin{equation}
								2\frac{d\tilde{H}}{d\tilde{t}}+3\tilde{H}^2+\frac{1}{2}\Big(\frac{d\mathcal{I}_1}{d\tilde{t}}\Big)^2\frac{\mathcal{I}_5}{\mathcal{I}_1}+\frac{\mathcal{I}_4}{2\mathcal{I}_1}+\frac{1}{\mathcal{I}_1}\Big(2\tilde{H}\frac{d\mathcal{I}_1}{d\tilde{t}}+\frac{d^2\mathcal{I}_1}{d\tilde{t}^2}\Big)=0
								\end{equation}
								\begin{equation}
								-6\frac{d\tilde{H}}{d\tilde{t}}-12\tilde{H}^2+\mathcal{I}'_5\Big(\frac{d\mathcal{I}_1}{d\tilde{t}}\Big)^2+2\mathcal{I}_5\Big(3\tilde{H}\frac{d\mathcal{I}_1}{d\tilde{t}}+\frac{d^2\mathcal{I}_1}{d\tilde{t}^2}\Big)-\mathcal{I}'_4=0 \label{eeq2}
								\end{equation}
								Let us now introduce energy-momentum tensor for the scalar field and, consequently, notions of its energy density and pressure:
								\begin{equation}
								-\kappa^2\tilde{T}^\Phi_{\mu\nu}=-2\mathcal{I}_5\tilde{\nabla}_\mu\mathcal{I}_1\tilde{\nabla}_\nu\mathcal{I}_1+\tilde{g}_{\mu\nu}\tilde{g}^{\alpha\beta}\mathcal{I}_5\tilde{\nabla}_\alpha\mathcal{I}_1\tilde{\nabla}_\beta\mathcal{I}_1+\mathcal{I}_4
								\end{equation}
								\begin{equation}
								\text{(energy density):} \tilde{\rho}_\Phi:=\kappa^2\Big(\tilde{T}^\Phi\Big)^0_{\:\:0}=-\mathcal{I}_5\Big(\frac{d\mathcal{I}_1}{d\tilde{t}}\Big)^2+\mathcal{I}_4
								\end{equation}
								\begin{equation}
								\text{(pressure):} \tilde{p}_\Phi:=-\frac{\kappa^2}{3}\Big(\tilde{T}^\Phi\Big)^i_{\:\:i}=-\mathcal{I}_5\Big(\frac{d\mathcal{I}_1}{d\tilde{t}}\Big)^2-\mathcal{I}_4
								\end{equation}
								Using energy density and pressure, we can write Friedmann equations \ref{eeq1}, \ref{eeq2} in a more familiar form:
								\begin{enumerate}
									\item $3\tilde{H}^2+3\tilde{H}\frac{d\mathcal{I}_1}{d\tilde{t}}+\frac{\tilde{\rho}_\Phi}{2\mathcal{I}_1}=0$
									\item $\frac{d\tilde{\rho}_\Phi}{d\tilde{t}}+3\tilde{H}(\tilde{\rho}_\Phi+\tilde{p}_\Phi)=-6\Big(\frac{d\tilde{H}}{d\tilde{t}}+2\tilde{H}^2\Big)\frac{d\mathcal{I}_1}{d\tilde{t}}$
								\end{enumerate}
								The second equation resembles the one resulting from taking divergence of the energy-momentum tensor defined for the scalar field. It this case, failure of the divergence to vanish can be viewed as a consequence of nonminimal coupling between the scalar field and curvature. The right hand side is a function of both the scalar field and redefined scale factor, and in fact it is precisely the curvature scalar multiplied by the time derivative of scalar field, playing the role of a mass term. This fact renders the theory difficult to analyse, and in case of inflationary behaviour simple conditions for the slow-roll, mentioned in the previous subsection, do not suffice. However, obtaining and investigating inflation is not impossible when the nonminimal coupling is present (at least in the metric case), as it was shown in \cite{infl1} and \cite{infl2}. The latter paper is more interesting from our point of view inasmuch as it tackles the problem using the language of invariants. In the paper by Kuusk \textit{et al.}, the exponential expansion is achieved by demanding an invariant quantity $\tilde{\epsilon}=-\frac{1}{\tilde{H}^2}\frac{d\tilde{H}}{d\tilde{t}}$ be much less than unity, $\tilde{\epsilon}\ll 1$. This requirement is accompanied by another one for realizing the slow-roll, but it is of purely technical nature and will not be discussed here. We can conclude our treatment of the Friedmann equations in Palatini approach with the following remark: due to the coupling between scalar field and curvature scalar, equations of motion are much more complicated and several further conditions need to be imposed; however, the procedure of obtaining inflationary behaviour is still perfectly feasible \cite{infl2}.
								\section{Conclusions}
								In this chapter the formalism we developed has been applied to $f(R)$ theories and Friedmann cosmology. At the beginning we analysed $f(R)$ theories of gravity in the metric and Palatini approaches using the language of invariants, reproducing in a self-consistent way the well-known fact that they are not commensurable. Metric theory is equivalent to Brans-Dicke theory with potential and $\omega=0$, and this identification can be held only in the Jordan frame, as in any other one nonexistent in the BD theory anomalous coupling of scalar field to matter shows up. On the other hand, Palatini approach yields yet another version of the BD theory, characterized by $\omega=-\frac{3}{2}$ in the Jordan frame. This very particular value of the free parameter renders the relation between scalar field and matter purely algebraic, whereas in the metric case the relation is dynamical. Thence, these two theories have different predictive power and describe different phenomena. We may ask a question whether any of these two theories is in agreement with the measurements. As we know, current Solar system experiments and cosmological data (binary pulsars \cite{puls}) prefer BD theories with very large values of $\omega$, greater than 40000, making the theory fine-tuned since the natural choice would be $\omega\sim1$. The fine-tuning can be abolished by assuming that the scalar field is massive and has a short range \cite{capo}. This argument is sound as in case of $f(R)$ theories the self-interaction potential never vanishes (that is, if we assume $f''(R)\neq0$), so that the theory cannot be ruled out based on the data obtained so far. 
								
								In the second part of the chapter Friedmann equations for an empty universe of vanishing spatial curvature were derived. Except for numerical factors, equations governing evolution of the scale factor of the Universe are exactly the same as in case of the metric approach \cite{infl2}. An important difference is that invariants entering the equations do not have the same meaning in both formalisms. Despite the fact that formulae have the same form, quantities entering them may differ in various ways. The most obvious example is the case of Palatini $f(R)$ theory in Einstein frame, where there is no dynamics of the scalar field due to vanishing of the invariant $\mathcal{I}^{(P)}_3$. A field without dynamics cannot cause inflationary behaviour and the following conclusion is that Palatini $f(R)$ in Einstein frame cannot account for exponential rate of expansion in the early Universe, at least not in the traditional scenario, where the driving force is a scalar field. However, if we go to the Jordan frame, obtaining inflation becomes possible under some special assumptions. Due to the nonminimal coupling between curvature and scalar field, Friedmann equations are rather complicated.
								
								What advantage do we have of working with invariants on cosmology? Having them, we can express certain quantities in a frame-independent way, analogous to the covariant equations and invariant quantities used in GR. Also, it is reasonable to assume that all observables should be identical functions of invariants in each frame. For example, let us investigate connection between the Einstein- and Jordan-like frames. We can write:
								$$\tilde{a}(\tilde{t})=e^\alpha a(t)=\frac{e^\alpha}{\sqrt{A}}\hat{a}(\hat{t})\equiv \frac{1}{\sqrt{\mathcal{I}_1}}\hat{a}(\hat{t})$$
								and
								$$d\tilde{t}=e^\alpha dt=\frac{e^\alpha}{\sqrt{A}} d\hat{t}\equiv\frac{1}{\sqrt{\mathcal{I}_1}}d\hat{t}$$
								so that:
								\begin{equation}
								R_{\text{ph}}=\int\frac{d\tilde{t}}{\tilde{a}(\tilde{t})}=\int\frac{d\hat{t}}{\hat{a}(\hat{t})}
								\end{equation}
								describes the particle horizon at a given moment of the cosmic time and, clearly, is a frame-independent invariant quantity. The same does not hold for the Hubble parameter:
								\begin{equation}
								\tilde{H}=\sqrt{\mathcal{I}_1}+\hat{H}-\frac{1}{2}\frac{d\text{ln}\mathcal{I}_1}{d\hat{t}}
								\end{equation}
								so that we must decide which frame to choose in order to describe expansion of the universe using the Hubble parameter. Canonical choice is the Jordan frame equivalent to the Jordan-like invariant frame when a proper choice of transformation has been made (namely, $\mathcal{C}=\alpha=0$ and $\mathcal{A}=\Phi=\mathcal{I}_1$).

			\chapter{Summary and outlook}
				Modern cosmology is a field of physics which contains probably the biggest number of mysteries and unsolved problems, which should not come as a surprise, as the theory deals with something that is unique by its very nature - the Universe. Cosmos as such can be hardly investigated in our terrestrial laboratories, so that we are forced to deduce its properties from the outcomes of our experiments and knowledge of laws of physics holding in the nearest vicinity. And still, formulated by Einstein one hundred years ago theory of gravity, general relativity, manages to account for most of the gravitational phenomena. The theory has been triumphant to an extend beyond reasonable expectations, working on a scale of the Solar system and, most likely, on the scale of entire Universe. However, the theory has its drawbacks leading to many problems, stimulating us to look for possible explanations or for a more fundamental theory. To name a few which are far from being sorted out: problem of nature of inflation and explaining provenance of the inflaton field, tiny value of cosmological constant, dark matter and dark energy. The last two are particularly worrying, as we know very little about roughly 90\% of the whole Universe. Apart from phenomenological aspects of GR, it has some problems on a theoretical level too. We still struggle to reconcile GR with quantum theory, treating the former not as a fundamental theory of gravity, but rather as an effective theory. These facts might be viewed as suggestions that the Einstein theory needs to be modified in order to encompass still unexplained phenomena and possibly establish a connection with more fundamental theories, such as string theory. Scalar-tensor theories are one of such attempts, decoupling metric structure of spacetime from its affine structure are another one, and the present thesis originated as a merge of these two ideas. Let us now sum up everything that has been discussed in the paper. 
				
				In the chapter on preliminaries important concepts used throughout the thesis were introduced. We started off with a notion of a conformal transformation, which was a tool used later on. Then, a couple of modified theories of gravity were introduced. In fact, there is a plethora of possible extensions and modifications of GR, and from the zoo of alternatives to Einstein gravity, $f(R)$ and scalar-tensor were chosen as the most straightforward, under certain circumstances being equivalent to one another. The prototype of all scalar-tensor theories was the Brans-Dicke theory introducing a scalar field nonminimally coupled to curvature. At that moment we made use of conformal transformation for the first time, exploiting the fact that it establishes a mathematical equivalence of solutions in different frames, thus enabling us to choose a frame where it is easier to solve equations of motion. Also, the very notion of a conformal frame was encountered there: we distinguished Einstein and Jordan frame. The former is characterized by an anomalous coupling between scalar field and the matter part of action, violating the weak equivalence principle. The latter features the scalar field coupled to the curvature scalar. The issue of which of these two frames should be deemed physical was discussed in the next chapter. We concluded the introductory chapter with introducing the Palatini formalism. We no longer consider connection to be a function of metric tensor. Instead, we think of them as two independent objects. We saw that in case of $f(R)$ theories Palatini approach led to a conclusion that the independent connection was an auxiliary field and could be eliminated from the field equations, yielding an alternated energy-momentum tensor. However, the Palatini approach features certain shortcomings that were presented at the end of the chapter. Still, this approach was dominant in the present paper, as there are not many authors writing about scalar-tensor theories in Palatini approach, and the field remains still unknown. 
				
				In the third chapter formulae describing conformal transformations in the Palatini formalism were introduced. The novelty was that the connection transformed according to a new function of the scalar field $\gamma_2$, independent of the function $\gamma_1$ changing the metric. The main difference in formulae describing geometric quantities in related frames stemmed from the fact that he covariant derivative of the metric tensor does not vanish anymore, since the connection is not Levi-Civita of the metric. Then, an action functional preserving its form under a conformal transformation of both the metric and the connection was presented, and transformation relations between coefficients of two related frames were sought. The action functional was found by looking for a form of the Lagrangian that does not get any new terms when changing the frame, and the new coefficients become functions of the old ones. The novelty here was the function $\mathcal{C}(\Phi)$ multiplying the term linear in velocities. It was shown that vanishing of the coefficient means that it is possible to find a frame in which the connection is Levi-Civita with respect to some conformally transformed metric tensor. Also, equations of motion were derived, but in the most general case, when none of the coefficients have been fixed by a choice of the functions $\{\gamma_1,\gamma_2,f\}$. In general, fixing all three functions give us a particular frame and translates to fixing three out of five arbitrary coefficients $\{\mathcal{A},\mathcal{B},\mathcal{C},\mathcal{V},\alpha\}$. Next, a group-like structure of the coefficients was presented; composing two subsequent transformations is equivalent to one with appropriately defined functions $\gamma_1$ and $\gamma_2$. Structure of the transformation relations allowed us to write out several quantities that remained invariant under a conformal change in a sense that their numerical value at a given spacetime point was the same in all frames, as well as their functional form was preserved - they depend on the same functions in each frame. However, due to the presence of the coefficient $\mathcal{C}$ some of the invariants do not have the same meaning as in case of the metric approach. This will translate into a changed dynamics of the scalar field in case of the $f(R)$ theories later on. We concluded the chapter by writing the action functional in terms of invariants. The Einstein-like invariant frame resembles the canonical Einstein frame, which can be obtain by putting $\mathcal{A}=1$, $\mathcal{B}=1$ and $\mathcal{C}=0$; in this instance the invariant metric and the metric used would coincide, $\hat{g}_{\mu\nu}=\mathcal{A}g_{\mu\nu}=g_{\mu\nu}$. Analogously, for the invariant Jordan-like frame a further specification of three of the functions: $\mathcal{A}=\Phi$, $\mathcal{C}=0$ and $\alpha=0$, leads to identification of the metric with the invariant one, $\tilde{g}_{\mu\nu}=e^{2\alpha}g_{\mu\nu}=g_{\mu\nu}$. 
				
				In the following chapter we focused on applications of formalism of invariants in Palatini scalar-tensor theories. First, we investigated a family of alternative theories of gravity, $f(R)$ theories, and showed their equivalence with scalar-tensor theories - or, to be more precise, with the Brans-Dicke theory with self-interacting potential added. We reproduced the well-known result that in the Jordan frame the metric version corresponds to the BD theory with $\omega=0$, whereas in the Palatini approach it corresponds to $\omega=-\frac{3}{2}$. This shows that both theories are fundamentally different, and their phenomenological predictions will diverge. Palatini formalism leads to a further diversification of potential theories of gravity. Conformal transformation relates frames that are mathematically equivalent, but not physically. Applying the Palatini approach results in another difference in the dynamics. This was clear when we considered $f(R)$ theories in four different versions. The second part of the chapter was dedicated to investigation of Friedmann equations in a simple case of an empty Universe of zero spatial curvature. Our objective was to check whether the model could reproduce a viable inflationary behaviour. In order to achieve this goal, we introduced invariant Hubble parameters for both invariant frames, and wrote down the Friedmann equations. In case of the Einstein-like frame, the inflation could be reproduced easily. It turned out that the Palatini $f(R)$ theory in Einstein frame cannot account for a correct dynamics of the scalar field, as it does not contain a kinetic term. In case of the Jordan-like frame, equations of motion were complicated as a result of nonminimally coupled scalar field. However, unlike the Einstein frame, the Jordan frame is a correct set-up for the Palatini $f(R)$ theories to investigate inflation. We did not dwell on the issue as it would go beyond the scope of the present thesis. 
				
				The thesis was merely a step toward formulating a comprehensive theory of scalar-tensor gravity. It must be challenged by the data obtained from various experiments, but one of the reasons alternative theories of gravity are being invented is not to replace the 'good old' general relativity, but rather to investigate its limitations and provide us with a fresh point of view on the Einstein theory. The outlook for the future work of the theory presented so far is to take up the topic at which the thesis ended and study the inflation in the Jordan-like frame, where the nonminimall coupling is present. A confrontation with the experiment must take place by computing the post-Newtonian parameters for the theory. Also, it will be interesting to calculate rotation curves of galaxies, dynamics of an expanding universe or investigating the process of structure formation.
	\begin{appendices}
				\chapter{Geodesic mapping}
				Let us assume that we have two symmetric linear connections on a manifold $\mathcal{M}$: $\Gamma^\alpha_{\mu\nu}$ and $\bar{\Gamma}^\alpha_{\mu\nu}$. The question is now: under what circumstances do these two connections have exactly the same geodesics? Equality of geodesics in this context means that for every geodesic curve $\gamma$ (defined according to $\Gamma$) parametrized by an affine parameter $s$ there exists a geodesic curve $\bar{\gamma}$ for the connection $\bar{\Gamma}$ parametrized by $\bar{s}$. The image of both curves is the same trajectory. The condition for it to occur is the following:
				\begin{theorem}
				(Herman, Weyl, 1921). Two symmetric connections defined on a manifold $\mathcal{M}$, $\Gamma^\alpha_{\mu\nu}$ and $\bar{\Gamma}^\alpha_{\mu\nu}$, share the same set of geodesics if and only if there exists a 1-form field $\omega_\gamma$ on $\mathcal{M}$ such that:
				$$ \bar{\Gamma}^\alpha_{\mu\nu}-\Gamma^\alpha_{\mu\nu}=\omega_\mu\delta^\alpha_\nu+\omega_\nu\delta^\alpha_\mu$$
				\end{theorem}
				\begin{proof}
				We will now attempt to show both necessity and sufficiency \footnote{The proof is given according to \cite{sok}}.
				\begin{itemize}
					\item \textit{Necessity}. We take an arbitrary geodesic line of $\Gamma^\alpha_{\mu\nu}$ with a tangent vector $u^\alpha=\frac{dx^\alpha}{ds}$. satisfying the equation $\frac{D}{ds}\frac{dx^\alpha}{ds}=\frac{d^2x^\alpha}{ds^2}+\Gamma^\alpha_{\mu\nu}\frac{dx^\mu}{ds}\frac{dx^\nu}{ds}=h(s)\frac{dx^\alpha}{ds}$. Because it is assumed to be a geodesic line also for $\bar{\Gamma}^\alpha_{\mu\nu}$ with a tangent vector:
					$$v^\alpha=\frac{dx^\alpha}{d\bar{s}}=\frac{dx^\alpha}{ds}\frac{ds}{d\bar{s}}, \quad \frac{ds}{d\bar{s}}:=q(s)>0$$
					where $q(s)$ is a function defined on the geodesic line. Vector $v^\alpha$ is parallel transported:
					$$\frac{dv^\alpha}{d\bar{s}}+\bar{\Gamma}^\alpha_{\mu\nu}\frac{dv^\mu}{d\bar{s}}\frac{dv^\nu}{d\bar{s}}=0$$
					We can now substitute $v^\alpha=qu^\alpha$ and divide it by $q^2$. Since
					$$\frac{1}{q^2}\frac{dq}{d\bar{s}}=\frac{1}{q}\frac{dq}{ds} \text{and} \frac{q}{q^2}\frac{du^\alpha}{d\bar{s}}=\frac{du^\alpha}{ds}$$
					we get:
					$$u^\alpha\frac{d}{ds}\text{ln}q+\frac{du^\alpha}{ds}+\bar{\Gamma}^\alpha_{\mu\nu}u^\mu u^\nu=0$$
					The derivative $\frac{du^\alpha}{ds}$ is simply $\frac{D}{ds}u^\alpha-\Gamma^\alpha_{\mu\nu}u^\mu u^\nu=-\Gamma^\alpha_{\mu\nu}u^\mu u^\nu$, so that:
					$$u^\alpha\frac{d}{ds}\text{ln}q=-(\bar{\Gamma}^\alpha_{\mu\nu}-\Gamma^\alpha_{\mu\nu})u^\mu u^\nu:=T^\alpha_{\mu\nu}u^\mu u^\nu$$
					Because both connections are symmetric, $T^\alpha_{\mu\nu}=T^\alpha_{\nu\mu}$. We now multiply the above equation by $u^\gamma$ and make it antisymmetric in the free indices $\alpha\gamma$:
					$$u^{[\alpha}u^{\gamma]}\frac{d}{ds}\text{ln}q=0=u^{[\gamma}T^{\alpha]}_{\mu\nu}u^\mu u^\nu=\delta^{[\gamma}_\lambda T^{\alpha]}_{\mu\nu}u^\mu u^\nu u^\lambda$$
					Tensor $u^\mu u^\nu u^\lambda$ is symmetric, so that only combinations $\delta^{[\gamma}_\lambda T^{\alpha]}_{\mu\nu}$ totally symmetric in their lower indices will give a non-zero contribution: $\delta^{[\gamma}_{(\lambda} T^{\alpha]}_{\mu\nu)}$. This equality must hold at every point of every geodesics, which means that it must be also satisfied for the vector $u^\alpha$ pointing in an arbitrary direction. It is possible only if all coefficients in this third-order polynomial are equal to zero. Writing them explicitly:
					$$0=\delta^{[\gamma}_{(\lambda} T^{\alpha]}_{\mu\nu)}\delta^{[\gamma}_\lambda T^{\alpha]}_{\mu\nu}=\frac{1}{6}(\delta^{\gamma}_{\lambda} T^{\alpha}_{\mu\nu}-\delta^{\alpha}_{\lambda} T^{\gamma}_{\mu\nu}+\delta^{\gamma}_{\mu} T^{\alpha}_{\nu\lambda}-\delta^{\alpha}_{\mu} T^{\gamma}_{\nu\lambda}+\delta^{\gamma}_{\nu} T^{\alpha}_{\lambda\mu}-\delta^{\alpha}_{\nu} T^{\gamma}_{\lambda\mu})$$
					We can contract two indices: $\gamma$ and $\lambda$, which gives us $(n+1)T^\alpha_{\mu\nu}-T^{\lambda}_{\nu\lambda}\delta^\alpha_\mu-T^\lambda_{\lambda\mu}\delta^\alpha_\nu=0$. We define 
					$$\omega_\gamma:=\frac{1}{n+1}T^\lambda_{\lambda\gamma}$$
					and we get:
					$$T^\alpha_{\mu\nu}=\omega_\mu\delta^\alpha_\nu+\omega_\nu\delta^\alpha_\mu$$
					which proves the necessity.
					\item \textit{Sufficiency}. Let $\bar{\Gamma}^\alpha_{\mu\nu}=\Gamma^\alpha_{\mu\nu}-2\omega_{(\mu}\delta^\alpha_{\nu)}$ for any 1-form $\omega$. We take now an arbitrary geodesic line of $\Gamma$ with a tangent vector $u^\alpha$; we shall show that this geodesic line is also geodesic of $\bar{\Gamma}$. We have to prove that there exists a function $q(s)$ such that vector $v^\alpha=qu^\alpha$ is parallel transported along the geodesic of $\bar{\Gamma}$. From the definition we have: $v^\alpha=\frac{dx^\alpha}{d\bar{s}}=\frac{dx^\alpha}{ds}\frac{ds}{d\bar{s}}=q(s)\frac{dx^\alpha}{ds}$. We define a vector:
					$$D^\alpha:=\frac{d}{d\bar{s}}v^\alpha+\bar{\Gamma}^{\alpha}_{\mu\nu}v^\mu v^\nu$$
					It is equal to:
					$$D^\alpha=\frac{dq}{d\bar{s}}u^\alpha+q\frac{du^\alpha}{d\bar{s}}+q^2\Gamma^\alpha_{\mu\nu}u^\mu u^\nu-q^2\omega_\mu\delta^\alpha_\nu u^\mu u^\nu=$$
					$$=q^2\Big(\frac{1}{q}\frac{dq}{ds}u^\alpha+\frac{du^\alpha}{ds}+\Gamma^\alpha_{\mu\nu}u^\mu u^\nu -\omega_\mu u^\mu u^\nu\Big)$$
					Because $\frac{du^\alpha}{ds}+\Gamma^\alpha_{\mu\nu}u^\mu u^\nu=0$, we end up with $D^\alpha=q^2\Big(\frac{1}{q}\frac{dq}{ds}u^\alpha-\omega_\mu u^\mu u^\nu\Big)$. Along every geodesics we have $u^\alpha=u^\alpha(s)$ and $\omega_\gamma=\omega_\gamma(x(s))$, so that the scalar $u^\mu\omega_\mu$ is uniquely determined function of $s$ for a given 1-form $\omega$. We postulate now $q>0$ and $\frac{d}{ds}\text{ln}q=\omega_\mu u^\mu$. Under this assumption $D^\alpha=0$, and $v^\alpha$ is parallel transported along $\bar{\gamma}$, which completes the proof.
				\end{itemize}
				\end{proof}
			\end{appendices}

\end{document}